\documentclass[reqno,11pt]{amsart}

\IfFileExists{mymtpro2.sty}{%
  \usepackage[subscriptcorrection]{mymtpro2}
}{}

\usepackage{a4,amssymb}

\newtheorem{theorem}{Theorem}[section]
\newtheorem{lemma}{Lemma}[section]
\newtheorem{corollary}{Corollary}[section]
\newtheorem{prop}{Proposition}[section]
\theoremstyle{definition}
\newtheorem{remark}[theorem]{Remark}

\newenvironment{acknowledgement}{\par\medskip\noindent\emph{Acknowledgement.}}

\newcommand{\labelnummer}{\mbox{\normalfont (\roman{numcount})}}%

\makeatletter

  {\let\curlabelspeicher\@currentlabel%
    \begin{list}{\labelnummer}%
      {\usecounter{numcount}\leftmargin0pt%
        \topsep0.5ex\partopsep2ex\parsep0pt\itemsep0ex\@plus1\p@%
        \labelwidth2.5em\itemindent3.5em\labelsep1em%
      }%
    \let\saveitem\item%
    \def\item{\saveitem%
      \def\@currentlabel{{\upshape\curlabelspeicher}$\,$\labelnummer}}%
    \let\savelabel\label%
    \def\label##1{\savelabel{##1}%
      \@bsphack%
        \ifmmode\else%
          \protected@write\@auxout{}%
          {\string\newlabel{##1item}{{\labelnummer}{\thepage}}}%
        \fi%
      \@esphack%
    }%
  }{\end{list}}%

\renewcommand{\appendix}{\def\thesection{\textsc{Appendix}}}

 \let\leq\le
 \let\geq\ge

 \let\Im\undefined

\DeclareMathOperator{\Im}{Im}

\DeclareMathOperator{\tr}{tr\kern1pt}

\makeatletter

\newif\ifper\pertrue
\def\per{.}

\def\bti{\@ifnextchar[\bbti\bbbti}
\def\bbti[#1]#2{#2, #1.}
\def\bbbti#1{#1.}

\def\z{\@ifnextchar[\zz\zzz}
\def\zz[#1]#2#3#4#5{\perfalse\emph{#2} \textbf{#3}, #4 (#5) [#1]}
\def\zzz#1#2#3#4{\emph{#1} \textbf{#2}, #3 (#4)\ifper\per\fi\pertrue}

\def\pub{\@ifstar\pubstar\pubnostar}
\def\pubnostar{\@ifnextchar[\@@pubnostar\@pubnostar}
\def\@@pubnostar[#1]#2#3#4{#2, #3, #4, #1\ifper\per\fi\pertrue}
\def\@pubnostar#1#2#3{#1, #2, #3\ifper\per\fi\pertrue}
\def\pubstar[#1]#2#3#4{\perfalse #2, #3, #4 [#1]\pertrue}

\makeatother

\newcommand{\bel}{\begin{equation} \label}
\newcommand{\ee}{\end{equation}}

\def\beq{\begin{equation}}
\def\eeq{\end{equation}}
\newcommand{\bea}{\begin{eqnarray}}
\newcommand{\eea}{\end{eqnarray}}
\newcommand{\beas}{\begin{eqnarray*}}
\newcommand{\eeas}{\end{eqnarray*}}

{

\newcommand{\R}{\mathbb{R}}

\newcommand{\Z}{\mathbb{Z}}

\newcommand{\N}{\mathbb{N}}

\newcommand{\C}{\mathbb{C}}
\newcommand{\E}{\mathbb{E}}

\begin{document}

\title[Density of states]{Dependence of the density of states on the probability distribution for discrete random Schr\"odinger operators}

\author[P.\ D.\ Hislop]{Peter D.\ Hislop}
\address{Department of Mathematics,
    University of Kentucky,
    Lexington, Kentucky  40506-0027, USA}
\email{peter.hislop@uky.edu}

\author[C.\ A.\ Marx]{Christoph A.\ Marx}
\address{Department of Mathematics,
Oberlin College,
Oberlin, Ohio 44074, USA}
\email{cmarx@oberlin.de}

\thanks{Version of \today}

\begin{abstract}
We prove that the the density of states measure (DOSm) for random Schr\"odinger operators on $\mathbb{Z}^d$ is weak-$^*$ H\"older-continuous in the probability measure. The framework we develop is general enough to extend to a wide range of discrete, random operators, including the Anderson model on the Bethe lattice, as well as random Schr\"odinger operators on the strip. An immediate application of our main result provides quantitive continuity estimates for the disorder dependence of the DOSm and the integrated density of states (IDS) in the weak disorder regime. These results hold for a general compactly supported single-site probability measure, without any further assumptions. The few previously available results for the disorder dependence of the IDS valid for dimensions $d \geq 2$ assumed absolute continuity of the single-site measure and thus excluded the Bernoulli-Anderson model. As a further application of our main result, we establish quantitative continuity results for the Lyapunov exponent of random Schr\"odinger operators for $d=1$ in the probability measure with respect to the weak-$^*$ topology.
\end{abstract}

\maketitle \thispagestyle{empty}

\tableofcontents

\vspace{.2in}



\section{Introduction: Dependence of the density of states on the probability measure}\label{sec:intro1}
\setcounter{equation}{0}

In this article, we quantify the dependence of the density of states on the \emph{single-site probability measure} for discrete random Schr\"odinger operators. We give estimates of the modulus of continuity, with respect to the probability distribution in the weak-$^*$ topology, of the density of states measure, the integrated density of states, the density of states function, and, in one dimension, the Lyapunov exponent.  One of the consequences of our results is that the density of states measure for the Bernoulli-Anderson model may be approximated by the density of states measure for smooth single-site probability measures. The latter question arises, for instance, in the study of local eigenvalue statistics in the localization regime.

We consider the formal Hamiltonian on $\ell^2 (\Z^d)$ with a random potential constructed from finite-rank projections and independent, identically distributed ($iid$) random variables,
\beq\label{eq:schr-op1}
H_\omega =  \Delta + \sum_{j} \omega_j P_j ~\mbox{,}
\eeq
where $\Delta$ is the finite-difference Laplacian. Here, the elements of $\omega = (\omega_j)$ are distributed according to a common, compactly supported Borel probability measure $\nu$ and the projections $P_j$ form a complete family of orthogonal projections with common rank $N \in \mathbb{N}$. A precise definition of the model (\ref{eq:schr-op1}) is given in [H1] of section \ref{subsec:model1}. For the usual Anderson model, the projections are rank-one, i.e. $N=1$. Models for $N > 1$ arise for instance in the study of multi-dimensional random polymers \cite{deBievre_Germinet_2000, SchulzBaldes_Jitomirskaya_Stolz_CMP_2003, DamanikSimsStolz_JFunAnal_2004}.

The {\em{density of states measure}} (DOSm) $n_\nu^{(\infty)}$ associated with the random lattice Schr\"odinger operator (\ref{eq:schr-op1}) is given by the spectral average
\beq \label{eq:DOSm1}
n_\nu^{(\infty)}(f) := \frac{1}{N} \E_{\nu^{(\infty)}} \{ {\rm Tr} (P_0 f(H_\omega) P_0) \} ~\mbox{,}
\eeq
where $\nu^{(\infty)}$ denotes the infinite product measure induced by $\nu$ (defined in (\ref{eq:m1}) below). The cumulative distribution function associated with the measure $n_\nu^{(\infty)}$ is commonly called the {\em{integrated density of states}} (IDS), denoted by $N_\nu (E) := n_\nu^{(\infty)}((-\infty, E))$ for $E \in \R$. {\em{If}} the DOSm $n_\nu^{(\infty)}$ is absolutely continuous with respect to Lebesgue measure on $\R$, the corresponding density (Radon-Nikodym derivative) will be referred to as the {\em{density of states function}} (DOSf), denoted by $\rho_\nu (E)$.

We mention that our results are not limited to the Hamiltonians in (\ref{eq:schr-op1}). Indeed, the framework we develop applies more generally to discrete random operators with a certain ``{\em{finite-range structure}}'' (see section \ref{sec:extensions_necfeat} for further details). The latter in particular includes the Anderson model on the Bethe lattice and random Schr\"odinger operators on the strip. Since the precise moduli of continuity will depend on the model under consideration, we will, for the sake of concreteness, present our results for operators of the form (\ref{eq:schr-op1}) and defer the discussion of generalizations to section \ref{sec:extensions}.

Finally, we note that an extension of the continuity results of the present article to continuum Schr\"odinger operators on $L^2 (\R^d)$, as well as to discrete models with non-compactly supported single-site measures, is the subject of a follow-up paper \cite{MarxHislop_contDosm_followup}, see also Remark \ref{rem_mainthm} below.

\subsection{Summary of the results}\label{subsec:model1}

In order to state our main results precisely, we first list the precise hypotheses on the model:
\begin{description}
\item [{[H1]}] The discrete Hamiltonian on $\ell^2 (\Z^d)$ has the form
\beq\label{eq:schr-op2}
H_\omega =  \Delta + \sum_{j \in {\mathcal{J}}} \omega_j P_j,
\eeq
where the components of $\omega := \{ \omega_j \}_{ j \in {\mathcal{J}}} \in \Omega := [-C,C]^{\Z^d}$ are iid random variables, distributed according to a common Borel probability measure $\nu$ with support in $[-C,C]$. The index set ${\mathcal{J}}$ is a lattice $K \Z^d$, for $N = K^d  \in \N$.
The rank $N$ projection $P_0$ projects onto the $N = |\Lambda_0|$ sites in the cube $\Lambda_0 :=[0, K-1]^d \subset \Z^d$ centered at the origin. The orthogonal projections $\{ P_k ~|~ k \in K\Z^d \}$ are generated by translation of the single rank $N$ projection $P_0$: Defining the unitary $U_k$ on $\ell^2 (\Z^d)$ by $f(x-k) = (U_kf)(x)$, for $k \in K\Z^d$, one has $P_k = U_k P_0 U_k^{-1}$. In summary, it follows that the spectrum of $H_\omega$ satisfies
 \beq\label{eq:schr-op-sp1}
\sigma ( H_\omega ) \subseteq [-2d-C, 2d+C] =: [-r,r] ~\mbox{for all } \omega \in \Omega ~\mbox{.}
\eeq
\end{description}
Ergodicity of the operator (\ref{eq:schr-op2}) implies that the DOSm can be defined by (\ref{eq:DOSm1}) and satisfies $\mathrm{supp}(n_\nu^{(\infty)}) = \sigma ( H_\omega )$, almost surely in $\omega$. In particular, the DOSm is a compactly supported probability measure with $\mathrm{supp}(n_\nu^{(\infty)}) \subseteq [-r,r]$.

We equip the space of Borel probability measures $\mathcal{P}([-C,C])$ on $[-C,C]$ with the {\em{weak-$^*$ topology}}, i.e. a sequence $\{ \nu_\alpha \}_{\alpha \in \N}$ weak-$^*$converges to a measure $\nu$ in the space $\mathcal{P}([-C,C])$, denoted by $\nu_\alpha \stackrel{w^\star}{\rightarrow} \nu$, if, for all $f \in \mathcal{C}([-C,C])$, one has $\nu_\alpha (f) \rightarrow \nu(f)$, as $\alpha \rightarrow \infty$. Here, we write $\nu(f) := \int f(x) ~d\nu(x)$, for $\nu \in \mathcal{P}([-C,C])$.

In view of quantitative results, it will be useful to work with a metric on $\mathcal{P}([-C,C])$. It is well known (see e.g. Dudley \cite{Dudley_1966}, Theorem 12 therein) that since $[-C,C]$ is a separable metric space, the topology of weak-$^*$ convergence on $\mathcal{P}([-C,C])$ is metrizable by the metric derived from the Lipschitz dual, i.e.
\begin{equation} \label{eq:metric}
d_w(\mu, \nu) := \sup \left\{ \left\vert \mu(f) - \nu(f) \right\vert ~:~ f \in \mathrm{Lip}([-C,C]) ~\mbox{ with } \Vert f \Vert_{\mathrm{Lip}} \leq 1 \right\} ~\mbox{,}
\end{equation}
for $\mu, \nu \in \mathcal{P}([-C,C])$. Here, $\mathrm{Lip}([-C,C])$ is the Banach space of Lipschitz functions on $[-C,C]$ together with the norm
\begin{equation} \label{eq:deflipnorm}
\Vert f \Vert_{\mathrm{Lip}} := \Vert f \Vert_\infty +  \sup_{x \neq y \in [-C,C]} \frac{| f(x) - f(y)| }{ |x-y|} =: \Vert f \Vert_\infty + L_f ~\mbox{.}
\end{equation}
While there are other common metrics that metrize weak-$^*$convergence on $\mathcal{P}([-C,C])$, in particular the closely related Wasserstein metric or the Prokhorov metric (see e.g. \cite{Dudley_1976_book, billingsley}), the metric defined in (\ref{eq:metric}) will be most natural in view of our applications (see examples 1--3 in section \ref{sec:qual-contDOSm1}).

\newpage
We can now summarize our main result. The following theorem combines the Theorems \ref{thm:qual1}, \ref{thm:quant1}, and \ref{thm:IDS-quant1}.
\begin{theorem} \label{thm:main}
Consider the model described in [H1].
\begin{itemize}
\item[(i)] For single-site probability measures $\nu \in \mathcal{P}([-C,C])$ and $E \in \mathbb{R}$, both the maps
\begin{align}
\nu \mapsto n_\nu^{(\infty)} ~\mbox{, } \nu \mapsto N_\nu(E) ~\mbox{,}
\end{align}
are continuous in the weak-$^*$ topology.
\item[(ii)] The modulus of continuity of the maps in part (i) is quantified by the following: there exist constants $\gamma > 0$, $C_2 > 0$, and $0 < \rho < 1$, only depending on $d$ and $N$, such that for all single-site measures $\mu, \nu \in \mathcal{P}([-C,C])$ with $d_w(\mu, \nu) < \rho$ one has
\begin{equation}
d_w( n_\mu^{(\infty)} , n_\nu^{(\infty)} ) \leq \gamma d_w(\mu, \nu)^{\frac{1}{1 + 2d}} ~\mbox{,}
\end{equation}
and, for all $E \in \mathbb{R}$,
\begin{equation}
| N_\mu (E) - N_\nu(E) | \leq \frac{C_2}{\log \left(   \frac{1}{d_w(\mu,\nu)}    \right)} ~\mbox{.}
\end{equation}
\end{itemize}
\end{theorem}
\begin{remark} \label{rem_mainthm}
We note that the hypothesis of compactness of the support of the single-site probability measures is necessary for the approach in the current paper. Indeed, as outlined in section \ref{subsec:contents1}, the proof of Theorem \ref{thm:main} relies on two key steps, the first of which, the ``finite-range reduction,'' uses polynomial approximation and thereby compactness of the support of the probability measure. The results of the present paper are extended to the situation of non-compactly supported single-site probability measures in \cite{MarxHislop_contDosm_followup}.
   In this paper, we will overcome this technical obstacle by working with resolvents instead of polynomial approximations. In addition, this modified approach, presented in \cite{MarxHislop_contDosm_followup}, applies to continuum random Schr\"odinger operators. {\em{The latter will however come at a price:}} the singularity of resolvents near the real axis will have to be compensated by higher regularity of the functions $f$ in Theorem \ref{thm:quant1}. In this case the right hand side of (\ref{eq:quant-cont1}) will depend on higher order derivatives of $f$. This will limit the continuity result for the DOSm to functions $f \in \mathcal{C}_c^k(\mathbb{R})$ for some $k = k(d) > 1$ instead of merely Lipschitz $f$, as in the present article, and also yields a smaller H\"older exponent.
\end{remark}

Theorem \ref{thm:main} in particular applies to the situation where $\{ \nu_\alpha \}_{\alpha \in \N}$ is a sequence of single-site probability measures converging in the space $\mathcal{P}([-C,C])$ to a probability measure $\nu$. In this context, part (i) of Theorem \ref{thm:main} implies continuity in $\alpha$ as $\alpha \to \infty$ and part (ii) provides \emph{quantitative bounds} on the rate of convergence for the DOSm and the IDS as $\alpha \rightarrow \infty$.

To further illustrate the usefulness of Theorem \ref{thm:main}, we present the following applications:
\begin{enumerate}
\item In section \ref{subsec:dosm-low1}, we show that Theorem \ref{thm:main} immediately implies the quantitative continuity for the DOSm and the IDS (for fixed energy!) in the disorder parameter $\lambda$ as $\lambda \to 0^+$ ({\em{weak disorder limit}}) for the rescaled model,
\beq\label{eq:schr-op1_rescale}
H_\omega =  \Delta + \lambda \sum_{j} \omega_j P_j ~\mbox{.}
\eeq
It is particularly noteworthy that our results given in Theorem \ref{thm:weak-cont1} -- \ref{thm:weak-contIDS1} hold for a {\em{general compactly supported probability measure, without any further assumptions}}. The very few available results for the disorder dependence of the IDS (for fixed energy) valid for {\em{higher}} dimensions $d \geq 2$ had to assume absolute continuity of the measure \cite{hks1, schenker1}. Even for $d=1$, where several results had been known previously (see the references and discussion in section \ref{subsec:previous-intro1}), the authors assumed appropriate decay of the Fourier transform of the single-site measure, which in particular ruled out Bernoulli measures.

\item In section \ref{subsec:le1} we show that for $d=1$, the Lyapunov exponent for each fixed energy is weak-$^*$ continuous in the probability measure (Theorem \ref{thm:qualiLE}). Under additional hypotheses (see [H3] in section \ref{subsec:le1}), we prove a bound on the convergence rate of the Lyapunov exponent as $\nu_\alpha \stackrel{w^\star}{\rightarrow} \nu$ (Theorem \ref{thm:LE-quant1}). Replacing the Lyapunov exponent by the sum of all non-negative Lyapunov exponents, these results extend to Schr\"odinger operators on the strip (see section \ref{sec:extensions_finiteRange_strip_Strip}). We mention that Theorem \ref{thm:qualiLE} recovers, for the case of Schr\"odinger operators, recent results of Bocker and Viana \cite{BockerViana_ETDS_2017} and Avila, Eskin, and Viana \cite{AvilaEskinViana}, which establish the continuity of the Lyapunov exponents for general products of random matrices. Theorem \ref{thm:LE-quant1} yields an additional quantification of the associated modulus of continuity.

\item Under the assumption that the DOSm is absolutely continuous with respect to the Lebesgue measure so that, for all $\alpha \in \mathbb{N}$, both the DOSf $\rho_{\nu_\alpha}$ and $\rho_\nu$ exist, we show in Theorem \ref{thm:dos1} of section \ref{subsec:dos-low1} that Theorem \ref{thm:main} implies the estimate
$$
| \rho_{\nu_\alpha} (E) - \rho_\nu(E) | \leq {C_3}{d_w(\nu_\alpha, \nu)}^{{\frac{\epsilon}{2+\epsilon}}{\frac{1}{1+2d}}},
$$
for all $E \in \R$ and some $\epsilon > 0$ related to the decay of the Fourier transform of the DOSf (see [H5] in section \ref{subsec:dosm-low1}). This applies, in particular, to the weak disorder limit, see Theorem \ref{thm:contiDOSfweakdisorder}.
\end{enumerate}


\subsection{Previous work on the continuity properties with respect to energy and disorder}\label{subsec:previous-intro1}

The regularity properties of the IDS with respect to the \emph{energy} for Schr\"odinger operators on $\ell^2 (\Z^d)$ (with rank-one perturbations) have been studied {\em{extensively}} by many authors. {\em{Since this is not the subject of this paper,}} we only mention a few results here; for a history and a more detailed survey of the literature, we refer to excellent review articles, for instance \cite{KirschMetzger_BarryFestschrift_2007, Veselic_review_2004}.

Pastur \cite{pastur1} proved that the IDS is continuous in the energy for one-dimensional random Schr\"odinger operators.  Craig and Simon \cite{craig-simon1} proved that for ergodic stationary potentials $V_\omega$ on $\Z^d$ satisfying $\E \{ \log  ( 1 + |V_\omega (0)|) \}$ the IDS is log-H\"older continuous. Delyon and Souillard \cite{DelyonSouillard_1984} provided a general proof of the continuity of the IDS with respect to $E$ for $d \geq 1$ for ergodic stationary potentials, but they do not give a quantitative estimate on the modulus of continuity.

Bourgain and Klein \cite{bourgain-klein1} established log-H\"older continuity of the outer density of states measure for Schr\"odinger operators on $\Z^d$, $d \geq 1$, which exists even for non-ergodic models. Most notably, in the same paper the authors establish also the (fractional) log-H\"older continuity of the outer density of states measure for continuum Schr\"odinger operators in dimensions $d=1,2,3$, which, prior to their work had not been known. There are better estimates on multidimensional random Schr\"odinger operators on $\ell^2 (\Z^d)$ if the single-site probability measure is H\"older continuous, in which case it can be shown that the IDS is also H\"older continuous. We refer to the papers of Combes, Hislop and Klopp \cite{chk1} and Rojas-Molina and Veseli\'c \cite{rojas-molina-veselic1} for refined results for random Schr\"odinger operators.

For random Schr\"odinger operators in dimension $d=1$, several higher regularity results for the energy dependence of the IDS are available. Simon and Taylor \cite{simontaylor} proved the IDS is $\mathcal{C}^\infty$ if the single-site measure is absolutely continuous with an appropriate decay of its Fourier transform. Further higher regularity results for the IDS specific to $d=1$, which allow for more general single-site distributions, were later obtained by Campanino and Klein \cite{campanino-klein86}, Klein and Speis \cite{KleinSpeis_JFunctAnal_1990}, and Bourgain \cite{Bourgain_Jdanalyse_2012}. Finally, we mention recent results by Hart and Vir\'ag \cite{HartVirag_CMP_2017} showing that the IDS of (\ref{eq:schr-op1_rescale}) in the weak disorder regime and $d=1$ is H\"older continuous with exponent $1 - c\lambda$ away from the band edges of the free Laplacian; the latter holds for any probability measure with finite support and thus in particular includes the Bernoulli-Anderson model. Higher-dimensional results ($d \geq 1$) for the regularity of the IDS with respect to $E$ for random Schr\"odinger operators were obtained by Bovier, Campanino, Klein, and Perez \cite{bovier-campanino-klein-perez1} in the high disorder regime.


Much fewer is known about the continuity of the IDS of (\ref{eq:schr-op1_rescale}) with respect to the \emph{disorder} (for {\em{fixed}} energy), in particular in the weak disorder regime $\lambda \to 0^+$. In this context, most available results are for $d=1$, all of which require appropriate decay conditions of the Fourier transform of the single-site measure. Here, we mention the results due to Bovier and Klein \cite{BovierKlein88}, Campanino and Klein \cite{campanino-klein90}, an Speis \cite{Speis_CMP_1992, Speis_JStatPhys_1991}, which are all specific to $d=1$ and rely on the supersymmetric replica method.

For higher dimensions, $d \geq 2$, the only available continuity results for the $\lambda$-dependence of the IDS (for fixed energy $E$) for random Schr\"odinger operators on $\mathbb{Z}^d$ in the weak disorder regime were established in \cite{schenker1} and \cite{hks1}. There, the authors prove H\"older continuity in $\lambda$ (for further details, see remark \ref{remark:weakcoupling} in section \ref{subsec:dosm-low1} of the present paper). Both of these results were obtained under the assumption that the single-site measure is absolutely continuous with a bounded density.


The DOSf for the Bethe lattice $\mathbb{B}$ was studied by Acosta and Klein \cite{AcostaKlein} who proved that if the single-site probability distribution is close to a Cauchy distribution, then the DOSf admits an analytic continuation in energy to strip around the real axis. On a Bethe strip $\mathbb{B} \times \{ 1, \ldots , m \}$, Klein and Sadel \cite{KleinSadel} studied the weak disorder regime of the DOSf. They showed that the DOSf $\rho_\lambda (E)$ is jointly continuous in $(\lambda, E)$ as $\lambda \rightarrow 0$ and for $E$ in the interval of absolutely continuous spectrum \cite[Theorem 1.4]{KleinSadel}. They do not give quantitative estimates on the moduli of continuity. Similar to the above mentioned results for lattice Schr\"odinger operators and $d=1$, the authors had to impose a decay condition on the Fourier transform of the single-site measure.


\subsection{Outline of the paper}\label{subsec:contents1}
Section \ref{sec:qual-contDOSm1} sets the stage by addressing the {\em{qualitative}} statements of continuity given in part (i) of Theorem \ref{thm:main}. The key ingredient is provided by Lemma \ref{lem:observation1}, which achieves a ``{\em{finite-range reduction}}'' by replacing the problem of comparing expectations with respect to product measures differing in {\em{infinitely}} many factors, by ones where only {\em{finitely}} many factors are different. Our proof explicitly quantifies the error (see Remark \ref{rem:quantifyerror-keylemma}), which in turn outlines a strategy for establishing the quantitative statements of continuity given in part (ii) of Theorem \ref{thm:main}. As shown in Lemma \ref{lem:rv-counting1}, the latter is made possible by the ``finite-range structure'' of the Hamiltonian in [H1], which causes the map $\omega \mapsto {\rm Tr} \{ P_0 f(H_\omega) P_0 \}$ to depend on the random potentials at only {\em{finitely many}} sites, for each given {\em{polynomial}} $f$. Section \ref{sec:quan-contDOSm1} consequently uses the strategy outlined in section \ref{sec:qual-contDOSm1} to establish the modulus of continuity given in part (ii) of Theorem \ref{thm:main}; for pedagogical reasons we split the discussion into two statements, Theorem \ref{thm:quant1} for the DOSm and Theorem \ref{thm:IDS-quant1} for the IDS. The arguments presented in section \ref{sec:quan-contDOSm1} in particular reduce the proof of these quantitative continuity results to verifying the Lipschitz continuity of certain ``single-site'' spectral averages (Lemma \ref{lem:finite-rank1}). A short proof of the latter is given in section \ref{sec:finite-rank-lemma1}. Section \ref{sec:appl-contDOSm1} presents the three applications of our main result summarized above in items (1) -- (3): the disorder dependence of the DOSm and the IDS in the weak disorder regime (section \ref{subsec:dosm-low1}), the dependence of the Lyapunov exponent on the probability measure (section \ref{subsec:le1}), and the dependence of the DOSf on the probability measure (section \ref{subsec:dos-low1}). Finally, section \ref{sec:extensions} shows that the framework developed in this paper is not limited to the model in [H1], but in fact applies to a wide range of discrete random operators: Section \ref{sec:extensions_necfeat} lists the necessary features for a model to be amenable to the framework of this article. The subsequent sections then consider a few examples of models which have been of particular interest in the literature, specifically the finite-range Anderson model and random Schr\"odinger operators on the strip (section \ref{sec:extensions_finiteRange_strip}), as well as the Anderson model on the Bethe lattice (section \ref{sec:extensions_bethe}).

The paper includes three appendices. In the first, Appendix A, section \ref{sec:appendix:rv-counting1}, we prove the counting lemma, Lemma \ref{lem:rv-counting1}, used in section \ref{sec:qual-contDOSm1}. Section \ref{sec:appendix:alt-fr1}, Appendix B, presents an alternate proof of the finite-rank lemma, Lemma \ref{lem:finite-rank1}, using the Helffer-Sj\"ostrand formula. Finally, section \ref{sec:appendix:nontangLE1} contains Appendix C in which we show that the assumption of H\"older continuity of the DOSm used in Proposition \ref{prop:le-quant1} cannot be relaxed to log-H\"older continuity using the methods employed to prove Proposition \ref{prop:le-quant1}.

\begin{acknowledgement}
We wish to thank Abel Klein and Svetlana Jitomirskaya for useful discussions about the literature, as well as for suggestions which influenced the presentation of the main results of this paper. We also wish to thank Ilya Kachkovskiy and Christian Sadel for useful comments. Chris Marx would like to thank the mathematics departments of both the University of Kentucky and UC Irvine, where this work was done, for their hospitality. Finally, Chris Marx also acknowledges financial support through Oberlin College's research leave program for junior faculty.
\end{acknowledgement}


\section{Qualitative continuity of the density of states in the probability measure} \label{sec:qual-contDOSm1}
\setcounter{equation}{0}
Assuming the set-up described in [H1], the main result of this section concerns the {\em{qualitative}} continuity of the DOSm and IDS in the single-site probability distribution as stated in part (i) of Theorem \ref{thm:main}. While we will show that the qualitative continuity result can be reduced to the basic fact that weak-$^*$ convergence of probability measures implies the same for the associated {\em{infinite}} product measures (see Corollary \ref{eq:weakstarconv} below), the line of arguments presented here will be refined in section \ref{sec:quan-contDOSm1} to obtain the corresponding {\em{quantitative}} estimates on the modulus of continuity given in part (ii) of Theorem \ref{thm:main}.

From now on we consider the following situation:
\vspace{.1in}
\noindent

\begin{description}
\item [{[H2]}] We let $\{ \nu_\alpha \}_{\alpha \in \N}$ be a sequence of of Borel probability measures in the space $\mathcal{P}([-C,C])$ converging in weak-$^*$ topology to a measure $\nu \in \mathcal{P}([-C,C])$. For ease of notation, we set $n^{(\infty)}: = n_\nu^{(\infty)}$ and $n_\alpha^{(\infty)}:= n_{\nu_\alpha}^{(\infty)}$, for $\alpha \in \mathbb{N}$. Moreover, we write
\begin{equation}
\eta_\alpha : = d_w( \nu_\alpha , \nu ) \to 0^+ ~\mbox{,}
\end{equation}
and assume without loss of generality that $0 < \eta_\alpha < 1$, for all $\alpha \in \mathbb{N}$.
\end{description}
\vspace{.1in}
\noindent

We emphasize that we do {\em{not}} make any further assumptions on the specific form of the measures $\{\nu_\alpha\}_{\alpha \in \mathbb{N}}$ and $\nu$. The following, however, lists some examples for the set-up described in [H2] which we will come back to in section \ref{sec:appl-contDOSm1} when discussing some applications of the general results obtained in this and the following section.
\vspace{.1in}
\noindent

\begin{description}
\item[Example 1 - Approximate delta functions]
Given $L \in \N$, weights $\{ w_n \in [0,1], 1\leq n \leq L\}$, $\sum_{n=1}^L w_n =1$, and points $\{\lambda_n \in [-C,C], 1\leq n \leq L\}$, let $\nu$ be the linear combination of delta functions
\begin{equation} \label{eq:lcdelta}
d \nu (x) = \sum_{n=1}^L w_n \delta ( x - \lambda_n ) ~\mbox{.}
\end{equation}
To define the sequence $\{\nu_\alpha\}_{\alpha \in \mathbb{N}}$ of measures approximating $\nu$, we replace the delta functions in (\ref{eq:lcdelta}) by ``approximate delta functions,'' i.e. for a given function $0 \leq \phi \in L^1([-C,C])$, $\Vert \phi \Vert_1 = 1$, and a strictly decreasing sequence $\{ \eta_\alpha\}_{\alpha \in \mathbb{R}}$ such that $0 < \eta_\alpha \searrow 0^+$ and $\eta_1 < 1$, we let
$$
d \nu_\alpha (x) = \sum_{n=1}^L \frac{w_n}{\eta_\alpha} ~ \phi \left( \frac{ x - \lambda_n }{ \eta_\alpha} \right) ~d x ~\mbox{.}
$$
Then, $\nu_\alpha \stackrel{w^\star}{\rightarrow} \nu$ and for every $f \in Lip([-C,C])$,
\begin{equation}
\vert \nu_\alpha(f) - \nu(f) \vert \leq C L_f \eta_\alpha \leq C \Vert f \Vert_{\mathrm{Lip}} \eta_\alpha ~\mbox{,}
\end{equation}
i.e. $d_w(\nu_\alpha, \nu) \leq \eta_\alpha$.

One special case of this example are smooth approximations of the Bernoulli probability measure for which $L=2$ and $w_1 = p$, $w_2 = 1-p$, for some $0 < p < 1$, and $\phi \in \mathcal{C}^\infty(\mathbb{R})$ with support in $[-C,C]$.

\vspace{.1in}
\noindent

\item[Example 2 - Absolutely continuous measures] Suppose both the limiting measure $\nu$ and the elements of the sequence $\{\nu_\alpha\}_{\alpha \in \mathbb{N}}$ are absolutely continuous (AC) measures in $\mathcal{P}([-C,C])$, i.e.
\begin{align}
\mathrm{d}\nu(x) = \phi(x) ~\mathrm{d} x ~\mbox{, }
\mathrm{d}\nu_\alpha (x) = \phi_\alpha(x) ~\mathrm{d} x ~\mbox{,}
\end{align}
for some density functions $0 \leq \phi, \phi_\alpha \in L^1([-C,C])$, $\Vert \phi \Vert_1 = 1$, $\Vert \phi_\alpha \Vert_1 = 1$, such that
\begin{equation} \label{eq_weakconvac}
1 > \eta_\alpha:= \Vert \phi_\alpha - \phi \Vert_1 \to 0 ~\mbox{, as $\alpha \to \infty$ .}
\end{equation}
We observe that given the absolute continuity of both $\nu_\alpha$ and $\nu$, the condition in (\ref{eq_weakconvac}) is in fact equivalent to $\nu_\alpha \to \nu$ in $w^\star$-topology. Moreover, for every $f \in Lip([-C,C])$ one has
\begin{equation}
\vert \nu_\alpha(f) - \nu(f) \vert \leq \Vert f \Vert_\infty \eta_\alpha \leq \Vert f \Vert_{\mathrm{Lip}} \eta_\alpha ~\mbox{,}
\end{equation}
i.e. $d_w(\nu_\alpha, \nu) \leq \eta_\alpha$.

\vspace{.1in}
\noindent

\item[Example 3 - Point measures] Let $\nu$ be as in (\ref{eq:lcdelta}). For each $\alpha \in \mathbb{N}$, let $\lambda^{(\alpha)} = (\lambda_n^{(\alpha)})_{1 \leq n \leq L} \in [-C,C]^L$ and $w^{(\alpha)} = (w_n^{(\alpha)})_{1 \leq n \leq L} \in [0,1]^L$, $\sum_{n=1}^L w_n^{(\alpha)} = 1$, be given and consider the measure
\begin{equation}
d \nu_\alpha(x) = \sum_{n=1}^L w_n^{(\alpha)} \delta ( x - \lambda_n^{(\alpha)} ) ~\mbox{.}
\end{equation}

It is straightforward to see that $\nu_\alpha \to \nu$ in $w^\star$-topology is equivalent to
\begin{equation}
\lambda_n^{(\alpha)} \to \lambda_n ~\mbox{ and } w_n^{(\alpha)} \to  w_n ~\mbox{, for each $1 \leq n \leq L$.}
\end{equation}
Moreover, letting
\begin{equation}
\eta_\alpha := \sum_{n=1}^L \{ \vert w_n^{(\alpha)} - w_n \vert + \vert \lambda_n^{(\alpha)} - \lambda_n \vert \} ~\mbox{,}
\end{equation}
one has for each $f \in Lip([-C,C])$ that
\begin{align}
\vert \nu_\alpha(f) - \nu(f) \vert & \leq \max_{1 \leq n \leq L} \{1 ; \vert \lambda_n \vert \} \cdot \Vert f \Vert_{\mathrm{Lip}} \eta_\alpha ~\mbox{, i.e. } \\
d_w(\nu_\alpha, \nu) & \leq \max_{1 \leq n \leq L} \{1 ; \vert \lambda_n \vert \}  \cdot \eta_\alpha  ~\mbox{.}
\end{align}

We will return to this example in section \ref{subsec:le1} where we apply our main result, Theorem \ref{thm:main}, to obtain continuity estimates for the Lyapunov exponent. The set-up introduced in the present example was considered in the recent works by Bocker and Viana \cite{BockerViana_ETDS_2017} and Avila, Eskin, and Viana \cite{AvilaEskinViana}, both of which establish the (qualitative) continuity of the Lyapunov exponents for general products of random matrices.

\end{description}
\vspace{.1in}
\noindent


\subsection{Approximating sequences of measures and convergence}\label{subsec:measures1}
As a technical preparation, we start by introducing certain approximations for the product measures $\nu_\alpha^{(\infty)}$ differing from the limiting measure $\nu^{(\infty)}$ on only a \emph{finite} number of lattice points. Here, we recall the definitions of the infinite product measures
\beq\label{eq:approx-m1}
\nu_\alpha^{(\infty)} := \bigotimes_{k \in \Z^d} \nu_\alpha, ~~~\alpha \in \N ,
\eeq
and
\beq\label{eq:m1}
\nu^{(\infty)} := \bigotimes_{k \in \Z^d} \nu  ~\mbox{,}
\eeq
associated with the single-site probability measures $\nu_\alpha$ and $\nu$, respectively.

The above mentioned approximating measures for $\nu_\alpha^{(\infty)}$ are then defined as follows: For $M \in \N$, we let
\beq\label{eq:Mapprox-m1}
\nu_\alpha^{(M)} := \left( \bigotimes_{k \in \Z^d; \|k\|_\infty \leq M} \nu_\alpha \right) \bigotimes
\left( \bigotimes_{k \in \Z^d; \|k\|_\infty > M} \nu \right), ~~~\alpha \in \N.
\eeq
We then claim:

\begin{lemma}\label{lem:observation1}
For any $F \in C(\Omega)$ and $\epsilon > 0$, there exists an integer $M \in \N$ (depending on both $F$ and $\epsilon$) such that
\beq\label{eq:finite-m-approx1}
| \nu_\alpha^{(M)} (F) - \nu_\alpha^{(\infty)}(F) | < \epsilon,  \forall \alpha \in \N.
\eeq
\end{lemma}

\begin{proof}
Let $\mathcal{F}$ denote the sub-algebra in $C(\Omega)$ consisting of functions $F \in C(\Omega)$ so that $F \in \mathcal{F}$ depends only on {\em{finitely-many}} $\omega_k$, $k \in \Z^d$. By the Stone-Weierstrass Theorem, the sub-algebra $\mathcal{F}$ is dense in $C(\Omega)$. In fact, one easily checks that $\mathcal{F}_0$ defined as
\begin{equation} \label{eq:subalgebraf0}
~\mbox{ the sub-algebra generated by }\left\{ F = \prod_{\|j\|_\infty \leq M} f_j (\omega_j) ~:~ f_j \in C([-C,C]), ~\|j\|_\infty \leq M, ~M \in \N \right\},
\end{equation}
is dense in $C(\Omega)$. Hence, for every $F \in C(\Omega)$ and $\epsilon > 0$, there exists $M \in \N$ and $F_M \in \mathcal{F}$, with $F_M$ depending only on the random variables $\{ \omega_k ~|~ \|k \|_\infty \leq M \}$, such that
\beq\label{eq:finite-fnc-approx1}
\|F - F_M \|_{\infty}  < \epsilon .
\eeq
By construction, for all $\alpha \in \N$, we have the equality
$\nu_\alpha^{(M)} (F_M) = \nu_\alpha^{(\infty)} (F_M)$, from which it follows that
\bea\label{eq:finite-m-approx2}
| \nu_\alpha^{(M)} (F) - \nu_\alpha^{(\infty)}(F) | & \leq & | \nu_\alpha^{(M)} (F) - \nu_\alpha^{(M)}(F_M) | \nonumber \\
 & & + | \nu_\alpha^{(M)} (F_M) - \nu_\alpha^{(\infty)}(F_M) | + | \nu_\alpha^{(\infty)} (F_M) - \nu_\alpha^{(\infty)}(F) |   \nonumber \\
 & & < 2 \epsilon,
 \eea
for all $\alpha \in \N$. This proves the lemma.
\end{proof}
\begin{remark} \label{rem:quantifyerror-keylemma}
The proof of Lemma \ref{lem:observation1} shows that given $F \in C(\Omega)$ and $\epsilon > 0$, the index $M \in \N$ can be determined {\em{explicitly}} by finding a function $\tilde{F} \in \mathcal{F}$ so that $\| F - \tilde{F} \|_\infty < \epsilon / 2$. This will play a key role in section \ref{sec:quan-contDOSm1} for obtaining our quantitive continuity theorem for the DOSm.
\end{remark}

Using approximation of $F \in C(\Omega)$ by elements of $\mathcal{F}_0$ defined in (\ref{eq:subalgebraf0}), Lemma \ref{lem:observation1} in particular implies:
\begin{corollary} \label{eq:weakstarconv}
If $\nu_\alpha \rightarrow \nu$ in the $w^\star$-topology, then $\nu_\alpha^{(\infty)} \rightarrow \nu^{(\infty)}$ in the $w^\star$-topology.
\end{corollary}

\subsection{Qualitative continuity}\label{subsec:qual1}

We can now use Corollary \ref{eq:weakstarconv} to establish the \emph{qualitative continuity} property of the DOSm and the IDS with respect to $w^\star$-convergence in the single-site probability measure.

We begin with the following simple result:
\begin{lemma}\label{lem:continuity1}
For every $f \in C([-r,r])$, the map
\beq\label{eq:cont1}
\omega := \{ \omega_j \}_{j \in \Z^d} \mapsto {\rm Tr} \{ P_0 f(H_\omega) P_0 \}
\eeq
is a continuous function on $\Omega$.
\end{lemma}

\begin{proof}
We first note that if $f(x) = x^n |_{[-r, r]}$, it follows from the linear dependence of $H_\omega$ on $\omega$ that ${\rm Tr} \{ P_0 f(H_\omega) P_0 \} \in \mathcal{F}$ (as defined at the beginning of the proof of Lemma \ref{lem:observation1}), whence the same holds true for every polynomial $f$. The lemma thus follows by a density argument and the continuous functional calculus:
\bea\label{eq:cont12}
|{\rm Tr} \{ P_0 f(H_\omega) P_0 \} -{\rm Tr} \{ P_0 g(H_\omega) P_0 \} | &=& |{\rm Tr} \{ P_0[ f(H_\omega) -g(H_\omega) ] P_0 \}| \nonumber \\
 & \leq & N \|f-g\|_\infty ~\mbox{.}
 \eea
\end{proof}

Application of Corollary \ref{eq:weakstarconv} to the continuous functions in (\ref{eq:cont1}) thus immediately yields the following theorem, which is equivalent to part (i) of Theorem \ref{thm:main}.
\begin{theorem}\label{thm:qual1} (Qualitative continuity of the DOSm.)
Assume the set-up described in [H1]-[H2]. If the sequence of probability measures $\nu_\alpha$ converges in the $w^\star$-topology to the probability measure $\nu$ as described in [H2], then the DOSm $n_\alpha^{(\infty)}$ $w^\star$-converges to $n^{(\infty)}$.
Moreover, the IDS converges pointwise:
\beq\label{eq:IDSconv1}
\lim_{\alpha \rightarrow \infty} N_\alpha (E) = N(E), ~~~ \forall  E \in \R  .
\eeq
\end{theorem}
We note that the ``moreover-statement'' of Theorem \ref{thm:qual1} relies on the continuity of the IDS in the energy \cite{DelyonSouillard_1984} and the standard fact that $w^\star$-convergence of probability measures implies point-wise convergence of the respective cumulative distributions at all points of continuity for the limiting measure. We mention that a quantitative version of the latter fact will be used in the proof of Theorem \ref{thm:IDS-quant1}, see (\ref{eq:ids-cont2}) below.

One of the main purposes of this article is to obtain a \emph{quantitative version} of the results of Theorem \ref{thm:qual1}, which will be the subject of section \ref{sec:quan-contDOSm1}.

\subsection{Dependence on the random variables}\label{subsec:dos-cont1}

We conclude this section with a refinement of how the function $\omega \mapsto {\rm Tr} \{ P_0 f(H_\omega) P_0 \}$ depends on the random variables for certain choices of $f$.

As mentioned in the proof of Lemma \ref{lem:continuity1}, if $f$ is a polynomial, the function $\omega \mapsto {\rm Tr} \{ P_0 f(H_\omega) P_0 \}$ depends only on {\em{finitely}}-many random variables. To obtain a \emph{quantitative version} of the continuity results in Theorem \ref{thm:qual1}, we will need an upper bound on the number of random variables upon which the map $\omega \mapsto {\rm Tr} \{ P_0 f(H_\omega) P_0 \}$ depends for each given polynomial $f$:
\begin{lemma}[``Counting lemma'']\label{lem:rv-counting1}
Under the hypothesis [H1], there is a strictly increasing, positive, and finite function $\Gamma : \N \rightarrow \R^+$, so that if $f(x)$ is a polynomial of degree $n$, then the map $\omega \mapsto {\rm Tr} \{ P_0 f(H_\omega) P_0 \}$ depends on at most $\Gamma (n)$ random variables $\{ \omega_k ~|~ \|k\|_\infty \leq M_n < \infty \}$, for some $M_n \in \mathbb{N}$ with $\Gamma (n) = (2 M_n + 1)^d$. Explicitly, $\Gamma (n)$ can be taken as
$$
\Gamma (n) = 2^d n^d.
$$
\end{lemma}
We prove this lemma in the appendix, section \ref{sec:appendix:rv-counting1}. There, we also give a more accurate estimate of $\Gamma(n)$ which explicitly shows the dependence on the rank $N$ of the potentials, see (\ref{eq:gammaaccurate}) in section \ref{sec:appendix:rv-counting1}.


\section{Quantitative continuity of the density of states in the probability measure}\label{sec:quan-contDOSm1}
\setcounter{equation}{0}

We are now able to prove our on the modulus of continuity of the DOSm and the IDS in the underlying single-site probability measure, stated in part (ii) of Theorem \ref{thm:main}. For pedagogical reasons we split the proof into two statements, Theorem \ref{thm:quant1} for the DOSm and Theorem \ref{thm:IDS-quant1} for the IDS.

\begin{theorem}\label{thm:quant1} (Quantitative continuity of the DOSm.)
Assuming the set-up described in [H1]-[H2], for every $f \in Lip ([-r,r])$ we have
\beq\label{eq:quant-cont1}
| n_\alpha^{(\infty)} (f) - n^{(\infty)} (f) | \leq \gamma \Vert f \Vert_{\mathrm{Lip}} ~\eta_\alpha^{\frac{1}{1+ 2 d}} ,
\eeq
for all $\alpha \geq \alpha_0$.
Here, the constants $\alpha_0, \gamma$ are independent of $f$, only depending on $d$ and $N$, and are determined explicitly by, respectively, (\ref{eq:weak-conv-dos3}) and (\ref{eq:weak-conv-dos-summ2}).
\end{theorem}

As we will show, the quantitative continuity estimate for the IDS can then obtained as a consequence of Theorem \ref{thm:quant1}. In \cite{craig-simon1}, Craig and Simon proved that the IDS for discrete ergodic Schr\"odinger operators on $\Z^d$ is log-H\"older continuous in the energy, i.e. for any ergodic, $L^\infty$-potential $V$, there exists a constant $C_I = C_I(d, \|V\|_\infty)$
such that for all energies $E \in \R$ and $0 < \epsilon \leq \frac{1}{2}$, one has
\beq\label{eq:ids-logholder1}
| N(E) - N(E + \epsilon)| = n^{(\infty)}([E, E+\epsilon]) \leq \frac{C_I}{\log \left( \frac{1}{\epsilon} \right) }.
\eeq
A more recent proof of this result, which applies more generally to a deterministic setting of Schr\"odinger operators and also includes operators on $\R^d$ for $d=1,2,3$ was obtained by Bourgain and Klein in \cite{bourgain-klein1}.

We will show that using Theorem \ref{thm:quant1} for a continuous approximation of $\chi_{(-\infty, E)}$, combined with \eqref{eq:ids-logholder1}, yields the following theorem for the IDS:
\begin{theorem}\label{thm:IDS-quant1} (Quantitative continuity of the IDS)
Under hypotheses [H1]-[H2], we have the bound
\beq\label{eq:IDS-cont1}
| N_\alpha (E) - N(E) | \leq \frac{ C_2}{ \log \left( \frac{1}{\eta_\alpha} \right)}
\eeq
for all $\alpha \geq \alpha_0$, and $E \in \R$. The constant $C_2 = C_2 ( d, N)$, where $N = {\rm rank} ~P_0$,
is defined in \eqref{eq:ids-cont4} and $\alpha_0$ is determined in (\ref{eq:weak-conv-dos3}).
\end{theorem}
\begin{remark} \label{remark:thm:IDS-quant1}
We mention that because of (\ref{eq:schr-op-sp1}), the left-hand side of (\ref{eq:IDS-cont1}) identically equals zero if $E \in \mathbb{R} \setminus (-r,r)$. Hence, Theorem \ref{thm:IDS-quant1} is non-trivial only for $E \in (-r,r)$.
\end{remark}

We begin with studying the effect of the finite-rank potentials. The following lemma is stated for a general bounded self-adjoint operator $H^{(0)}$ acting on $\ell^2(\mathbb{Z}^d)$. For {\em{fixed}} $\ell \in \mathcal{J}$, where $\mathcal{J}$ is defined as in [H1], we consider the family of finite-rank perturbations given by
\begin{equation} \label{eq:finiteranklemmasetup}
H_\lambda^{(\ell)}:= H^{(0)} + \lambda P_\ell ~\mbox{, } \lambda \in [-C,C] ~\mbox{.}
\end{equation}
Here, $P_\ell$ is the rank-$N$ projection defined in [H1]. Let $[a,b] \subset \mathbb{R}$ be an interval so that
\beq\label{eq:inclusion1}
\bigcup_{\lambda \in [- C, C]} ~\sigma (H_\lambda^{(\ell)}) \subseteq [a,b].
\eeq
Then, we claim:

\begin{lemma}\label{lem:finite-rank1} (Finite-rank Lemma.)
Given the set-up described in [H2], (\ref{eq:finiteranklemmasetup}), and (\ref{eq:inclusion1}), one has for all $f \in Lip([a,b])$ that
\bea\label{eq:rank-N2}
{\left| \int_\R ~{\rm Tr} ~ \left( P_0 f(H_\lambda^{(\ell)}) P_0 \right) ~ d \nu_\alpha (\lambda) -
\int_\R ~{\rm Tr} ~ \left( P_0 f(H_\lambda^{(\ell)}) P_0 \right) ~ d \nu (\lambda) \right| }  & \leq & 2 N^2 \Vert f \Vert_{\mathrm{Lip}} \eta_\alpha \nonumber \\
 &   =: & c_f \eta_\alpha ~\mbox{,}
\eea
where $\Vert f \Vert_{\mathrm{Lip}}$ is defined in (\ref{eq:deflipnorm}).
\end{lemma}

Deferring the proof of Lemma \ref{lem:finite-rank1} to section \ref{sec:finite-rank-lemma1}, we proceed with the proof of Theorem \ref{thm:quant1}.

\begin{proof}[Proof of Theorem \ref{thm:quant1}]
1. For $r$ defined in (\ref{eq:schr-op-sp1}), consider $f \in Lip([-r,r])$. Let $f_n(x)$ denote the $n^{th}$-Bernstein polynomial rescaled so that $f_n$ approximates $f$ uniformly on $[-r, r]$, i.e. if $\phi : [-r, r] \rightarrow [0,1]$ denotes the isomorphism given by $\phi (x) = (x+r) / (2r)$, we have
$$
f_n := B_n [f \circ \phi^{-1} ] \circ \phi ~\mbox{.}
$$
Here, for $g \in C([0,1])$, $B_n[g]$ is the standard $n^{th}$-Bernstein polynomial approximating $g$, given by
$$
B_n[g](x) = \sum_{k=0}^n  \binom{n}{k} g \left( \frac{k}{n} \right) x^k (1-x)^{n-k} .
$$
It is well-known (see e.g. \cite{beals1}) that for $g \in C([0,1])$ with modulus of continuity $W_g$ on $[0,1]$, the approximation by $B_n[g]$ satisfies
\beq\label{eq:poly-approx1}
\| B_n[g] - g \|_\infty \leq c_b W_g(n^{-1/2}) ,
\eeq
where $c_b$ is an absolute constant. We note that it is shown in \cite{sikkema} that the $n^{-1/2}$ dependence in (\ref{eq:poly-approx1}) is, in general, optimal and that the optimal value for the constant $c_b$ is given by
\begin{equation} \label{eq:bernsteinconst}
c_b = \dfrac{4306 + 837 \sqrt{6}}{5832} \approx 1.08989 ~\mbox{.}
\end{equation}

Returning to $f$ and $f_n$, and taking into account the rescaling to $[-r,r]$, we obtain from \eqref{eq:poly-approx1} that
\beq\label{eq:poly-approx2}
\| f_n - f \|_\infty \leq 2 r c_b L_f n^{-1/2} =: b_f n^{-1/2}.
\eeq

\noindent
2. It follows that for all $f \in Lip([-r, r])$, we have
\bea\label{eq:rank-N3}
\left| {\rm Tr} ~ \left( P_0 f(H_\omega) P_0 \right) - {\rm Tr} ~ \left( P_0 f_n(H_\omega) P_0 \right) \right| &
\leq & N \| f - f_n \|_\infty  \nonumber \\
 & \leq  & N \frac{b_f}{n^{1/2}} .
\eea
Applying Lemmas \ref{lem:observation1} and \ref{lem:rv-counting1} to $F (\omega) = {\rm Tr} ~ \left( P_0 f(H_\omega) P_0 \right)$, we conclude that for
\begin{equation}
M_n = (\Gamma (n)^{\frac{1}{d}} - 1)/2 ~\mbox{,}
\end{equation}
one has
\beq\label{eq:finite-rv1}
| \nu_\alpha^{(M_n)} (F) - \nu_\alpha^{(\infty)}(F) | < \frac{2 b_f}{n^{\frac{1}{2}}},
\forall \alpha \in \N.
\eeq
For $1 \leq j \leq \Gamma(n)$, fix a labelling $j \mapsto \ell_j$ of points in the cube
$$
\Lambda_n := \{ k \in \Z^d ~|~ \| k \|_\infty \leq M_n \}.
$$
Define a sequence of measures $\mu_{\alpha, j}$
by
\beq\label{eq:finite-m1}
\mu_{\alpha, j} := \left( \bigotimes_{\ell_k : 1 \leq k \leq j} \nu \right) \bigotimes
\left( \bigotimes_{\ell_k : j < k \leq \Gamma(n)} \nu_\alpha \right) \bigotimes
\left( \bigotimes_{k \in \Z^d \backslash \Lambda_n } \nu \right) ~\mbox{, for $1 \leq j \leq \Gamma(n) - 1$.}
\eeq
With respect to these measures, we then estimate
\bea\label{eq:finite-rv2}
| \nu_\alpha^{(M_n)} (F) - \nu^{(\infty)}(F) | &\leq &| \nu_\alpha^{(M_n)} (F) - \mu_{\alpha,j}(F) | \nonumber \\
& & +  \sum_{j=1}^{\Gamma(n) - 2} | \mu_{\alpha, j} (F) - \mu_{\alpha, j+1}(F) | + | \mu_{\alpha, \Gamma(n) - 1} (F) - \nu^{(\infty)}(F) |, \nonumber \\
 & &
\eea
for all $\alpha \in \N$.

\noindent
3. We observe that each of the terms on the right in \eqref{eq:finite-rv2} are of the type considered in Lemma \ref{lem:finite-rank1}. For example, for $1 \leq j \leq \Gamma(n) -2$, a sample term for the sum in \eqref{eq:finite-rv2} is
\bea\label{eq:finite-ex1}
| \mu_{\alpha, j } (F) - \mu_{\alpha, j+1}(F) |&=& \frac{1}{N} \left| \int \prod_{k \in \Z^d \backslash \Lambda_n} ~d \nu (\omega_k)  \int \prod_{k=1}^j ~d \nu (\omega_k) \int \prod_{k=j+2}^{\Gamma(n)} ~d \nu_\alpha (\omega_k)
 \right. \nonumber \\
  &  & \left. \times \left\{ \int {\rm Tr} ( P_0 f(H_\omega) P_0 ) d \nu_\alpha (\omega_{j+1}) -  \int {\rm Tr} ( P_0 f(H_\omega) P_0 ) d \nu (\omega_{j+1}) \right\} \right| . \nonumber \\
  & &
  \eea
To apply Lemma \ref{lem:finite-rank1} to the term in curly brackets in \eqref{eq:finite-ex1}, we only consider the dependence
of $H_\omega$ on the random variable $\omega_{j+1}$, while fixing all other random variables. Thus, writing $\ell = j+1$ and $\lambda = \omega_{j+1}$, the Schr\"odinger operator in \eqref{eq:schr-op2} may be recast in the form considered in Lemma \ref{lem:finite-rank1}
\beq\label{eq:rankN1}
H_\lambda^{(\ell)} := H_{\ell^\perp} + \lambda P_\ell, ~~ ~~~\lambda \in [-C,C] ~\mbox{,}
\eeq
with $H^{(0)} = H_{\ell^\perp}$. In this manner, the operator $H_\lambda^{(\ell)}$ is a rank-$N$ perturbation of $H_{\ell^\perp}$. Application of Lemma \ref{lem:finite-rank1} to this and similar terms thus yields
  \beq\label{eq:estimate1}
  | \nu_\alpha^{(M_n)} (F) - \nu^{(\infty)}(F) | \leq  \frac{\Gamma(n)}{N} c_f \eta_\alpha.
  \eeq
  Consequently, combining \eqref{eq:finite-m-approx2} and \eqref{eq:finite-ex1}, we obtain for all $\alpha \in \N$,
  \beq\label{eq:weak-conv-dos1}
   | n_\alpha^{(\infty)} (f) - n^{(\infty)} (f) | \leq \frac{2 b_f}{n^{\frac{1}{2}}}  + \frac{\Gamma(n)}{N} c_f \eta_\alpha .
\eeq

\noindent
4. We emphasize that so far $n \in \N$ was fixed and arbitrary. In particular, $n$ and $\alpha$ in \eqref{eq:weak-conv-dos1} are {\em{independent}} of each other. To extract a rate of convergence from
\eqref{eq:weak-conv-dos1}, we will choose $n$ adapted to the decay of $\eta_\alpha$. To this end, we define a function $h := (\Gamma^{-1})^{\frac{1}{2}}$, where $\Gamma^{-1}$ is the inverse function of $\Gamma$ considered in Lemma \ref{lem:rv-counting1}, and fix a parameter $0 < \xi < 1$ to be determined below.

Given $\alpha \in \N$, we take $n_\alpha \in \mathbb{N}$ so that
$$
h^2 \left( \frac{1}{\eta_\alpha^\xi} \right) - 1 \leq n_\alpha \leq  h^2 \left( \frac{1}{\eta_\alpha^\xi} \right).
$$
Using these choices in \eqref{eq:weak-conv-dos1}, we find:
 \bea\label{eq:weak-conv-dos2}
   | n_\alpha^{(\infty)} (f) - n^{(\infty)} (f) | & \leq & \frac{2 b_f}{n^{\frac{1}{2}}}  + \frac{\Gamma(n)}{N} c_f \eta_\alpha  \nonumber \\
   & \leq & \frac{2 b_f}{\left( \Gamma^{-1} \left( \frac{1}{\eta_\alpha^\xi} \right) - 1 \right)^{\frac{1}{2}}  }  + \frac{1}{N} c_f \eta_\alpha^{1 - \xi}  .
\eea

Given the explicit form of $\Gamma (n)$ in  Lemma \ref{lem:rv-counting1}, we can simplify \eqref{eq:weak-conv-dos2} further. According to Lemma \ref{lem:rv-counting1}, we have
\begin{equation} \label{eq:explicithfunction}
h(x) = \left( \frac{x}{2^d} \right)^{\frac{1}{2 d}}.
\end{equation}
We let
\begin{equation}
\xi_0 := ( 1 + (2 d)^{-1})^{-1}
\end{equation}
and take $\alpha_0 \in \N$ such that
\beq\label{eq:weak-conv-dos3}
\eta_{\alpha_0}^{\xi_0} \leq \max \left\{ 1, \frac{1}{\Gamma(2)} \right\} = \frac{1}{2^{2d}} ~\mbox{, for all $\alpha \geq \alpha_0$.}
\eeq
Let $\alpha \geq \alpha_0$. Then, by the choice of $\alpha_0$ in (\ref{eq:weak-conv-dos3}), for all $\xi$ with $\xi_0 \leq \xi < 1$, one has $(\Gamma (2))^{-1} \geq \eta_\alpha^\xi$. In particular, we have $\Gamma^{-1} (\eta_\alpha^{-\xi})\geq 2$, or
\beq\label{eq:weak-conv-dos4}
\Gamma^{-1} \left( \frac{1}{\eta_\alpha^\xi} \right) - 1 = h^2 \left( \frac{1}{\eta_\alpha^\xi} \right) -1
\geq \frac{1}{2} h^2  \left( \frac{1}{\eta_\alpha^\xi} \right) = \frac{1}{2} \Gamma^{-1} \left( \frac{1}{\eta_\alpha^\xi} \right).
\eeq
Thus, using \eqref{eq:weak-conv-dos2} and the  explicit form of $h$ in (\ref{eq:explicithfunction}), we obtain for all $\alpha \geq \alpha_0$:
 \bea\label{eq:weak-conv-dos-summ2}
   | n_\alpha^{(\infty)} (f) - n^{(\infty)} (f) |  & \leq &
   \max \left\{ 4 b_f , \frac{c_f}{N} \right\} \left[ \eta_\alpha^{\frac{\xi}{2 d}} + \eta_\alpha^{1 - \xi} \right] \nonumber \\
   &\leq& 4 \max \left\{ 4 c_b r , N  \right\} \Vert f \Vert_{\mathrm{Lip}} \left[ \eta_\alpha^{\frac{\xi}{2 d}} + \eta_\alpha^{1 - \xi} \right] \nonumber \\
   & =: & \gamma \Vert f \Vert_{\mathrm{Lip}} \left[ \eta_\alpha^{\frac{\xi}{2 d}} + \eta_\alpha^{1 - \xi} \right] ~\mbox{.}
   \eea
Finally, optimizing \eqref{eq:weak-conv-dos-summ2} with respect $\xi$, we obtain $\xi = \xi_0 =( 1 + (2 d)^{-1})^{-1}$, which is consistent with the lower bound on $\xi$ used to determine $\alpha_0$.
This concludes the proof.
\end{proof}

\begin{proof}[Proof of Theorem \ref{thm:IDS-quant1}]
Given remark \ref{remark:thm:IDS-quant1}, it suffices to consider $E \in (-r, r)$. Then, given $0 < a_\pm$ arbitrary,
we define two functions $f_\pm : [-r,r] \rightarrow [0, \infty)$ by
\beq\label{eq:cut-off1}
f_- (x) = \left\{ \begin{array}{ll}
                     1 & {\rm if} ~~ x \in [-r, E- a_-] \\
                      1 - \frac{1}{a_-} (x- (E-a_-)) &  {\rm if} ~~x \in (E-a_-, E] \\
                      0 & {\rm if} ~~ x > E
                      \end{array}
                      \right.
 \eeq
and
\beq\label{eq:cut-off2}
f_+ (x) = \left\{ \begin{array}{ll}
                     1 & {\rm if} ~~ x \in [-r, E] \\
                      1 - \frac{1}{a_+} (x- E) &  {\rm if} ~~x  \in (E, E+a_+ ] \\
                      0 & {\rm if} ~~ x > E + a_+ .
                      \end{array}
                      \right.
 \eeq
By construction, one has
$$
0 \leq f_- \leq \chi_{[-r, E)} \leq f_+ ~\mbox{, } \Vert f_\pm \Vert_{\mathrm{Lip}} \leq 1 + \frac{1}{a_\pm} ~\mbox{.}
$$
In particular, \eqref{eq:quant-cont1} implies that for all $\alpha \geq \alpha_0$, we have that
\bea\label{eq:ids-cont2}
| N_\alpha (E) - N(E) | & \leq & \max_{\pm} | n_\alpha^{(\infty)}(f_\pm) - n^{(\infty)}(f_\pm) |  + n^{(\infty)}([E-a_-, E+a_+]) \nonumber \\
  & \leq &  \gamma \left( \frac{1}{\min\{a_- ; a_+\}} + 1 \right) \eta_\alpha^{\frac{1}{1 + 2 d}} + n^{(\infty)}([E-a_-, E+a_+]) .
  \eea
  From the log-H\"older continuity (\ref{eq:ids-logholder1}) of the IDS, we see that the logarithmic decay dominates the polynomial term in \eqref{eq:ids-cont2}. Hence, letting
  $a_- = a_+ = \frac{1}{2} \eta_\alpha^\xi > 1$ with $\xi > 0$ to be determined, we obtain for all $\alpha \geq \alpha_0$:
\bea\label{eq:ids-cont3}
| N_\alpha (E) - N(E) | & \leq & \gamma (2 \eta_\alpha^{-\xi} + 1)  \eta_\alpha^{\frac{1}{1 + 2 d}} + \frac{C_I}{\log \left(\frac{1}{\eta_\alpha^\xi}\right) } \nonumber \\
 & \leq & 3 \gamma \frac{1}{ \left(\frac{1}{\eta_\alpha^\xi}\right)^{\frac{1}{\xi(1 + 2 d)} -1}} +
  \frac{C_I}{\log \left(\frac{1}{\eta_\alpha^\xi}\right) } .
    \eea
 We optimize \eqref{eq:ids-cont3} with respect to $\xi$ using the fact that $y^\beta \geq \log y$, for all $y \in (0, \infty)$, if and only if $\beta \geq \frac{1}{e}$. We find that \eqref{eq:ids-cont3} is optimized for
 $$
 \xi = \left( \frac{1}{1 + e} \right) \left( \frac{1}{1 + 2 d} \right) ,
 $$
 in which case we conclude that for all $\alpha \geq \alpha_0$,
  \beq\label{eq:ids-cont4}
| N_\alpha (E) - N(E) |  \leq  \max \{ 3 \gamma, C_I \} 2 (1+e)(1 + 2 d)\frac{1}{\log \left(\frac{1}{\eta_\alpha}\right) } ~\mbox{.}
\eeq
This verifies Theorem \ref{thm:IDS-quant1} with $C_2 = \max \{ 3 \gamma, C_I \} 2 (1+e)(1 + 2 d)$.
\end{proof}


\section{Proof of the finite-rank lemma}\label{sec:finite-rank-lemma1}
\setcounter{equation}{0}

As discussed in section \ref{sec:intro1}, the quantitative continuity statement for the DOSm in Theorem \ref{thm:quant1} reduces to the Lipschitz continuity of certain ``single-site'' spectral averages, formulated in Lemma \ref{lem:finite-rank1}. We recall the definition of the family of rank-$N$ projections $H_\lambda^{(\ell)}$ in (\ref{eq:rankN1}) and the set-up described in (\ref{eq:inclusion1}).

By the definition of the metric $d_w$ in (\ref{eq:metric}) and using that
\begin{equation} \label{eq:trivialnormbd}
\sup_{\lambda \in [-C,C]} \vert  {\rm Tr} ~ \left(P_0 f(H_\lambda^{(\ell)}) P_0 \right) \vert \leq N \Vert f \Vert_\infty ~\mbox{, for all $f \in \mathcal{C}([a,b])$, }
\end{equation}
Lemma \ref{lem:finite-rank1} follows immediately from the claim stated below in Proposition \ref{prop:finite-rank-prop2}. While we believe that a statement of the form of Proposition \ref{prop:finite-rank-prop2} should be in the literature, we were not able to find the exact formulation we needed. For the benefit of the reader and to keep the paper self-contained, we include a short proof in this section. We note that the proof we present here is a development of a proof by B. Simon for a related statement (see \cite[Proposition 2]{simon1}). An alternative argument, using the Helffer-Sj\"ostrand functional calculus, is given in section \ref{sec:appendix:alt-fr1}.

We also observe that a weaker version of Proposition \ref{prop:finite-rank-prop2}, which however requires higher regularity of $f$, can be obtained directly from standard properties of operator-valued Lipschitz functions \cite{AleksandrovPeller_review_2016}. Indeed, for functions $f: \mathbb{R} \to \mathbb{R}$ satisfying $\widehat{(f^\prime)} \in L^1(\mathbb{R})$, the operator $f(H_\lambda^{(\ell)}) - f(H_{\lambda_0}^{(\ell)}) \in \mathcal{S}_1$ is trace class (see also (\ref{eq:fourier2}) below), whence a combination of the Theorems 1.1.1 and 3.6.5 in \cite{AleksandrovPeller_review_2016} implies
\begin{align} \label{eq:commentoplip_1}
| {\rm Tr} ~ \left[ P_0 \left( f(H_\lambda^{(\ell)}) - f(H_{\lambda_0}^{(\ell)}) \right) P_0 \right] | & \leq \Vert f(H_\lambda^{(\ell)}) - f(H_{\lambda_0}^{(\ell)}) \Vert_{\mathcal{S}_1} \nonumber \\
& \leq \Vert f \Vert_{OL(\mathbb{R})} ~\Vert H_\lambda^{(\ell)} - H_{\lambda_0}^{(\ell)} \Vert_{\mathcal{S}_1} = N \Vert f \Vert_{OL(\mathbb{R})}  ~\mbox{,}
\end{align}
where the operator-valued Lipschitz norm of $f$ admits the bound
\begin{equation} \label{eq:commentoplip_2}
\Vert f \Vert_{OL(\mathbb{R})} \leq \int_\mathbb{R} \vert \widehat{(f^\prime)}(k) \vert ~\mbox{.}
\end{equation}
The estimates (\ref{eq:commentoplip_1}) - (\ref{eq:commentoplip_2}) follow in essence from the same arguments we use below, see (\ref{eq:fourier2}). The difference however is that we expand the trace on the left-hand side of (\ref{eq:commentoplip_1}) directly instead of using an upper bound by the trace-norm. This turns out to yield an estimate valid for Lipschitz $f$ (albeit at the cost of the larger constant $N^2$ on the right-hand side of (\ref{eq:rank-N5}), compared to $N$ in (\ref{eq:commentoplip_1})).

\begin{prop}\label{prop:finite-rank-prop2}
Under the assumptions of Lemma \ref{lem:finite-rank1}, for all $f \in Lip([a,b])$, with Lipschitz constant $L_f$, one has
\beq\label{eq:rank-N5}
| {\rm Tr} ~ \left[ P_0 \left( f(H_\lambda^{(\ell)}) - f(H_{\lambda_0}^{(\ell)}) \right) P_0 \right] |
\leq 2 N^2 L_f | \lambda - \lambda_0|,
\eeq
for all $\lambda, \lambda_0 \in [- C, C]$.
\end{prop}

\begin{proof}
1. We will first establish (\ref{eq:rank-N5}) for $f \in C_c^\infty(\R)$ and then use an approximation argument to extend to Lipschitz functions (see item 5. below).

\noindent
2. Let $\lambda, \lambda_0 \in [- C, C]$ be arbitrary and fixed.  For $f \in C_c^\infty(\R)$, we can express the difference of the operator on the left of \eqref{eq:rank-N5} using the Fourier transform:
\beq\label{eq:fourier1}
f(H_\lambda^{(\ell)}) - f(H_{\lambda_0}^{(\ell)}) = \frac{1}{\sqrt{2 \pi}}  \int_\R ~\hat{f} (k) \left[ e^{ik H_\lambda^{(\ell)}} - e^{ik H_{\lambda_0}^{(\ell)}}  \right] ~dk ~\mbox{.}
\eeq
Applying Duhamel's formula to the integrand of \eqref{eq:fourier1}, we obtain
\beq\label{eq:fourier2}
f(H_\lambda^{(\ell)}) - f(H_{\lambda_0}^{(\ell)}) = \frac{(\lambda - \lambda_0)}{\sqrt{2 \pi}}  \int_\R ~dk ~\widehat{(f^\prime)} (k) ~\int_0^1 ~\left[ e^{ik \theta H_\lambda^{(\ell)}} P_\ell  e^{ik (1- \theta) H_{\lambda_0}^{(\ell)}}  \right] ~d\theta .
\eeq
For use below, for $\beta \in \{ 0 , \ell \}$, let $\{ \psi_j^{(\beta)},  1 \leq j \leq N \}$ be an orthonormal basis for ${\rm Ran} ~ P_\beta$, a rank $N$ projector. The projector $P_\beta$ may then be written as
\beq\label{eq:rankN-proj1}
P_\beta = \sum_{j=1}^N | \psi_j^{(\beta)} \rangle \langle \psi_j^{(\beta)} | ~\mbox{.}
\eeq

\noindent
3. The spectral theorem allows us to define a complex-valued Borel measure $\mu_{\lambda, \lambda_0}$ on $[a,b]^2$ as follows. For $\tilde{\lambda} \in \{ \lambda, \lambda_0\}$, and $j, j' \in \{1, \dots, N\}$, we define a measure
$\mu_{\tilde{\lambda}}^{(0j,\ell j')}$ on $[a,b]$ by
\beq\label{eq:measure1}
\langle \psi_j^{(0)}, g(H_{\tilde{\lambda}}) \psi_{j'}^{(\ell)} \rangle =
\int_\R ~g(x) ~ d\mu_{\tilde{\lambda}}^{(0j,\ell j')} (x),
\eeq
for all $g \in C_c (\R)$.
We then define $\mu_{\lambda, \lambda_0}$ on $[a,b]^2$ as the product measure
\beq\label{eq:measure2}
d \mu_{\lambda, \lambda_0} (x,y) := \sum_{j, j'=1}^N  d \mu_{\lambda}^{(0j,\ell j')} (x) \otimes d\mu_{\lambda_0}^{(\ell j',0j)} (y).
\eeq
We note that \eqref{eq:measure2} implies that the total variation of the measure $\mu_{\lambda}^{(0j ,\ell j')}$ satisfies
\bea\label{eq:measure-total-var1}
| \mu_{\tilde{\lambda}}^{(0j,\ell j')} | & = & \sup_{g \in C_c(\R ), \|g\|_\infty \leq 1} \left| \int_\R ~ g(x) ~d\mu_{\tilde{\lambda}}^{(0j,\ell j')}(x) \right| \leq 1 ~\mbox{,}
 \eea
so that the total variation of $\mu_{\lambda, \lambda_0}$ can be bounded above by
\beq\label{eq:measure-total-var2}
| \mu_{\lambda, \lambda_0} | \leq 2 N^2 ~\mbox{.}
\eeq
Here, we used the fact that for two measures $\mu_1$ and $\mu_2$, with $|\mu_j| \leq 1$, for $j=1,2$, one has
$|\mu_1 \otimes \mu_2| \leq |\mu_1| + |\mu_2| \leq 2$.

\noindent
4. Using the representation in \eqref{eq:fourier2}, we obtain
\bea\label{eq:rank-N51}
\lefteqn{{\rm Tr} ~ \left[ P_0 \left( f(H_\lambda^{(\ell)}) - f(H_{\lambda_0}^{(\ell)}) \right) P_0 \right]} \\
 &=& (\lambda - \lambda_0) \int d \mu_{\lambda, \lambda_0} (x,y) ~\int_0^1 ~d\theta \int_\R
 ~\frac{dk}{\sqrt{2 \pi}} ~\widehat{(f^\prime)} (k) ~ e^{ik( \theta x + (1- \theta)y )} \nonumber \\
 &=& {(\lambda - \lambda_0)} \int ~d \mu_{\lambda, \lambda_0} (x,y) ~\int_0^1 ~d\theta ~{f^\prime}(\theta x + (1- \theta)y) ~\mbox{.}
\eea
Combining this result with estimate in \eqref{eq:measure-total-var2} thus yields
\beq\label{eq:rank-N6}
|{{\rm Tr} ~ \left[ P_0 \left( f(H_\lambda^{(\ell)}) - f(H_{\lambda_0}^{(\ell)}) \right) P_0 \right]} |
 \leq 2 N^2 |\lambda - \lambda_0| L_f([a,b]) ~\mbox{,}
\eeq
where
\begin{equation}
L_f([a,b]) = \sup_{x \neq y \in [a,b]} \left\vert \dfrac{f(x) - f(y)}{x-y} \right\vert ~\mbox{,}
\end{equation}
denotes the Lipschitz constant of $f$ on $[a,b]$. We note that this uses the fact that the measure $d \mu_{\lambda, \lambda_0}$ in (\ref{eq:rank-N51}) is supported on $[a,b]^2$.
This proves \eqref{eq:rankN1} for $f \in C_c^\infty (\R)$.

\noindent
5. To extend the result to $f \in Lip([a,b])$ we use the following simple approximation argument\footnote{We recall that $\mathcal{C}^\infty([a,b])$ is a {\em{proper}} closed subspace of $Lip([a,b])$ with norm defined in (\ref{eq:deflipnorm}), hence a simple density argument cannot be used because the right hand side of (\ref{eq:rank-N5}) depends on the Lipschitz constant.}. Without loss of generality, we may assume that $f \in Lip([a,b])$ is non-constant, in particular $L_f > 0$. Then, we may extend $f$ linearly to a compactly supported Lipschitz function $\tilde{f}$ on $\mathbb{R}$ such that $L_{\tilde{f}} = L_f$. Fixing a $\mathcal{C}^\infty$-mollifier $0 \leq \phi$, $\mathrm{supp} \phi \subseteq [-1,1]$, $\Vert \phi \Vert_1 =1$, we define, for $\epsilon > 0$, $\tilde{f}_\epsilon := \tilde{f} * \phi_\epsilon$ where $\phi_\epsilon(x) = \frac{1}{\epsilon} \phi(\frac{x}{\epsilon})$. By construction, one then has for all $\epsilon > 0$ that
\begin{equation}
L_{\tilde{f}_\epsilon}([a,b]) \leq L_{\tilde{f}}([a -\epsilon, b + \epsilon]) = L_f ~\mbox{.}
\end{equation}
Thus, using (\ref{eq:rank-N6}), we conclude for all $\lambda, \lambda_0 \in [-C,C]$ and $\epsilon > 0$ that
\begin{equation}
|{{\rm Tr} ~ \left[ P_0 \left( \tilde{f}_\epsilon(H_\lambda^{(\ell)}) - \tilde{f}_\epsilon(H_{\lambda_0}^{(\ell)}) \right) P_0 \right]} |
 \leq 2 N^2 |\lambda - \lambda_0| L_{\tilde{f}_\epsilon}([a,b]) \leq 2 N^2 |\lambda - \lambda_0| L_f ~\mbox{,}
\end{equation}
where the right-most side is {\em{independent}} of $\epsilon$. In particular, letting $\epsilon \to 0^+$, we obtain (\ref{eq:rank-N5}) for all $f \in Lip([a,b])$.
\end{proof}

We observe that Proposition \ref{prop:finite-rank-prop2} has the following immediate consequence.

\begin{corollary}\label{cor:rankN-4}
If $f \in C^k ([a,b])$, for some $k \in \N$, then the map
\beq\label{eq:rankN-4}
\lambda \in \R \mapsto {\rm Tr} \left\{ P_0 \left[ f(H_\lambda^{(\ell)}) - f(H_{\lambda_0}^{(\ell)}) \right]  P_0 \right\} =: g(\lambda),
\eeq
is in $C^k([- C, C])$, with
\beq\label{eq:rankN-error1}
\| g^{(k)} \|_\infty \leq 2 N^2 \|f^{(k)} \|_\infty.
\eeq
\end{corollary}

\begin{remark}
Analogous reasoning allows to strengthen Simon's result  \cite[Proposition 2]{simon1} by reducing the regularity assumption of $f$ from $f \in C_c^\infty (\R)$ to $f \in C_c^1 (\R)$.
\end{remark}

\begin{proof}
It suffices to consider the case $k=1$. In this case, \eqref{eq:rankN1} yields the claim if we prove that $g$ is differentiable. To this end, given $f \in C^1 ([a,b])$, we approximate $f$ by a sequence of polynomials $\{ f_n \}$ so that
\beq\label{eq:poly-approx3}
\| f_n - f \|_{C^1 ([a,b])} \stackrel{n \rightarrow \infty}{\rightarrow} 0 .
\eeq
We note that for each $n$, the function $g_n(\lambda) := {\rm Tr} \left\{ P_0 f_n(H_\lambda^{(\ell)}) P_0 \right\}$
is a polynomial in $\lambda$ and hence smooth. Now by \eqref{eq:trivialnormbd}, \eqref{eq:rank-N5}, and \eqref{eq:poly-approx3} the sequence $\{ g_n \}$ is Cauchy in $C^1 ([a,b])$ which implies that $g \in C^1([a,b])$.
\end{proof}


\section{Applications}\label{sec:appl-contDOSm1}
\setcounter{equation}{0}

In this section, we provide some applications of the quantitative continuity results in Theorems \ref{thm:quant1} and \ref{thm:IDS-quant1}.

\subsection{Continuity of the density of states in the disorder for the weak disorder regime}\label{subsec:dosm-low1}

We consider the discrete, random Schr\"odinger operator
\beq\label{eq:schr-op-low1}
H_\omega (\lambda) := H_\omega =  \Delta + \lambda \sum_{j \in {\mathcal{J}}} \omega_j P_j ~\mbox{, }
\eeq
where everything is as in [H1] with the only exception that the potential energy term is scaled by the {\em{disorder parameter}} $\lambda \geq 0$, thereby explicitly quantifying the disorder strength.
The elements of the sequence $\omega \in \Omega = [-1,1]^{\Z^d}$ are assumed to be $iid$ random variables with a common probability measure $\mu \in \mathcal{P}([-1,1])$.

As an application of the Theorems \ref{thm:quant1} and \ref{thm:IDS-quant1} we will quantify the dependence of the DOSm and the IDS (i.e. for {\em{fixed}} energy) {\em{on the disorder parameter}} $\lambda$ as $\lambda \rightarrow 0^+$, i.e. in the {\em{weak disorder regime}}. As mentioned in section \ref{subsec:previous-intro1}, for the Anderson model ($N=1$) and {\em{arbitrary}} $d \in \mathbb{N}$, the question of continuity of the IDS with respect to the disorder in the weak disorder regime and  has been addressed in \cite{hks1, schenker1} under the assumption that the single-site measure $\mu$ is absolutely continuous with a bounded density. Under this hypothesis, the authors prove H\"older continuity of the IDS with respect to $\lambda$ as $\lambda \rightarrow 0^+$ (see also remark \ref{remark:weakcoupling} below).

We explicitly note that, since Theorem \ref{thm:quant1} and \ref{thm:IDS-quant1} hold for general compactly supported probability measures, we will {\em{not}} need to assume any specific form of the measure $\mu$. The general framework developed in section \ref{sec:quan-contDOSm1} will thus allow us to drop the hypotheses of absolute continuity of the single-site measure $\mu$ imposed in \cite{hks1, schenker1}, thereby extending (by completely different means) the result of H\"older continuity of the IDS in $\lambda$ (for fixed energy) as $\lambda \rightarrow 0^+$ obtained in \cite{hks1, schenker1} to general compactly supported single-site measures $\mu$. Moreover, we add to this a quantitive continuity result on the dependence of the DOSm on $\lambda$ in the weak disorder regime.

We also note that even for $d=1$, where many more results about the $\lambda$-dependence of the IDS in the weak disorder regime exist in the literature (see section \ref{subsec:previous-intro1}), the known results had to assume a decay condition of the Fourier transform of the single-site measure, which is not needed in our work.

The key observation which allows us to view the weak disorder behavior through the framework of section \ref{sec:quan-contDOSm1} is to note that \eqref{eq:schr-op-low1} is equivalent to the random Schr\"odinger operator in \eqref{eq:schr-op2} after a rescaling of the random variables. We let
 $$
 \tilde{\omega}_k := \lambda \omega_k, ~~~k \in \mathcal{J},
 $$
which results in a {\em{recaled single-site measure}}
 \begin{equation} \label{eq:couplingrescaled meas}
 d \nu_\lambda (x) := d \mu(\frac{x}{\lambda}) ~\mbox{,}
 \end{equation}
 which, since $\mu$ is supported in $[-1,1]$, satisfies $\mathrm{supp} \nu_\lambda \subseteq [-\lambda, \lambda]$.
In particular, we see that the limit $\lambda \rightarrow 0^+$ is equivalent to the $w^\star$-limit
 \begin{equation}
 d \nu_\lambda (x)  \rightarrow d \nu(x) = \delta (x) ~dx ~.
\end{equation}
Specifically, using the metric on $\mathcal{P}([-1,1])$ defined in (\ref{eq:metric}), one has for all $0 < \lambda \leq 1$,
\begin{equation}
d_w(\nu_\lambda , d \nu) \leq \lambda ~\mbox{.}
\end{equation}
We mention that this may be viewed as a generalization of Example 1 described in section \ref{sec:qual-contDOSm1}.

Hence, application of Theorem \ref{thm:qual1} and \ref{thm:quant1} immediately yields the following qualitative and quantitative continuity of the DOSm in $\lambda$ as $\lambda \to 0^+$:
\begin{theorem}\label{thm:weak-cont1} (Weak disorder continuity of the DOSm.)
For the model \eqref{eq:schr-op-low1} with underlying single-site measure $\mu \in \mathcal{P}([-1,1])$, the DOSm $n_\lambda^{(\infty)}$ is $w^\star$-continuous
as $\lambda \rightarrow 0^+$, that is
\beq\label{eq:dosm-weak1}
n_\lambda^{(\infty)}  \stackrel{w^\star}{\longrightarrow}  n_{\lambda = 0}^{(\infty)}, ~~~~\lambda \rightarrow 0^+ ~.
\eeq
Moreover, there exists $\lambda_0 > 0$ such that for every $f \in Lip([-2d-\lambda_0, 2d + \lambda_0])$, one has
\beq\label{eq:dosm-weak2}
| n_\lambda^{(\infty)}(f) - n_{\lambda = 0}^{(\infty)}(f) | \leq \gamma \Vert f \Vert_{\mathrm{Lip}} \lambda^{\frac{1}{1 + 2 d}} ~,
\eeq
for all $0 \leq \lambda \leq \lambda_0$. The constant $\gamma$ is defined in \eqref{eq:weak-conv-dos-summ2} and
$\lambda_0$ is determined by \eqref{eq:weak-conv-dos3} as
\beq\label{eq:lambda-zero1}
\lambda_0 = \left(\frac{1}{2}\right)^{1 + 2d} ~.
\eeq
\end{theorem}

The DOSm for the case $\lambda = 0$ is that of the free Laplacian on $\ell^2(\mathbb{Z}^d)$. In particular, it is absolutely continuous, i.e. $\mathrm{d}n_{\lambda = 0}^{(\infty)}(E)= \rho_{\lambda = 0}^{(d)}(E) \mathrm{d}E$, and the density of states function (DOSf) can be expressed as
\begin{eqnarray} \label{eq:D0Sm-free1}
\rho_{\lambda = 0}^{(1)}(E) & = & \dfrac{1}{2\pi} \dfrac{1}{\sqrt{1 - (\frac{E}{2})^2}} \chi_{(-2,2)}(E) \nonumber \\
\rho_{\lambda = 0}^{(d)}(E) & = & ( \underbrace{\rho_{\lambda = 0}^{(1)} \ast \dots \ast \rho_{\lambda = 0}^{(1)} }_{\mbox{$d$-times}} )(E) ~\mbox{, for $d \geq 2$ .}
\end{eqnarray}
While for $d=1$, the free DOSf exhibits a square-root singularity at the edges of the spectrum (``van Hove singularity''), for $d \geq 2$, the expression as a $d$-fold convolution of $L^1$-functions shows that $\rho_{\lambda = 0}^{(d)} \in \mathcal{C}([-2d,2d])$, with increasing regularity as $d$ increases (see also (\ref{eq:decayFourierFreeLapl}) below).

In particular, the IDS for $\lambda = 0$ is H\"older continuous
\beq\label{eq:IDS-free1}
N_{\lambda = 0} (E+ \epsilon) - N_{\lambda = 0}(E) = n_{\lambda =0}^{(\infty)}( [E, E+\epsilon]) \leq c_0 \epsilon^\delta,
\eeq
where, for convenience, we take the constants $c_0, \delta > 0$ to be {\em{uniform in $E$}}, i.e. only depending on the dimension $d$. From (\ref{eq:D0Sm-free1}), the $E$-independent H\"older exponent for $d=1$ is  $\delta = \frac{1}{2}$, while for $d \geq 2$, one can take $\delta = 1$.

Using (\ref{eq:IDS-free1}), we thus obtain the following quantitative behavior of the IDS at weak disorder. Here, we note that the $\log$-H\"older dependence in Theorem \ref{thm:IDS-quant1} resulting from (\ref{eq:ids-logholder1}) is improved to H\"older as a consequence of (\ref{eq:IDS-free1}).
\begin{theorem}\label{thm:weak-contIDS1} (Weak disorder continuity of the IDS.)
For the model \eqref{eq:schr-op-low1} with underlying single-site measure $\mu \in \mathcal{P}([-1,1])$, there exists a constant $c_3 > 0$ such that for all $0 \leq \lambda \leq \lambda_0$
and every $E \in [-2d - \lambda_0, 2d + \lambda_0]$, one has
\beq\label{eq:weak-IDSconv1}
| N_\lambda (E) - N_{\lambda = 0}(E)| \leq c_3 \lambda^{\left( \frac{\delta}{1+\delta}\right) \left( \frac{1}{1 + 2 d} \right)}.
\eeq
where $c_3 = 2 \max \{ 3 \gamma, c_0 \}$ and $\lambda_0$ is given in (\ref{eq:lambda-zero1}).
\end{theorem}
Similar to remark \ref{thm:weak-contIDS1}, the restriction of the energy in Theorem \ref{thm:weak-contIDS1} is in principal not necessary since for energies outside the given closed interval the left-hand side of (\ref{eq:weak-IDSconv1}) is a priori zero.

\begin{remark} \label{remark:weakcoupling}
As mentioned above, Theorem \ref{eq:weak-IDSconv1} extends earlier results in \cite{hks1, schenker1} which prove the H\"older continuity of the IDS in $\lambda$ for $\lambda \to 0^+$ assuming that the underlying single-site measure $\mu$ is AC with bounded density. While we do not impose any assumptions on $\mu$, we mention that the result of \cite[Theorem 1.2]{hks1} gives a dimension independent H\"older exponent of $\frac{1}{8}$ while Theorem \ref{thm:weak-contIDS1} yields $\frac{1}{9}$ for $d=1$ (taking $\delta = 1/2$) and $\frac{1}{2(1+2d)}$, for $d \geq 2$ (taking $\delta = 1$).
\end{remark}

\begin{proof}[Proof of Theorem \ref{thm:weak-contIDS1}]
From \eqref{eq:ids-cont2}, with $\eta_\alpha = \lambda$, we take $a_- = a_+ = \frac{1}{2} \lambda^\xi$, with $\xi > 0$ to be determined, so that
\beq\label{eq:ids-weak2}
| N_\lambda (E) - N_{\lambda = 0} (E) | \leq 3 \gamma \lambda^{\frac{1}{1 + 2 d} - \xi} + c_0 \lambda^{\delta \xi}, \eeq
for all $\lambda \leq \lambda_0$, where $\lambda_0$ is as in Theorem \ref{thm:weak-cont1}. Optimizing \eqref{eq:ids-weak2} with respect to $\xi$, we arrive at $\xi = \frac{\delta}{(1 + \delta) ( 1 + 2d)}$, establishing the result.
\end{proof}

To conclude this section, we mention that Theorem \ref{thm:main} can also be used to obtain results about the dependence of the IDS and DOSm on the disorder in the regime where $\lambda > 0$. Indeed, the same rescaling argument of the probability measure as in (\ref{eq:couplingrescaled meas}) implies that for all $\lambda_0 > 0$,
\begin{equation}
\nu_\lambda \stackrel{w^\star}{\rightarrow} \nu_{\lambda_0} ~\mbox{, as $\lambda \to \lambda_0$ .}
\end{equation}
For completeness, we state the theorem below. The ``moreover'' statement in part (ii) uses the Wegner estimate (\ref{eq:wegner}) and a similar argument than used in the proof of Theorem (\ref{thm:weak-contIDS1}).
\begin{theorem} \label{thm:poscoupling}
Consider the model \eqref{eq:schr-op-low1} with underlying single-site measure $\mu \in \mathcal{P}([-1,1])$. Fix $\lambda_0 > 0$ and denote by
\begin{equation}
\eta_\lambda:= d_w(\nu_\lambda , \nu_{\lambda_0}) ~\mbox{, } \lambda > 0 ~\mbox{,}
\end{equation}
where $\nu_\lambda$ is the rescaled probability measure defined in (\ref{eq:couplingrescaled meas}).
\begin{itemize}
\item[(i)] The DOSm $n_\lambda^{(\infty)}$ is $w^\star$-continuous as $\lambda \rightarrow \lambda_0$, that is
\beq\label{eq:dosm-weak1_pos}
n_\lambda^{(\infty)}  \stackrel{w^\star}{\longrightarrow}  n_{\lambda_0}^{(\infty)}, ~~~~\lambda \rightarrow \lambda_0 ~\mbox{.}
\eeq
Moreover, there exists $\delta > 0$ such that for every $f \in Lip([-2d-(\lambda_0 + \delta), 2d + (\lambda_0+ \delta)])$, one has
\beq\label{eq:dosm-weak2_pos}
| n_\lambda^{(\infty)}(f) - n_{\lambda_0}^{(\infty)}(f) | \leq \gamma \Vert f \Vert_{\mathrm{Lip}} \eta_\lambda^{\frac{1}{1 + 2 d}},
\eeq
for all $\lambda >0$ with $\vert \lambda - \lambda_0 \vert < \delta$.
\item[(ii)] For each fixed $E \in \mathbb{R}$, the IDS $\lambda \mapsto N_\lambda(E)$ is continuous at $\lambda_0$ and there exists $\delta>0$ such that for all $\lambda >0$ with $\vert \lambda - \lambda_0 \vert < \delta$, one has
\beq \label{eq:IDS-cont1_pos}
| N_\lambda (E) - N_{\lambda_0}(E) | \leq \frac{ C_2}{ \log \left( \frac{1}{\eta_\lambda} \right)} ~\mbox{.}
\eeq
Moreover, if $d \mu = h(x) dx$ with $h \in L^1 \cap L^\infty$, then there exists $\widetilde{C_2} = \widetilde{C_2}(\Vert h \Vert_\infty, \lambda_0, N, d)$ such that (\ref{eq:IDS-cont1_pos}) is improved to
\begin{equation}
| N_\lambda (E) - N_{\lambda_0}(E) | \leq \widetilde{C_2} \eta_\lambda^\frac{1}{2(1 + 2d)} ~\mbox{.}
\end{equation}
\end{itemize}
The constants $\gamma, C_2$ are determined in the Theorems \ref{thm:quant1} and \ref{thm:IDS-quant1}.
\end{theorem}
We mention that the distance $\eta_\lambda = d_w(\nu_\lambda , \nu_{\lambda_0})$ in Theorem \ref{thm:poscoupling} can be further quantified in terms of $\lambda$, depending on the explicit form of the single-site measure $\mu$. For instance, using Examples 2 and 3 of section \ref{sec:qual-contDOSm1}, one estimates
\begin{equation} \label{eq:poscouplingestim_1}
\eta_\lambda \leq \Vert h(\frac{x}{\lambda}) - h(\frac{x}{\lambda_0}) \Vert_1 \leq C_{\lambda_0} \vert \lambda - \lambda_0 \vert^\kappa ~\mbox{,}
\end{equation}
if $d\mu(x) = h(x) dx$ with $\kappa$-H\"older continuous density function $h(x)$, or
\begin{equation} \label{eq:poscouplingestim_2}
\eta_\lambda \leq C_{\lambda_0} \vert \lambda - \lambda_0 \vert ~\mbox{,}
\end{equation}
if $\mu$ is Bernoulli. In both (\ref{eq:poscouplingestim_1}) - (\ref{eq:poscouplingestim_2}), $C_{\lambda_0}$ is a constant depending on $\lambda_0$.


\subsection{Continuity of the Lyapunov exponent in the probability distribution}\label{subsec:le1}

For $d=1$ and given $E_0 \in \C$, the {\em{Lyapunov exponent}} $L(E_0)$ characterizes the averaged growth rate of solutions to the finite-difference equation $H_\omega \psi = E_0 \psi$. By the {\em{Thouless formula}}, it may be expressed in terms of the DOSm by
\beq\label{eq:le1}
L_\nu(E_0) = \int_\R ~\log |E^\prime - E_0| ~dn_\nu^{(\infty)} (E') \in [0, \infty) .
\eeq
In this section, whenever not mentioned otherwise, $n_\nu^{(\infty)}$ denotes the DOSm for a Hamiltonian satisfying [H1] equipped with an {\em{arbitrary}} compactly supported single-site probability measure.

The validity of the Thouless formula (\ref{eq:le1}) for all $E_0 \in \C$ in particular implies the existence of non-tangential limits onto the real axis
\beq\label{eq:le2}
\lim_{\epsilon \rightarrow 0^+} L_\nu(E + i \epsilon) = L_\nu(E), ~~~ \forall E \in \R.
\eeq
Moreover, for $E_0 = E + i \epsilon \in \C \backslash \R$, the integrand in \eqref{eq:le1} is smooth.
Thus, applying Theorem \ref{thm:quant1} with $d=1$ for the function
\beq\label{eq:le-fnc1}
f(E^\prime) = \log | E + i \epsilon - E^\prime |,
\eeq
one has
\beq \label{eq:le-fnc2}
\Vert f \Vert_{\mathrm{Lip}} \leq \frac{1}{\epsilon} + \log \frac{1}{\epsilon} \leq \frac{2}{\epsilon} ~\mbox{,}
\eeq
for all $0 < \epsilon \leq \epsilon_0$ where
\begin{equation} \label{eq:lecomplexenergies_eps}
\epsilon_0:= \dfrac{1}{\sqrt{\delta(E)^2 + 1}} ~\mbox{, } \delta(E):= \max_{\pm} \vert E \pm r \vert ~\mbox{,}
\end{equation}

We hence obtain as an immediate corollary of Theorem \ref{thm:quant1}:
\begin{prop}\label{prop:LE1}
Consider the set-up in [H1]--[H2] with $d = 1$ and $E \in \R$ fixed. Then, there exists $\epsilon_0$ given in (\ref{eq:lecomplexenergies_eps}) such that for all $0 < \epsilon \leq \epsilon_0$, one has
\beq\label{eq:LE1}
| L_{\nu_\alpha} (E + i \epsilon) - L_\nu(E + i \epsilon) | \leq \gamma \frac{2}{\epsilon} \eta_\alpha^{\frac{1}{3}},
\eeq
for all $ \alpha \geq \alpha_0$. The constant $\alpha_0$ is given in \eqref{eq:weak-conv-dos3} and $\gamma$ is given in \eqref{eq:weak-conv-dos-summ2}.
\end{prop}
\begin{remark}
For energies $E \subseteq [-r,r]$, which by (\ref{eq:schr-op-sp1}) contains the almost-sure spectrum of $H_\omega$, $\epsilon_0$ in (\ref{eq:lecomplexenergies_eps}) can be chosen uniformly in the energy since $\delta(E) \leq 2r$.
\end{remark}

Adapting $\epsilon$ to the decay of $\eta_\alpha$ in Proposition \ref{prop:LE1} and using (\ref{eq:le2}), we can conclude the qualitative continuity of the Lyapunov exponent in the probability distribution:
\begin{theorem} \label{thm:qualiLE}
Consider the set-up described in [H1] and $d=1$. Then, for each fixed $E \in \mathbb{C}$, the map
\begin{equation}
\mathcal{P}([-C,C]) \ni \nu \mapsto L_\nu(E)
\end{equation}
is continuous in the weak-$^\star$ topology.
\end{theorem}
\begin{remark} \label{rem:qualilyap}
\begin{itemize}
\item[(i)]
For the case of Schr\"odinger operators, Theorem \ref{thm:qualiLE} thereby recovers recent results by Bocker and Viana \cite{BockerViana_ETDS_2017} (for $2 \times 2$-matrices) and Avila, Eskin, and Viana  \cite{AvilaEskinViana} (for $n \times n$-matrices), which establish the continuity of the Lyapunov exponents for general products of random matrices in the underlying probability measure. In their work the authors were particularly interested in weak$^*$-limits of point measures of the form considered in Example 3 of section \ref{sec:qual-contDOSm1}. We mention that the question of continuity of the Lyapunov exponents in the probability distribution goes back to a paper by Furstenberg and Kifer \cite{Furstenberg_Kiefer_IsrealJMath_1981} where already certain partial results were obtained (see Theorem B in \cite{Furstenberg_Kiefer_IsrealJMath_1981}). For a more detailed account of the history and a recent, more comprehensive list of related results, we refer the reader to section 2.3 in \cite{BockerViana_ETDS_2017}.
\item[(ii)] We mention that our methods allow to extend Theorem \ref{thm:qualiLE} to random Schr\"odinger operators on the strip, in which case the Lyapunov exponent is replaced the sum of all non-negative Lyapunov exponents; see section \ref{sec:extensions_finiteRange_strip_Strip} for further details.
\end{itemize}
\end{remark}

\begin{proof}
We show that for the set-up described in [H2] one has that $L_{\nu_\alpha(E)} \to L_\nu(E)$ as $\alpha \to \infty$, for all fixed $E \in \mathbb{C}$. It suffices to consider $E \in \mathbb{R}$, since otherwise the claim follows directly from Proposition \ref{prop:LE1}. For any fixed $0 < \zeta < \frac{1}{3}$, take $\epsilon = \eta_\alpha^\zeta$ and apply Proposition \ref{prop:LE1} and (\ref{eq:le2}). Then, as $\alpha \to \infty$, one has
\begin{align} \label{eq:lesuccesapprox}
| L_{\nu_\alpha} (E) - L_\nu(E) |  & \leq | L_{\nu_\alpha} (E) - L_{\nu_\alpha} (E + i \epsilon) | + | L_{\nu_\alpha} (E+ i \epsilon) - L_\nu(E + i \epsilon) | \nonumber \\
&  + | L_\nu(E + i \epsilon) - L_\nu(E) | = 2 o(1) + 2 \gamma \eta_\alpha^{\frac{1}{3} - \zeta} ~\mbox{,}
\end{align}
which verifies the claim.
\end{proof}

Our next goal is to obtain a {\em{quantitative}} analogue of Theorem \ref{thm:qualiLE} which characterizes the modulus of continuity of the Lyapunov exponent in the probability measure. By the successive approximation argument used in (\ref{eq:lesuccesapprox}), the latter will follow if we establish a {\em{quantitative}} version of \eqref{eq:le2} which quantifies the $o(1)$ terms in (\ref{eq:lesuccesapprox}).

To this end, we recall that a measure $\mu \in \mathcal{P}([-C,C])$ is called \emph{$\beta$-continuous at $E \in  \R$}, with $0 < \beta \leq 1$, if
there exists a constant $0 < d_\beta=d_\beta(E) < \infty$ so that
\beq\label{eq:beta-cont-m1}
\mu([E-\epsilon, E+\epsilon]) \leq d_\beta \epsilon^\beta , ~~~\forall \epsilon \geq 0.
\eeq
Applied to the DOSm $n_\nu^{(\infty)}$, condition \eqref{eq:beta-cont-m1} is equivalent to the $\beta$-H\"older continuity of the IDS {\em{locally}} at $E$, i.e.
\beq\label{eq:beta-cont-ids1}
| N_\nu(E + \epsilon) - N_\nu(E - \epsilon)| \leq d_\beta \epsilon^\beta, ~~~\forall \epsilon \geq 0.
\eeq
The $\beta$-continuity of a probability measure may be established by studying the behavior of the boundary-values of the Poisson transform $P_{n_\nu^{(\infty)}}(E + i \epsilon)$ of the DOSm $n^{(\infty)}$ defined by
\beq\label{eq:poisson-transf1}
P_{n_\nu^{(\infty)}}(E+i \epsilon) := \int_\R \frac{\epsilon}{(E - E^\prime )^2 + \epsilon^2 } ~dn_\nu^{(\infty)}(E^\prime), ~~~ \epsilon > 0.
\eeq
It is well-known (see, for example, \cite{marx1}) that \eqref{eq:beta-cont-m1} is equivalent to proving that
\beq\label{eq:poisson-transf2}
\limsup_{\epsilon \rightarrow 0^+} \epsilon^{1 - \beta} P_{n_\nu^{(\infty)}} (E + i \epsilon) < \infty.
\eeq

For points $E \in \mathbb{R}$ where the DOSm is $\beta$-H\"older continuous, we obtain the following quantitative version of \eqref{eq:le2}.
\begin{prop}\label{prop:le-quant1}
Suppose that, for some $E \in \mathbb{R}$, the DOSm $n_\nu^{(\infty)}$ is $\beta$-continuous at $E$ as specified in \eqref{eq:beta-cont-m1} with constants $0 < \beta \leq 1$ and $d_\beta$ (depending on $E$). Then, for all $\epsilon > 0 $, the Lyapunov exponent satisfies
 \beq\label{eq:le3}
 | L_\nu(E+i\epsilon) - L_\nu(E) | \leq \frac{\pi}{2 \sin \left( \frac{\pi \beta}{2} \right)} ~d_\beta \epsilon^\beta ~\mbox{.}
 \eeq
\end{prop}
\begin{remark}
Since the DOSm is in general only $\log$-H\"older continuous as quantified by (\ref{eq:ids-logholder1}), it is a valid question whether an analogue of Proposition \ref{prop:le-quant1} could be obtained which allows to drop the hypothesis of $\beta$-continuity of the DOSm at $E \in \R$, possibly resulting in a weaker modulus of continuity in $\epsilon$ on the right hand side of (\ref{eq:le3}). We address this question in Appendix \ref{sec:appendix:nontangLE1} where we conclude that if a quantitative estimate on the boundary-value of the Lyapunov can be achieved with weaker conditions on the DOSm, more information than the upper bound in (\ref{eq:ids-logholder1}) will be necessary.
\end{remark}

\begin{proof}
We first consider the auxiliary functions $f$ defined by
\beq\label{eq:le-aux1}
\eta \in [0, \infty) \mapsto f(\eta) := L_\nu(E + i \eta^{\frac{1}{\beta}}) ~\mbox{.}
\eeq

Note that
\beq\label{eq:poisson-le1}
\frac{d}{d \epsilon} L_\nu(E + i \epsilon) = P_{n_\nu^{(\infty)}} (E + i \epsilon), ~~~ \epsilon > 0, E \in \R ,
\eeq
whence by \eqref{eq:le2}, the function $f$ is continuous for $\eta \geq 0$ and differentiable for $\eta > 0$ with
\beq\label{eq:poisson-le2}
f^\prime (\eta) = \frac{1}{\beta} P_{n_\nu^{(\infty)}} ( E + i \eta^{\frac{1}{\beta}}) \eta^{\frac{1}{\beta}-1} .
\eeq
In particular, for each $\eta > 0$, there exists $0 < \eta_0 < \eta$ so that
\beq\label{eq:poisson-le3}
\left| \frac{L_\nu(E + i \eta^{\frac{1}{\beta}}) - L_\nu(E) }{\eta} \right| = \frac{1}{\beta} P_{n_\nu^{(\infty)}} ( E + i \eta_0^{\frac{1}{\beta}}) \eta_0^{\frac{1}{\beta}-1}.
\eeq
Using the change of variables $\eta = \epsilon^{\beta}$, we thus see that
\beq\label{eq:poisson-le4}
\left| \frac{L_\nu(E + i \epsilon) - L_\nu(E) }{\epsilon^\beta} \right| = \frac{1}{\beta} P_{n_\nu^{(\infty)}} ( E + i \epsilon_0) \epsilon_0^{1-\beta},
\eeq
for some $0 < \epsilon_0 < \epsilon$.
To examine the boundary-value behavior of (\ref{eq:poisson-le4}), we define the function $M_{n^{(\infty)}}^E$ by
$$
M_{n_\nu^{(\infty)}}^E (\delta) : = n_\nu^{(\infty)}([E-\delta, E+\delta]) ~\mbox{, for $\delta > 0$.}
$$
The right hand side of (\ref{eq:poisson-le4}) may be expressed in the form
\bea\label{eq:poisson-le5}
\epsilon_0^{1-\beta}  P_{n_\nu^{(\infty)}} ( E + i \epsilon_0) & = &  \epsilon_0^{1-\beta} \int_0^{\infty} \frac{\epsilon_0}{\delta^2 + \epsilon_0^2} ~dM_{n_\nu^{(\infty)}}^E (\delta)  \nonumber \\
 & \leq & \epsilon_0^{2-\beta} d_\beta \int_0^{\infty} \frac{2\delta^{\beta + 1}}{[\delta^2 + \epsilon_0^2]^2} ~d\delta = \frac{\beta \pi}{2 \sin \left( \frac{\pi \beta}{2} \right)} d_\beta ~.
 \eea
Therefore, combining \eqref{eq:poisson-le4}--\eqref{eq:poisson-le5}, we obtain the claim.
 \end{proof}

\begin{remark} \label{remark:generalizKontani}
\begin{itemize}
\item[(i)] The computation in \eqref{eq:poisson-le5} played an important role in the proof of  \cite[Proposition 3.2]{marx1}.
\item[(ii)] Equality \eqref{eq:poisson-le4} implies
\beq\label{eq:poisson-le6}
\limsup_{\epsilon \rightarrow 0^+}
\left| \frac{L_\nu(E + i \epsilon) - L_\nu(E) }{\epsilon^\beta} \right| = \limsup_{\epsilon \rightarrow 0^+} \frac{1}{\beta} P_{n_\nu^{(\infty)}} ( E + i \epsilon) \epsilon^{1-\beta},
\eeq
which expresses the local continuity properties of the DOSm as a fractional derivative of the Lyapunov exponent.
Result \eqref{eq:poisson-le6} generalizes the starting point of Kotani theory where, for $E \in \R$ with $L_\nu(E) = 0$ and $\beta = 1$, one has
\beq\label{eq:le-kotani1}
\limsup_{\epsilon \rightarrow 0^+} \left( \frac{L_\nu(E + i \epsilon)}{\epsilon} \right) =  \limsup_{\epsilon \rightarrow 0^+} P_{n_\nu^{(\infty)}}(E + i \epsilon) .
\eeq
In this case, the theorem of de la Vall\'ee Poussin guarantees that the right side of \eqref{eq:le-kotani1} is {\em{a priori finite}} for Lebesgue a.e.\ $E \in \R$.
\item[(iii)] From general quantitative results on the continuity of the Lyapunov exponent of random cocycles (for fixed underlying probability measure!), it could be directly extracted that for all energies $E \in \mathbb{R}$ with $L_\nu(E) > 0$, the function $\epsilon \mapsto L_\nu(E + i \epsilon)$ is H\"older continuous in $\epsilon$, see \cite{DuarteKlein_monograph}, Theorem 5.1 therein. In the statement of Proposition \ref{prop:le-quant1}, we however do {\em{not}} assume that $L_\nu(E) > 0$.
\end{itemize}
\end{remark}

Proposition \ref{prop:le-quant1} shows that the continuity of $L_\nu(E + i \epsilon)$ as $\epsilon \rightarrow 0^+$ is determined by the continuity of the DOSm locally at $E$. Since $w^*$-convergence of measures does not, in general, preserve local continuity properties of measures, given Proposition  \ref{prop:le-quant1}, the following hypothesis will be necessary to extrapolate the results of Proposition \ref{prop:LE1} to the real line.
\begin{description}
\item[{[H3]}] Assuming [H1]--[H2] and $d=1$, suppose that $E \in \R$ satisfies the following conditions: there exists a constant $0 < D < \infty$ and an exponent $0 < \beta \leq 1$ such that for all $\epsilon > 0$:
   \bea\label{eq:ids-unif1}
    | N_\alpha (E + \epsilon) - N_\alpha(E - \epsilon) |  & \leq  & D \epsilon^\beta, ~~~ \forall \alpha  \\
     | N (E + \epsilon) - N(E - \epsilon) |  & \leq  & D \epsilon^\beta.
\eea
\end{description}
Under this hypothesis, we can prove the following result.

\begin{theorem}\label{thm:LE-quant1}
Consider the set-up described in [H1]-[H2] for $d=1$. Assume that [H3] holds for a given $E \in \R$ for some constants $\beta , D$ as in (\ref{eq:ids-unif1}), which may depend on $E$. Then, there exists $\alpha_L \in \mathbb{N}$ determined in (\ref{eq:thresholdle}) such that for all $\alpha \geq \alpha_L$,
\beq\label{eq:LE3}
| L_{\nu_\alpha} (E) - L_\nu(E) | \leq C_L \eta_\alpha^{\frac{1}{3}\left( \frac{\beta}{\beta + 1}\right)},
\eeq
where the constant $C_L = C_L(D, \beta)$ is given in (\ref{eq:le-const1}).
\end{theorem}
\begin{remark} \label{rem:quantilyap}
As for Theorem \ref{thm:qualiLE} (see remark \ref{rem:qualilyap}, part (ii)), we mention that the result of Theorem \ref{thm:LE-quant1} extends to random Schr\"odinger operators on the strip, in which case the Lyapunov exponent is replaced by the sum of all non-negative Lyapunov exponents; see section \ref{sec:extensions_finiteRange_strip_Strip} for further details.
\end{remark}

\begin{proof}
Take $\alpha_L \in \mathbb{N}$ such that
\begin{align} \label{eq:thresholdle}
\alpha_L \geq \alpha_0 ~\mbox{ and } \eta_{\alpha_L}^{\zeta_0} = \dfrac{1}{\delta(E)^2 + 1} ~\mbox{, }
\end{align}
where
\begin{equation}
\zeta_0:= \frac{1}{3(1 + \beta)} ~\mbox{.}
\end{equation}

Then, combining the Propositions \ref{prop:LE1} and \ref{prop:le-quant1} with [H3], we conclude similar to (\ref{eq:lesuccesapprox}) that for fixed $\zeta_0 \leq \zeta < \frac{1}{3}$ and for all $\alpha \geq \alpha_L$
\bea\label{eq:le-cont2}
| L_\alpha (E) - L(E) | & \leq & \frac{\pi D}{\sin \left( \frac{\pi \beta}{2} \right)} \eta_\alpha^{\zeta \beta} + 2 \gamma \eta_\alpha^{ \frac{1}{3} - \zeta} ~\mbox{.}
\eea
As the right hand side of (\ref{eq:le-cont2}) is optimized for $\zeta = \zeta_0$, we otain that claim with
\beq\label{eq:le-const1}
C_L = 2 \max \left\{ 2 \gamma, \frac{\pi D}{\sin \left( \frac{\pi \beta}{2} \right)} \right\} ~\mbox{.}
\eeq
\end{proof}

\subsubsection{The Lyapunov exponent in the weak disorder limit}
To conclude our discussion in this section, we apply Theorem \ref{thm:LE-quant1} to quantify the $\lambda$-dependence of the Lyapunov exponent in the weak disorder limit for the model given in (\ref{eq:schr-op-low1}). The dependence of the map $\lambda \mapsto L_\lambda(E)$ for fixed $E \in \mathbb{R}$ has been studied in several earlier papers. To provide some context for our discussion, we will briefly summarize some of the available results. For a more detailed account of the known results, we refer the reader to, e.g., \cite{Schulz-Baldes_GAFA_2004, Schulz-Baldes_ OperThyAdvAppl_2007}.

In their monograph \cite{PasturFigotin_book} (Theorem 14.6, therein), Pastur and Figotin use a perturbative argument to prove an asymptotic formula for the $L_\lambda(E)$ near $\lambda = 0$. Specifically, they show that for all $E \in (-2,2) \setminus \{0\}$, one has
\begin{equation} \label{eq:leolambda2}
L_\lambda(E) = c(E) \lambda^2 (1 + \mathcal{O}(\lambda)) ~\mbox{.}
\end{equation}
This result was later generalized to the case of random Schr\"odinger operators on the strip by Schulz-Baldes in \cite{Schulz-Baldes_GAFA_2004}, Theorem 2 therein, in which case a finite set of energies has to be excluded. For Schr\"odinger operators on $\mathbb{Z}$ with strongly mixing potentials, the $\mathcal{O}(\lambda^2)$ dependence was shown by Bourgain and Schlag in \cite{bourgain_schlag_CMP_2000}. We also mention that the $\mathcal{O}(\lambda^2)$-dependence in (\ref{eq:leolambda2}) is expected on physical grounds \cite{KappusWegner_1981, Thouless_PRL_1977}.

Assuming an appropriate decay of the Fourier transform of the single-site probability measure, Speis proved that for all energies $E \in (-2,2)$, the map $\lambda \mapsto L_\lambda(E)$ is continuous near $\lambda = 0$. The latter result was based on the super-symmetric replica method and develops ideas by Campanino and Klein \cite{campanino-klein86}, who, under similar assumptions, had established that $\lambda \mapsto L_\lambda(E)$ is $\mathcal{C}^\infty$ near $\lambda = 0$ for a certain {\em{dense}} set of energies $E \neq 0$ known as the {\em{Kappus-Wegner anomalies}}.

Even though, Theorem \ref{thm:LE-quant1} cannot reproduce the expected $\mathcal{O}(\lambda^2)$ dependence in (\ref{eq:leolambda2}), our method has the advantage that it does not break down at the center of the band $E = 0$ and also has the potential to address the band edges, $E = \pm 2$.

First recall that for the free Laplacian (i.e. $\lambda = 0$ in (\ref{eq:schr-op-low1})) and $d=1$ the spectrum is the closed interval $[-2,2]$ and the Lyapunov exponent satisfies
\begin{equation} \label{eq:lyapfreelaplacian}
L_{\lambda = 0}(E) = \begin{cases} 0  & ~\mbox{, if $E \in [-2,2]$ ,} \\ \log \left\vert  \dfrac{E + \sqrt{E^2 - 4}}{2}    \right\vert > 0 & ~\mbox{, if $E \in \mathbb{C} \setminus [-2,2]$ .} \end{cases}
\end{equation}

The behavior of $\lambda \mapsto L_\lambda(E)$ as $\lambda \to 0^+$ for fixed energies $E \in \mathbb{C} \setminus [-2,2]$ is straight-forward and follows from arguments along the lines of Proposition \ref{prop:LE1}. Moreover, general quantitative results on the continuity of the Lyapunov exponents for a certain large class of random cocycles \cite{DuarteKlein_monograph}, see Theorem 5.1 therein, a-priori imply that the map $\lambda \mapsto L_\lambda(E)$ is H\"older continuous at $\lambda = 0$ at all $E$ where $L_{\lambda = 0}(E) > 0$. By (\ref{eq:lyapfreelaplacian}), this is satisfied for all $E \in \mathbb{C}\setminus [-2,2]$.

We will thus focus on the more interesting situation where $E \in [-2,2]$, which will be handled as an application of Theorem \ref{thm:LE-quant1}. Here, we also mention that while Theorem 5.1 in \cite{DuarteKlein_monograph} does predict the {\em{qualitative}} continuity of the Lyapunov exponent as $\lambda \to 0^+$, i.e.
\begin{equation} \label{eq:layplambdaDuartKlein}
\lim_{\lambda \to 0^+} L_{\lambda}(E) = L_{\lambda =0}(E) ~\mbox{, for each } E \in [-2,2] ~\mbox{,}
\end{equation}
conclusions about the modulus of continuity based on \cite{DuarteKlein_monograph} are not possible since $L_{\lambda = 0}(E) = 0$ for every $E \in [-2,2]$.

In order to apply Theorem \ref{thm:LE-quant1} to the weak-disorder limit for the model described in (\ref{eq:schr-op-low1}), hypothesis [H3] needs to be verified. First, observe that the explicit expression for the DOSf of the free Laplacian in dimension $d=1$ given in (\ref{eq:D0Sm-free1}), limits $\beta$ in [H3] to $\beta = 1$ for $E \in (-2,2)$ where the DOSf for the free 1d-Laplacian is locally smooth, and to $\beta = \frac{1}{2}$ at the edges of the spectrum $E = \pm 2$ (``van Hove-singularity'').

For all $E \in (-2,2)$ we can verify [H3] as an application of \cite{hstoiciu} (see Theorem 2, therein), which establishes the continuity of the DOSf in $\lambda$ for {\em{fixed}} $E \in (-2,2)$ as $\lambda \to 0^+$, subject to the hypothesis that the Fourier transform $\chi(t):= \frac{1}{\sqrt{2 \pi}} \int \mathrm{e}^{-itE} d\mu(E)$ of the single-site measure $\mu$ is smooth and satisfies the decay condition
\begin{equation} \label{eq:decaycharactfun_cond}
\lim_{|t| \rightarrow \infty} \chi^{(j)} (t) = 0 ~\mbox{, for all $j \geq 0$ .}
\end{equation}

In particular, this implies that for each $E \in (-2,2)$ and $0< \lambda_0$ as in (\ref{eq:lambda-zero1}), the quantity $\sup_{0 \leq \lambda \leq \lambda_0} \rho_{\lambda}(E) < + \infty$
exists and is finite, whence the constants in [H3] can be taken to be
\begin{eqnarray} \label{eq:lyapweakinsideconstant}
D := \sup_{0 \leq \lambda \leq \lambda_0} \rho_{\lambda}(E) ~\mbox{, } \beta =1 ~\mbox{.}
\end{eqnarray}

Thus, for all $E \in (-2,2)$, we can apply Theorem \ref{thm:LE-quant1}, which results in:
\begin{theorem} \label{thm:lyapweakcoupling}
Consider the model described in (\ref{eq:schr-op-low1}) with a single-site measure $\mu \in \mathcal{P}([-1,1])$ which satisfies the decay condition in (\ref{eq:decaycharactfun_cond}). Then there exists $\lambda_L > 0$, such that for all $E \in (-2,2)$ and $0 \leq \lambda \leq \lambda_L$, one has
\begin{equation}
0 \leq L_{\lambda}(E) \leq C_L \lambda^{1/6} ~\mbox{.}
\end{equation}
Here, $C_L$ is given in (\ref{eq:le-const1}) with constants $\beta = 1$ and $D$ as in (\ref{eq:lyapweakinsideconstant}) and $\lambda_L$ can be taken as $\lambda_L = \frac{1}{37^3}$.
\end{theorem}

Finally, we mention that in principle, Theorem \ref{thm:LE-quant1} has also potential to yield a result for the band-edges $E = \pm 2$, {\em{provided one can show}} that for some $\lambda_1 > 0$,
\begin{equation} \label{eq:lebandedge}
\sup_{0 \leq \lambda \leq \lambda_1} \dfrac{ N_\lambda(E+\epsilon) - N_\lambda(E-\epsilon) }{\epsilon^\frac{1}{2}} ~\mbox{,}
\end{equation}
which would imply that [H3] is satisfied.
\begin{theorem}
Consider the model described in (\ref{eq:schr-op-low1}) and suppose one can show that (\ref{eq:lebandedge}) holds for $E = \pm 2$. Then there exists $\lambda_L^{(b)} > 0$, such that for all $0 \leq \lambda \leq \lambda_L^{(b)}$,
\begin{equation}
0 \leq L_{\lambda}(E) \leq C_L^{(b)} \lambda^{1/6} ~\mbox{.}
\end{equation}
Here, $C_L^{(b)}$ is given in (\ref{eq:le-const1}) with constants $\beta = 1/2$ and $D$ given by (\ref{eq:lebandedge}) and $\lambda_L$ can be taken as $\lambda_L = \min\{\frac{1}{37^{9/2}}, \lambda_1\}$ where $\lambda_1$ is so that (\ref{eq:lebandedge}) holds.
\end{theorem}

\subsection{Continuity of the density of states function in the probability distribution}\label{subsec:dos-low1}

In this final application, we examine the potential implications of Theorem \ref{thm:IDS-quant1} for the DOSf. For this, we will assume:
\begin{description}
\item [{[H4]}] Both the DOSm $n_\alpha^{(\infty)}$, for all $\alpha$, and $n^{(\infty)}$ are absolutely continuous. This implies the existence of density of states functions $\rho_\alpha(E)$ and $\rho(E)$ so that
\beq\label{eq:dos1a}
d n_\alpha^{(\infty)}(E) =: \rho_\alpha (E) ~dE, \forall \alpha \in \N
\eeq
and
\beq\label{eq:dos1b}
d n^{(\infty)}(E) =: \rho (E) ~dE .
\eeq
\end{description}

Given the assumptions [H1]--[H2], the DOSf $\rho_\alpha (E)$ and $\rho (E)$ are trivially of compact support, supported in $[-r,r]$. We recall that since $w^\star$-limits do not in general preserve the components of the Lebesgue decomposition of a sequence of measures, it will be necessary for us to assume absolute continuity of the limit $n^{(\infty)}$ even if the elements of the sequence $n_\alpha^{(\infty)}$ are absolutely continuous.

Following, denote by $\hat{\rho_\alpha}$ and $\hat{\rho}$ the Fourier transform of, respectively, $\rho_\alpha(E)$ and $\rho(E)$. Applying Theorem \ref{thm:quant1} for $f(E) = \frac{1}{\sqrt{2 \pi}} e^{-itE}$, we obtain for $\alpha \geq \alpha_0$ and $t \in \R$:
\beq\label{eq:dos-ft1}
| \hat{\rho}_\alpha(t) - \hat{\rho}(t) | \leq \gamma (|t| + 1) \eta_\alpha^{\frac{1}{1 + 2 d}} ~.
\eeq
We observe that, upon replacing the Fourier transforms of the DOSf by the Fourier transforms of the DOSm, (\ref{eq:dos-ft1}) holds even without assuming absolute continuity of the DOSm as in [H4].

Since the right side of \eqref{eq:dos-ft1} is linear in $t$, we need to impose decay conditions of $\hat{\rho}_\alpha$ and $\hat{\rho}$ in order to be able to take the inverse Fourier transform. We make the following assumption.
\begin{description}
\item[{[H5]}] We assume [H1]--[H2] and [H4]. In addition, we suppose that there exists a constant $0 < D_1 < \infty$ and $\epsilon > 0$ such that both of the following holds for all $t \in \R$:
    \bea
    | \hat{\rho}_\alpha (t) | & \leq  & \frac{D_1}{|t|^{1 + \epsilon}} , ~~~ \forall \alpha \in \N   \label{eq:ft-dos-decay1} \\
    | \hat{\rho} (t) | & \leq  & \frac{D_1}{|t|^{1 + \epsilon}} . \label{eq:ft-dos-decay2}
    \eea
\end{description}
Note that the constants are independent of $\alpha$.

\begin{theorem}\label{thm:dos1}
Consider the set-up described in [H1]--[H2] and assume that [H4]--[H5] hold. Then, for all $\alpha \geq \alpha_0$ with $\alpha_0$ as in (\ref{eq:weak-conv-dos3}) and $E \in \R$, we have:
\beq\label{eq:dos2}
|{\rho}_\alpha(E) - {\rho}(E) | \leq C \eta_\alpha^{\frac{\epsilon}{2 + \epsilon}\frac{1}{1 + 2 d}} ~\mbox{,}
\eeq
where the constant $C = C(D_1, \epsilon, d)$ is given in \eqref{eq:dos-cnst1}.
\end{theorem}

\begin{proof}
Let $\alpha \geq \alpha_0$ and $1 \leq A$ to be determined. Using \eqref{eq:dos-ft1} and [H5], the inverse Fourier transform gives
\bea\label{eq:dos3}
| \rho_\alpha(E) - \rho(E) | & \leq & \int_{-A}^A | \hat{\rho}_\alpha(t) - \hat{\rho}(t) | \frac{dt}{\sqrt{2 \pi}} +
\int_{|t| \geq A} | \hat{\rho}_\alpha(t) - \hat{\rho}(t) | \frac{dt}{\sqrt{2 \pi}} \nonumber \\
 & \leq & \frac{4 \gamma}{\sqrt{2 \pi}} A^2 \eta_\alpha^{\frac{1}{1 + 2 d}} +  \frac{4 D_1}{\epsilon \sqrt{2 \pi}} A^{- \epsilon} .
\eea
Taking $A = \eta_\alpha^{- \xi}$, for some $0 < \xi$ and optimizing in $\xi$, we obtain
\beq\label{eq:dos4}
| \rho_\alpha(E) - \rho(E) | \leq 4 \sqrt{ \frac{2}{\pi} } \max \{ \gamma ; \frac{D_1}{\epsilon} \} \eta_\alpha^{\frac{\epsilon}{2 + \epsilon}\frac{1}{1 + 2 d}},
\eeq
which determines the constant $C$ in \eqref{eq:dos2} as
\beq\label{eq:dos-cnst1}
C = C(D_1, \epsilon, d) = 4 \sqrt{ \frac{2}{\pi} } \max \{ \gamma ; \frac{D_1}{\epsilon} \} ~.
\eeq
\end{proof}
We note that the estimate \eqref{eq:dos2} is uniform in $E$ as a result of the non-locality of the Fourier transform.

\subsubsection{The density of states function in the weak-disorder regime}
To conclude this section, we comment on the application of Theorem \ref{thm:dos1} to capture the behavior of the DOSf in the weak disorder regime. Considering the model described in (\ref{eq:schr-op-low1}), application of Theorem \ref{thm:dos1} amounts to taking $\eta_\alpha = \lambda \to 0^+$, {\em{provided that one can show that the hypotheses [H4]--[H5] hold}}.

To start, we note that for any $d \in \mathbb{N}$ and $\lambda > 0$, [H4] always holds (for $\eta_\alpha = \lambda$) if the single-site measure $\mu$ underlying the model in (\ref{eq:schr-op-low1}) is AC with bounded density, i.e. if
\begin{equation} \label{eq:boundeddensity}
d \mu(x) = h(x) d x ~\mbox{, for some } 0 \leq h \in L^1\cap L^\infty([-1,1]) ~\mbox{with } \Vert h \Vert_1 = 1 ~\mbox{.}
\end{equation}
Indeed by the Wegner estimate, (\ref{eq:boundeddensity}) implies that $\rho_\lambda$ exists for all $\lambda > 0$ as a function in $L^1 \cap L^\infty([-\lambda,\lambda])$ (in particular, $\rho_\lambda(E)$ is defined for Lebesgue a.e. $E \in [-\lambda, \lambda]$) and satisfies
\begin{equation} \label{eq:wegner}
\Vert \rho_\lambda \Vert_\infty \leq \frac{N \Vert h \Vert_\infty}{\lambda} ~\mbox{.}
\end{equation}

In view of [H5], we also observe that (\ref{eq:ft-dos-decay2}) automatically holds for the free Laplacian for dimensions $d \geq 3$:
\begin{prop} \label{prop:freelapl_dosf}
For the free Laplacian $H = \Delta$ on $\ell^2(\mathbb{Z}^d)$, one has for every $d \in \mathbb{N}$
\begin{equation} \label{eq:decayFourierFreeLapl}
\vert \hat{\rho}_{\lambda = 0}^{(d)}(t) \vert \leq \dfrac{2^{-d} \pi^{-d/2}}{\vert t \vert^{d/2}} ~\mbox{.}
\end{equation}
Here, $\hat{\rho}_{\lambda = 0}^{(d)}$ is the Fourier-transform of the DOSf for the free Laplacian as given in (\ref{eq:D0Sm-free1}) and the exponent in the decay on the right-hand side of (\ref{eq:decayFourierFreeLapl}) is sharp.
\end{prop}
\begin{proof}
Using the explicit expression of the DOSf in (\ref{eq:D0Sm-free1}), we obtain for $d=1$
\begin{equation}
\hat{\rho}_{\lambda = 0}^{(1)}(t) = \dfrac{1}{\sqrt{2 \pi}} J_0(2 t) ~\mbox{,}
\end{equation}
where $J_0$ is the zeroth-order Bessel function of the first kind. In particular, for arbitrary $d \in \mathbb{N}$, (\ref{eq:D0Sm-free1}) implies that
\begin{equation} \label{eq:dosffouriertr:lapl}
\hat{\rho}_{\lambda = 0}^{(d)}(t) = \left[ \dfrac{1}{\sqrt{2 \pi}} J_0(2 t) \right]^d ~\mbox{.}
\end{equation}

The decay estimate claimed in (\ref{eq:decayFourierFreeLapl}) thus follows using that $\vert J_0(y) \vert \leq \vert y \vert^{-1/2}$ for all $y \in \mathbb{R}$.
\end{proof}

In summary, for an AC single-site measure $\mu$ satisfying (\ref{eq:boundeddensity}), Theorem \ref{thm:dos1} would immediately produce the following corollary if we assume that (\ref{eq:ft-dos-decay1}) in [H5] holds for $d \geq 3$:
\begin{theorem} \label{thm:contiDOSfweakdisorder}
Consider the model described in (\ref{eq:schr-op-low1}) with an AC single-site measure $\mu$ satisfying (\ref{eq:boundeddensity}) and $d \geq 3$. Assume that (\ref{eq:ft-dos-decay1}) in [H5] holds for $0 < \eta_\alpha = \lambda \leq \widetilde{\lambda_0}$, for some $0<\widetilde{\lambda_0}$ and $0 < D_1, \epsilon$. Then, there exists $\lambda_1 > 0$, such that $\rho_\lambda$ is continuous in $L^\infty$-norm as $\lambda \to 0^+$ and satisfies:
\begin{equation} \label{eq:contiDOSfweakdisorder}
\Vert \rho_\lambda - \rho^{(d)}_{\lambda =0} \Vert_\infty \leq C^\prime \lambda^\delta ~\mbox{, for all } 0 \leq \lambda \leq \lambda_1 ~\mbox{.}
\end{equation}
Here, $\lambda_1 = \min\{ \widetilde{\lambda_0} , \lambda_0\}$ with $\lambda_0$ given in (\ref{eq:lambda-zero1}), $\delta > 0$ is determined by
\begin{equation}
\delta = \left( \dfrac{\epsilon^\prime}{2 + \epsilon^\prime} \right) \left(\dfrac{1}{1 + 2d} \right) ~\mbox{, with } \epsilon^\prime = \min \left\{ \epsilon ~;~ \frac{d}{2} - 1 \right\} ~\mbox{,}
\end{equation}
and
\begin{equation}
C^\prime = 4 \sqrt{ \frac{2}{\pi} } \max \{ \gamma, \frac{D_1^\prime}{\epsilon^\prime} \} ~\mbox{, where }  D_1^\prime = \max\{ D_1 ~;~ 2^{-d} \pi^{-d/2} \} ~\mbox{.}
\end{equation}
\end{theorem}
We mention that the definitions of the constants $\epsilon^\prime$ and $D_1^\prime$ in Theorem \ref{thm:contiDOSfweakdisorder} take into account the decay of the Fourier-transform for the free Laplacian in (\ref{eq:decayFourierFreeLapl}).

While we do conjecture that (\ref{eq:ft-dos-decay1}) in [H5] holds for $\eta_\alpha = \lambda$ and $d \geq 3$, at the present moment we cannot provide a proof. Further evidence in favor of this conjecture is however provided by a result of Carmona and Lacroix in their monograph \cite{carmonalacroix}, where they prove an explicit expression for the Fourier transform of the DOSm for the case that the random variables in the Hamiltonian (\ref{eq:schr-op-low1}) with $N=1$ are iid according to a {\em{Cauchy distribution}} (also known as the {\em{Lloyd model}}), i.e.
\begin{equation} \label{eq:cauchy}
d \nu_\lambda(x) = \dfrac{1}{\pi} \dfrac{\lambda}{x^2 + \lambda^2} ~\mbox{ , $\lambda > 0$.}
\end{equation}
In this situation, based on a discrete version of the Feynman-Kac formula and ideas going back to Molchanov they show (see the discussion of the Lloyd model on p. 329 in \cite{carmonalacroix}) that
\begin{equation} \label{eq:carmonalacroix_ftdosf}
\widehat{\rho}_\lambda(t) = \widehat{\rho}_{\lambda = 0}^{(d)}(t) \mathrm{e}^{-\lambda \vert t \vert} ~\mbox{,}
\end{equation}
where $\widehat{\rho}_{\lambda = 0}^{(d)}(t)$ is given in (\ref{eq:dosffouriertr:lapl}). We emphasize that their proof of (\ref{eq:carmonalacroix_ftdosf}) crucially relies on the fact that the Fourier transform of a Cauchy distribution is an exponential. Note that  (\ref{eq:carmonalacroix_ftdosf}) in particular implies that $\rho_\lambda = \rho_{\lambda = 0}^{(d)} * h_\lambda$.

While our approach developed in the sections \ref{sec:qual-contDOSm1}-\ref{sec:quan-contDOSm1} requires compactness of the support of the single-site distribution and hence does not have an immediate extension to e.g. (\ref{eq:cauchy}), we note that a simple modification of the cut-off argument used in the proof of Theorem \ref{thm:dos1} allows to however give a direct proof of the continuity of the DOSf as $\lambda \to 0^+$ for the Lloyd model (\ref{eq:cauchy}), at least if $d \geq 3$. Since this was not addressed in \cite{carmonalacroix}, we add the argument here for completeness:
\begin{theorem}
Consider the Lloyd model, i.e. the Hamiltonian in (\ref{eq:schr-op-low1}) for $N=1$ with a single-site distribution given by (\ref{eq:cauchy}). For dimensions $d \geq 3$, the DOSf is continuous as $\lambda \to 0+$ and satisfies
\begin{equation}
\Vert \rho_{\lambda} - \rho_{\lambda = 0} \Vert_\infty \leq D_L \lambda^{\frac{d-2}{d+2}} ~\mbox{.}
\end{equation}
Here, the constant $D_L = D_L(d)$ is determined explicitly in (\ref{eq:hoeldercarmonlacroix}).
\end{theorem}
\begin{proof}
For $A > 0$ to be determined later, using (\ref{eq:carmonalacroix_ftdosf}), inverse Fourier transform yields in analogy to (\ref{eq:dos3}),
\bea \label{eq:cl_1}
\vert \rho_\lambda(E) - \rho_{\lambda = 0}(E) \vert \leq \int_{-A}^A \vert  \hat{\rho}_{\lambda = 0}^{(d)}(t) \vert \vert \mathrm{e}^{-\lambda \vert t \vert} - 1 \vert ~\dfrac{d t}{\sqrt{2 \pi}} + \sqrt{\dfrac{2}{\pi}} \int_{\vert t \vert \geq A} \vert  \hat{\rho}_{\lambda = 0}^{(d)}(t) \vert ~dt ~\mbox{,}
\eea
where the inequality holds for Lebesgue a.e. $E$ and in $L^2(\mathbb{R})$.

Using the decay estimate (\ref{eq:decayFourierFreeLapl}) for the second integral in (\ref{eq:cl_1}) and that, for all $y \in \mathbb{R}$ on has, $\vert \mathrm{e}^{-y} - 1 \vert \leq \mathrm{e} \vert y \vert$ and $\vert J_0(y) \vert \leq 1$, to estimate the first integral in (\ref{eq:cl_1}), we obtain
\bea
\Vert \rho_\lambda - \rho_{\lambda = 0} \Vert_\infty \leq \dfrac{\mathrm{e}}{(2 \pi)^{d + 1/2}} \cdot \lambda A^2 + \dfrac{2^{-d + 5/2} \pi^{-(d+1)/2}}{d-2} A^{-\frac{d}{2} + 1} ~\mbox{.}
\eea

Hence, taking $A = \lambda^{-\xi}$ where $\xi > 0$ is to be optimized, we obtain the claim with $\xi = \frac{2}{d+2}$ and
\begin{equation} \label{eq:hoeldercarmonlacroix}
D_L = \dfrac{1}{(d-2) 2^{d - \frac{9}{2}} \pi^{(d+1)/2}     } ~\mbox{.}
\end{equation}
\end{proof}


\section{Extensions and generalizations }\label{sec:extensions}

The purpose of this section is to outline a few generalizations of the developed theory. As pointed out at the end of section \ref{sec:intro1}, our results are not limited to the random lattice Schr\"odinger operators specified in [H1]. In section \ref{sec:extensions_necfeat}, we therefore list the necessary features a model needs to possess in order to be amenable to the framework developed in this paper. The subsequent sections then consider a few examples of models which have been of particular interest in the literature, specifically the finite-range Anderson model and random Schr\"odinger operators on the strip (section \ref{sec:extensions_finiteRange_strip}), as well as the Anderson model on the Bethe lattice (section \ref{sec:extensions_bethe}).

\subsection{Necessary features of the model} \label{sec:extensions_necfeat}

While the main result of this paper in Theorem \ref{thm:main} was formulated for the specific model described in [H1], its proof presented in the sections \ref{sec:qual-contDOSm1}--\ref{sec:quan-contDOSm1} in fact only relies on the following necessary features of the model, which are shared by many more discrete random operators.

To list these necessary features, let $\mathcal{J} \neq \emptyset$ be a countable index set, $(\mathcal{K},\rho)$ be a fixed compact metric space, $\mathbb{G}$ be an infinite graph, and $L \in \mathbb{N}$ be given. Suppose that for each $\omega \in \Omega = \mathcal{K}^\mathcal{J}$, $H_\omega$ is a bounded self-adjoint operator on $\ell^2(\mathbb{G}; \mathbb{C}^L)$ and that the elements of $\omega = (\omega_j)_{j \in \mathcal{J}} \in \Omega$ are idd random variables, distributed according to a common Borel probability measure $\nu \in \mathcal{P}(\mathcal{K})$. Let
\beq
\nu^{(\infty)} := \bigotimes_{k \in \mathcal{J} } \nu  ~\mbox{,}
\eeq
denote the probability measure on $\Omega$. We note that $L > 1$, will allow us to e.g. take into account Schr\"odinger operators on a strip, see section \ref{sec:extensions_finiteRange_strip_Strip}.

Since every compact metric space is seperable, Dudley's result in \cite{Dudley_1966}, Theorem 12 therein, implies that, as before, weak-$^*$ convergence of the Borel probability measures $\mathcal{P}(\mathcal{K})$ on $\mathcal{K}$ is metrizable by the metric defined in (\ref{eq:metric}) (with $\mathrm{Lip}([-C,C])$ replaced by $\mathrm{Lip}(\mathcal{K})$).

In order to apply the framework developed in this paper {\em{to prove the qualitative and quantitative continuity of the DOSm}} in the underlying probability distribution, we require the model to satisfy each of the following three properties:
\begin{itemize}
\item[(P1)] {\bf{Basic spectral assumptions:}} There exists $r >0$ such that the spectrum of $H_\omega$ satisfies,
\begin{equation}
\sigma(H_\omega) \subseteq [-r,r] ~\mbox{, for all $\omega \in \Omega$ .}
\end{equation}
Moreover, the density of states measure (DOSm) $n_\nu^{(\infty)}$ can be defined as a spectral average of the form
\begin{equation} \label{eq:DOSm_general}
n_\nu^{(\infty)}(f) := \frac{1}{N} \E_{\nu^{(\infty)}} \{ {\rm Tr} (P_0 f(H_\omega) P_0) \} ~\mbox{, for } f \in Lip([-r,r]) ~\mbox{,}
\end{equation}
where $P_0$ is a finite-rank, orthogonal projection on $\ell^2(\mathbb{G}; \mathbb{C}^L)$ and $N=\mathrm{rk} P_0$.
\vspace{.1in}
\noindent

\item[(P2)] {\bf{Finite-range structure:}} For each $n \in \mathbb{N}$, the map
\begin{equation}
    \omega = (\omega_j) \mapsto {\rm Tr} (P_0 (H_\omega)^n P_0)
\end{equation}
 depends on only {\em{finitely}} many variables $\omega_j$ whose number can be bounded above by some strictly increasing {\em{counting function}} $\Gamma: \mathbb{N} \to \mathbb{N}$.
\vspace{.1in}
\noindent

\item[(P3)] {\bf{Lipschitz property:}} Given $j \in \mathcal{J}$, we write $\omega = (\omega_j, \omega_{\neq j})$. We then require that for each $j \in \mathcal{J}$, every {\em{fixed}} $\omega_{\neq j}$, and all $f \in Lip([-r,r])$, the function
\begin{equation}
\mathcal{K} \ni \lambda \mapsto \mathrm{Tr} (P_0 f(H_{(\lambda, \omega_{\neq j})} P_0) ~\mbox{,}
\end{equation}
is Lipschitz such that for all $\lambda, \lambda_0 \in \mathcal{K}$, one has
\begin{equation}
\left\vert \mathrm{Tr} (P_0 f(H_{(\lambda, \omega_{\neq j})} P_0) - {\rm Tr} (P_0 f(H_{(\lambda_0, \omega_{\neq j})} P_0) \right\vert \leq \gamma \Vert f \Vert_{\mathrm{Lip}} ~\rho(\lambda, \lambda_0) ~\mbox{.}
\end{equation}
Here, $\gamma \in \mathbb{R}$ is a constant, possibly depending on $N=\mathrm{rk} P_0$, but {\em{independent}} of both $f$ and $\omega$.
\end{itemize}

As pointed out in sections \ref{sec:qual-contDOSm1}-\ref{sec:quan-contDOSm1} (see e.g. the remarks following the statement of Theorem \ref{thm:qual1}), continuity of the DOSm in the probability distribution immediately implies the respective results for the IDS provided the latter depends continuously on energy. Hence, to {\em{prove the continuity statements of Theorem \ref{thm:main} for the IDS}}, in addition to the properties (P1)-(P3) above, the model also needs to satisfy that:
\begin{itemize}
\item[(P4)] for each $\nu \in \mathcal{P}(\mathcal{K})$, the IDS $N_\nu(E):= n_\nu^{(\infty)}((-\infty, E))$ is a continuous function in $E$ with a known modulus of continuity.
\end{itemize}

Before turning to specific examples of models which satisfy these properties, let us briefly identify the role of (P1) -- (P3) for the model in [H1], considered so far. Property (P1) is obviously satisfied by the very set-up described in [H1]. The finite-range structure (P2) for the model in [H1] is the subject of Lemma \ref{lem:rv-counting1}. We mention that by the same argument than presented in the proof of Lemma \ref{lem:continuity1}, every model possessing (P1) and (P2) automatically satisfies that ${\rm Tr} \{ P_0 f(H_\omega) P_0 \} \in C(\Omega)$  for every $f \in C([-r,r])$. The Lipschitz property for [H1] was verified in Proposition \ref{prop:finite-rank-prop2}, which in turn implies the finite-rank Lemma (Lemma \ref{sec:finite-rank-lemma1}) by the remarks preceding (\ref{eq:rankN1}).

Finally, we observe that while the proof strategy presented in the sections \ref{sec:qual-contDOSm1}--\ref{sec:quan-contDOSm1} applies to every model satisfying the properties (P1)-(P3), the specific modulus of continuity of the DOSm in the probability distribution will be determined by the counting function $\Gamma(n)$ in (P2). Indeed, as can be seen from (\ref{eq:weak-conv-dos2}), the modulus of continuity of the DOSm is goverened by the relative competition between the first and the second term on the right hand side of (\ref{eq:weak-conv-dos2}).

\subsection{The finite-range Anderson model and random Schr\"odinger operators on the strip} \label{sec:extensions_finiteRange_strip}
The two models in this section present minor modifications of the model in [H1]. In particular, the counting function $\Gamma(n)$ in both cases is the same as in Lemma \ref{lem:rv-counting1}. By the remarks made at the end of section \ref{sec:extensions_necfeat}, this implies that the {\em{validity of Theorem \ref{thm:main} extends to both random Schr\"odinger operators on the strip as well as to the finite-range Anderson model}}.

\subsubsection{The finite-range Anderson model} \label{sec:extensions_finiteRange_strip_FiniteRange}

The finite-range Anderson model, sometimes also known as the ``generalized Anderson model,'' see e.g. \cite{DamanikSimsStolz_JFunAnal_2004, Bucaj_KunzSoulliardprep2016}, is obtained by modifying the Hamiltonian in (\ref{eq:schr-op2}) so to allow for a {\em{non-uniform potential profile}} in the unit $\Lambda_0$ (and hence in each of its translates).

To introduce the model formally, for $j \in \mathcal{J}= K \mathbb{Z}^d$, we denote by $\tau_j$ the translation on $\mathbb{Z}^d$ from the origin to site $j$,
\begin{equation}
\tau_j(k):=k + j ~\mbox{, } k \in \mathbb{Z}^d ~\mbox{,}
\end{equation}
and let $U_j$ be the associated induced unitary on $\ell^2(\mathbb{Z}^d)$ as defined in [H1]. Moreover, we write $\pi_k := \vert \delta_k \rangle \langle \delta_k \vert$ for the projection onto the standard basis vector $\delta_k$ of $\ell^2(\mathbb{Z}^d)$. Fixing $K \in \mathbb{N}$, we then define the {\em{Anderson model with finite-range potential}} as
\begin{description}
\item [{[$H1^\prime$]}] the discrete Hamiltonian on $\ell^2(\mathbb{Z}^d)$ of the form
\begin{equation} \label{eq:genAndmodel}
H_\omega =  \Delta + \sum_{j \in {\mathcal{J}}} \omega_j \left( \sum_{j^\prime \in \Lambda_j} \Theta(\tau_j^{-1}(j^\prime)) \pi_{j^\prime}  \right) ~\mbox{,}
\end{equation}
where $\Theta: [0, K-1]^d \cap \mathbb{Z}^d \to \mathbb{R}$ is a given function, subsequently referred to as ``{\em{profile function}},'' which models the potential profile in each of the lattice units
\begin{equation}
\Lambda_j: = \tau_j( [0,K-1]^d \cap \mathbb{Z}^d) ~\mbox{.}
\end{equation}
\end{description}
Observe that taking $\Theta \equiv 1$ reduces [$H1^\prime$] to [H1].

The model in [$H1^\prime$] clearly satisfies the properties (P1) and (P2) of section \ref{sec:extensions_necfeat}; indeed, the counting lemma, Lemma \ref{lem:rv-counting1}, remains unaffected when replacing [H1] by [$H1^\prime$].

Property (P3) is verified by fairly obvious modifications of the proof of Proposition \ref{prop:finite-rank-prop2}. To this end, fixing $l \in \mathcal{J}$, we first modify the operator in (\ref{eq:finiteranklemmasetup}) according to
\begin{equation} \label{eq:finiteranklemmasetup_modif}
\widetilde{H_\lambda}^{(\ell)}:= H^{(0)} + \lambda \widetilde{P_\ell} ~\mbox{, } \lambda \in [-C,C] ~\mbox{,}
\end{equation}
where
\begin{equation}
\widetilde{P_\ell} :=\sum_{j^\prime \in \Lambda_j} \Theta(\tau_j^{-1}(j^\prime)) \pi_{j^\prime} ~\mbox{.}
\end{equation}
Further, as in (\ref{eq:inclusion1}), suppose $[a,b]$ is a closed interval such that
\beq\label{eq:inclusion1_genand}
\bigcup_{\lambda \in [- C, C]} ~\sigma (\widetilde{H_\lambda}^{(\ell)} ) \subseteq [a,b].
\eeq

Then, obvious modifications of the proof of Proposition \ref{prop:finite-rank-prop2} imply:
\begin{prop} \label{lem:finite-rank1_modif}
Given the set-up described in (\ref{eq:finiteranklemmasetup_modif})-(\ref{eq:inclusion1_genand}). For all $f \in Lip([a,b])$ with Lipschitz constant $L_f$ and all $\lambda, \lambda_0 \in [-C,C]$, one has that
\bea\label{eq:rank-N2_modif}
\left| {\rm Tr} ~ \left( P_0 f(\widetilde{H_\lambda}^{(\ell)}) P_0 \right) - {\rm Tr} ~ \left( P_0 f(\widetilde{H_{\lambda_0}}^{(\ell)}) P_0 \right) \right| & \leq & 2 N^2 \Vert \Theta \Vert_\infty L_f ~\mbox{.}
\eea
\end{prop}
%

\subsubsection{Random Schr\"odinger operators on the strip} \label{sec:extensions_finiteRange_strip_Strip}

Let $\emptyset \neq \mathcal{K}$ be a fixed {\em{compact}} subset of the $L \times L$ symmetric matrices over $\mathbb{R}$, equipped with a metric $\rho$ derived from any fixed matrix norm. We consider Schr\"odinger operators on $\mathbb{Z}$ with {\em{matrix-valued}} potentials, randomly sampled from $\mathcal{K}$, i.e. for $\omega = (\omega_n) \in \mathcal{K}^\mathbb{Z}$, we define a bounded self-adjoint operator $H_\omega$ on $\ell^2(\mathbb{Z}; \mathbb{C}^L)$ by
\begin{equation} \label{eq:SchrodStrip}
(H_\omega \psi)_n = \psi_{n-1} + \psi_{n+1} + \omega_n \psi_n ~\mbox{,}
\end{equation}
where the elements of $\omega = (\omega_n) \in \Omega$ are iid random variables distributed according to a common single-site measure $\nu \in \mathcal{P}(\mathcal{K})$. The DOSm is then defined as a spectral average of the form (\ref{eq:DOSm_general}) where $N=L$ and
\begin{equation}
P_0 = \sum_{k=1}^L \vert \delta_0 \otimes e_k \rangle \langle \delta_0 \otimes e_k \vert ~\mbox{.}
\end{equation}
Here, $e_k$, $1 \leq k \leq L$, is the standard basis of $\mathbb{C}^L$ and $\delta_0 = (\delta_{0,j})_{j \in \mathbb{Z}} \in \ell^2(\mathbb{Z})$.

Clearly this model satisfies all the properties (P1) - (P4) of section \ref{sec:extensions_necfeat}. Indeed, aside from obvious changes replacing the distance in $\mathbb{R}$ by the metric $\rho$ on $\mathcal{K}$, the proofs presented in the sections \ref{sec:qual-contDOSm1}--\ref{sec:finite-rank-lemma1} remain unchanged.

We briefly comment on the consequences of Theorem \ref{thm:main} for the continuity of Lyapunov exponents in the probability distribution for random Schr\"odinger operators on the strip. This follows up on our earlier remarks to Theorem \ref{rem:qualilyap} (remark \ref{rem:qualilyap} part (ii)) and Theorem \ref{thm:LE-quant1} (remark \ref{rem:quantilyap}).

Given $E \in \mathbb{C}$ and $n \in \mathbb{N}$, let
\begin{equation} \label{eq:transfer}
A_n^E(\omega) := A^E(\omega_{n-1}) \cdot A^E(\omega_{n-2}) \cdot A^E(\omega_0) ~\mbox{,}
\end{equation}
with
\begin{equation}
A^E(\omega_n):= \begin{pmatrix} E ~\mathrm{I}_L - \omega_n & - \mathrm{I}_L \\ \mathrm{I}_L & 0_l \end{pmatrix} ~\mbox{,}
\end{equation}
where $\mathrm{I}_L$ and $0_L$ is the $L \times L$ identity and zero matrix, respectively. We note that for each $n \in
\mathbb{N}$, $A_n^E(\omega)$ is a conjugate symplectic $2L \times 2L$ matrix over $\mathbb{C}$. Consequently, let $\sigma_1(A_n^E(\omega)) \geq \sigma_2(A_n^E(\omega)) \geq \dots \sigma_L(A_n^E(\omega)) \geq 1$ be the $L$ singular values of $A_n^E(\omega)$ with magnitude greater than or equal to 1. Then, by the Osceledec-Ruelle theorem (see e.g. \cite{simon_kotani_CMP1988}, section 5), the limits
\begin{equation}
0 \leq \gamma_j(E):= \lim_{n \to \infty} \dfrac{1}{n} \log \sigma_j(A_n^E(\omega)) ~\mbox{, } 1 \leq j \leq L ~\mbox{,}
\end{equation}
exist and are finite for $\nu^{(\infty)}$-a.e. $\omega \in \Omega$. The non-random numbers $\gamma_j(E)$, $1 \leq j \leq L$ are called the non-negative Lyapunov exponents associated with the random Schr\"odinger operator in (\ref{eq:SchrodStrip}).

For Schr\"odinger operators on the strip the Thouless formula, which generalizes (\ref{eq:le1}) for $L \geq 1$, is then given by
\begin{equation} \label{eq:thouless_strip}
L(E) := \sum_{j=1}^L \gamma_j(E) =  \int_\R ~\log |E^\prime - E| ~dn_\nu^{(\infty)} (E') ~\mbox{.}
\end{equation}
A proof of (\ref{eq:thouless_strip}) can e.g. be found in the appendix of \cite{simon_kotani_CMP1988}.

As a consequence of (\ref{eq:thouless_strip}), all the results of section \ref{subsec:le1}, in particular the Theorems \ref{thm:qualiLE} and \ref{thm:LE-quant1}, carry over to Schr\"odinger operators on the strip as continuity statements for the sum of the non-negative Lyapunov exponents.

\subsection{Anderson model on the Bethe lattice} \label{sec:extensions_bethe}

The Bethe lattice $\mathbb{B}$ is an infinite regular graph with no loops (i.e. a tree) with coordination number $k \geq 3$ (number of nearest neighbors at each vertex)\footnote{The case $k=2$ corresponds to $\mathbb{Z}$, which was treated before.}. We denote by $\mathcal{V}_\mathbb{B}$ the vertices of $\mathbb{B}$. To define the Anderson model on the Bethe lattice, we consider the random operator on $\ell^2(\mathcal{V}_\mathbb{B})$ given by
\begin{eqnarray} \label{eq:hamiltonianbethe}
H_\omega & = & \Delta_\mathbb{B} + \sum_{x \in \mathcal{V}_\mathbb{B}} \omega_x \pi_x ~\mbox{,} \nonumber \\
(\Delta_\mathbb{B} \psi)(x) & := & \sum_{y \in \mathcal{V}_\mathbb{B}: y \sim x} \psi(y) ~\mbox{,}
\end{eqnarray}
where $y \sim x$ in the definition of the graph Laplacian $\Delta_\mathbb{B}$ denotes the $k$ nearest neighboring vertices $y$ of $x$ and $\pi_x = \vert \delta_x \rangle \langle \delta_x \vert$ is the orthogonal projection onto the standard basis vector $\delta_x \in \ell^2(\mathcal{V}_\mathbb{B})$ associated with the vertex $x$. The Anderson model on the Bethe lattice was first proposed in the physics literature by Abou-Chacra, Anderson, and Thouless \cite{Abou-ChacraThoulessAnderson_JPhysC_1973}, and has ever since enjoyed considerable attention in both the mathematics and the physics community; we refer e.g. to \cite{warzel_ICMP_2013} for a review of known results and a more detailed list of references.

Replacing (\ref{eq:schr-op2}) by (\ref{eq:hamiltonianbethe}) in [H1], we will refer to this modified set-up as [H1$\mathbb{B}$]. In view of property (P1), we note that (\ref{eq:schr-op-sp1}) in [H1] becomes
\begin{equation}
\sigma ( H_\omega ) \subseteq [-2 \sqrt{k-1} - C, 2 \sqrt{k-1} + C] =: [-r_\mathbb{B},r_\mathbb{B}] ~\mbox{,}
\end{equation}
for all $\omega \in \Omega = [-C,C]^{\mathcal{V}_\mathbb{B}}$. Finally upon (arbitrarily) distinguishing one vertex as a the root ``0'' of the tree, the DOSm associated with $\nu \in \mathcal{P}([-C,C])$ can be defined in complete analogy to (\ref{eq:DOSm1}) as
\begin{equation} \label{eq:DOSm1_bethe}
n_\nu^{(\infty)}(f) := \E_{\nu^{(\infty)}} \{ {\rm Tr} (\pi_0 f(H_\omega) \pi_0) \} ~\mbox{.}
\end{equation}

It is obvious that [H1$\mathbb{B}$] satisfies the properties (P1)-(P3) of section \ref{sec:extensions_necfeat}. In particular, similar arguments than in the proof of the counting lemma, Lemma \ref{lem:rv-counting1}, imply that for each fixed $n \in \mathbb{N}$, the function $\omega \mapsto {\rm Tr} (\pi_0 (H_\omega)^n \pi_0)$ can depend on at most the number of vertices in the $\lfloor \frac{n}{2} \rfloor$-th generation of the tree:\begin{align} \label{eq:gammabethe}
1 + k \sum_{j=0}^{\lfloor \frac{n}{2} \rfloor - 1} (k-1)^{j} = 1 + k \dfrac{(k-1)^{\lfloor \frac{n}{2} \rfloor} - 1}{k - 2} \leq 3 k^{\frac{n}{2}} =: \Gamma_\mathbb{B}(n) ~\mbox{.}
\end{align}
Property (P2) therefore holds. Finally, the Lipschitz property (P3) holds by the same of proof than presented in Proposition \ref{prop:finite-rank-prop2} with $N=1$.

In summary, straight-forward modifications of the arguments presented in section \ref{sec:quan-contDOSm1} imply the following quantitative continuity for the DOSm:
\begin{theorem} \label{thm:quant_bethe}
Given the set-up described in [H1$\mathbb{B}$] and [H2], there exists $\alpha_\mathbb{B} \in \mathbb{N}$ such that for every $f \in Lip ([-r_\mathbb{B},r_\mathbb{B}])$ and $\alpha \geq \alpha_\mathbb{B}$, one has
\beq\label{eq:quant-cont1-bethe}
| n_\alpha^{(\infty)} (f) - n^{(\infty)} (f) | \leq \dfrac{\gamma_{\mathbb{B}} ~\Vert f \Vert_{\mathrm{Lip}} }{\sqrt{ \log\left(\eta_\alpha^{- \frac{2 \mathrm{e}}{1 + 2 \mathrm{e}}   }  \right)    }} ,
\eeq
The constants $\alpha_\mathbb{B}, \gamma_{\mathbb{B}}$ are independent of $f$ and are determined by (\ref{eq:alpha0_bethe}) and (\ref{eq:holderconst_bethe}), respectively.
\end{theorem}
\begin{proof}
We let $f \in Lip([-r_\mathbb{B},r_\mathbb{B}])$ and follow the general outline of the proof of Theorem \ref{thm:quant1}. In particular, letting
\begin{equation} \label{eq:xi0_bethe}
\xi_0 :=  \frac{2 \mathrm{e}}{1 + 2 \mathrm{e}} ~\mbox{,}
\end{equation}
and taking $\xi \geq \xi_0$, to be determined later, we obtain in complete analogy to (\ref{eq:weak-conv-dos2}) that for all $\alpha \in \mathbb{N}$:
 \bea\label{eq:weak-conv-dos2_bethe}
   | n_\alpha^{(\infty)} (f) - n^{(\infty)} (f) | & \leq & \frac{2 b_f \sqrt{\log(k)}}{ \sqrt{ 2 \log(\eta_\alpha^{-\xi} /3 ) - \log(k) } }  + c_f \eta_\alpha^{1 - \xi}  ~\mbox{.}
\eea
Here, we used the explicit form of $\Gamma_\mathbb{B}(n)$ given in (\ref{eq:gammabethe}).

Choosing $\alpha_\mathbb{B} \in \mathbb{N}$ such that
\begin{equation} \label{eq:alpha0_bethe}
\eta_{\alpha_\mathbb{B}} \leq \mathrm{e}^{-3 \log(k) / \xi_0} ~\mbox{, for all $\alpha \geq \alpha_\mathbb{B}$ ,}
\end{equation}
with $\xi_0$ as in (\ref{eq:xi0_bethe}), $\xi \geq \xi_0$ and $k \geq 3$ imply that for all $\alpha \geq \alpha_\mathbb{B}$, one has
\begin{equation}
2 \log( \eta_\alpha^{-\xi} /3 ) - \log(k) \geq \log(\eta_\alpha^{-\xi}) ~\mbox{,}
\end{equation}
whence (\ref{eq:weak-conv-dos2_bethe}) yields
\begin{equation}  \label{eq:weak-conv-dos3_bethe}
 | n_\alpha^{(\infty)} (f) - n^{(\infty)} (f) | \leq \frac{2 b_f \sqrt{\log(k)}   }{ \sqrt{\log(\eta_\alpha^{-\xi} )  } } + c_f (\eta_\alpha^{-\xi})^{\frac{1 - \xi}{\xi}}  \mbox{.}
 \end{equation}
To optimize (\ref{eq:weak-conv-dos3_bethe}) in $\xi$ we use, as in the proof of Theorem  \ref{thm:IDS-quant1}, that $y^\beta \geq \log(y)$, for all $y \in (0, + \infty)$, if and only if $\beta \geq \frac{1}{\mathrm{e}}$. We therefore conclude that (\ref{eq:weak-conv-dos3_bethe}) is optimized by taking $\xi = \xi_0$.

In summary, we conclude that for all $\alpha \geq \alpha_\mathbb{B}$, one has
\begin{equation}
| n_\alpha^{(\infty)} (f) - n^{(\infty)} (f) | \leq \dfrac{2( 2 r c_b \sqrt{\log(k)} + 1 )}{ \sqrt{ \log\left( \eta_\alpha^{\xi_0} \right)   } } ~\mbox{,}
\end{equation}
which determines the constant $\gamma_\mathbb{B}$ in (\ref{eq:quant-cont1-bethe}) as
\begin{equation} \label{eq:holderconst_bethe}
\gamma_\mathbb{B} = 2( 2 r_\mathbb{B} c_b \sqrt{\log(k)} + 1 ) ~\mbox{,}
\end{equation}
where $c_b$ is given in (\ref{eq:bernsteinconst}).
\end{proof}

We cannot give a Bethe lattice analogue of our continuity result for the IDS (Theorem \ref{thm:IDS-quant1}), as it is not known whether for a {\em{general}} probability measure $\nu \in \mathcal{P}([-C,C])$ the IDS for the Anderson model on the Bethe lattice is continuous; we refer the reader to section \ref{subsec:previous-intro1} for a short review of some known results. We can however address the weak-disorder limit $\lambda \to 0^+$ for {\em{both}} the DOSm and the IDS of
\beq\label{eq:schr-op-low1-bethe}
H_\omega (\lambda) := \Delta_{\mathbb{B}} + \lambda \sum_{j \in {\mathcal{J}}} \omega_j P_j ~\mbox{, }
\eeq
with $\omega$ distributed according to an arbitrary fixed single-site measure $\mu \in \mathcal{P}([-1,1])$. Referring to the above-mentioned open problem of whether the IDS for the general Anderson model on the Bethe lattice is continuous, we note that for $\lambda = 0$ in (\ref{eq:schr-op-low1-bethe}), the DOSm is known to be absolutely continuous $\mathrm{d}n_{\lambda = 0}^{(\infty)}(E)= \rho_{\lambda = 0}^{(\mathbb{B})}(E) \mathrm{d}E$ with DOSf given by (see e.g. \cite{aizenmanWarzel_book}),
\begin{equation}
\rho_{\lambda = 0}^{(\mathbb{B})}(E) = \dfrac{k}{2 \pi} \dfrac{\sqrt{4(k-1) - E^2}}{k^2 - E^2} \chi_{[-2 \sqrt{k}, 2 \sqrt{k} ]}(E)  ~\mbox{.}
\end{equation}
In particular, this implies that for all $E \in \mathbb{R}$ and $\epsilon > 0$:
\begin{equation} \label{eq:IDS_BetheLapl}
N_{\lambda = 0}(E + \epsilon) - N_{\lambda = 0}(E) = n_{\lambda = 0}([E, E + \epsilon]) \leq c_{\mathbb{B}} \epsilon ~\mbox{, }
\end{equation}
where one can take
\begin{equation} \label{eq:IDS_BetheLapl_1}
c_{\mathbb{B}} =\Vert \rho_{\lambda = 0}^{(\mathbb{B})} \Vert_\infty = \begin{cases}  \dfrac{k}{4 \pi \sqrt{k^2 - 4(k-1)}} & ~\mbox{, if } k \in [3,6] \cap \mathbb{N} ~\mbox{,} \\ \dfrac{\sqrt{4 (k-1)}}{k}  & ~\mbox{, if } k \geq 7  ~\mbox{.} \end{cases}
\end{equation}

Application of Theorem \ref{thm:quant_bethe} therefore immediately yields:
\begin{theorem} \label{thm:contbethe_weak}
Consider the Anderson model on the Bethe lattice ($k \geq 3$) given in (\ref{eq:schr-op-low1-bethe}) with underlying single-site measure $\mu \in \mathcal{P}([-1,1])$. There exist constants $\lambda_\mathbb{B} , \widetilde{\gamma_\mathbb{B}}, \widetilde{c}_\mathbb{B}$ such that for all $0 < \lambda \leq \lambda_\mathbb{B}$:
\begin{itemize}
\item[(i)] for every $f \in Lip ([-2 \sqrt{k-1} - \lambda_0 , 2 \sqrt{k-1} + \lambda_0])$, one has
\begin{equation}
| n_\lambda^{(\infty)} (f) - n_{\lambda = 0}^{(\infty)} (f) | \leq \dfrac{\widetilde{\gamma_{\mathbb{B}}} ~\Vert f \Vert_{\mathrm{Lip}}}{\sqrt{ \log\left(\lambda^{- \frac{2 \mathrm{e}}{1 + 2 \mathrm{e}}   }  \right) }} ,
\end{equation}
\item[(ii)] for all energies $E \in [-2 \sqrt{k-1} - \lambda_0 , 2 \sqrt{k-1} + \lambda_0]$, one has
\begin{equation} \label{eq:ids_bethe_weak}
\vert N_\lambda(E) - N_{\lambda =0 }(E) \vert \leq \dfrac{\widetilde{c}_\mathbb{B}}{[ \log(\lambda^{- \frac{2 \mathrm{e}}{1 + 2 \mathrm{e}}   }) ]^{1/4}} ~\mbox{.}
\end{equation}
\end{itemize}
Explicitly, $\lambda_\mathbb{B} = \mathrm{e}^{- 3 \log(k)/ \xi_0}$, where $\xi_0$ is given in (\ref{eq:xi0_bethe}), and the constants $\widetilde{\gamma_\mathbb{B}}, \widetilde{c}_\mathbb{B}$ are determined in, respectively, (\ref{eq:hoelderconst_dosm_bethe_weak}) and (\ref{eq:hoelderconst_ids_bethe_weak}).
\end{theorem}
We note that, as earlier (see e.g. our comments right after Theorem \ref{thm:weak-contIDS1}), the restriction of the energy in part (ii) of Theorem \ref{thm:contbethe_weak} could be dropped.

As mentioned in section \ref{subsec:previous-intro1}, Klein and Sadel had previously analyzed the regularity of the IDS in $\lambda$ for energies {\em{inside the spectrum of the free Laplacian}}, $[2 \sqrt{k-1}, 2\sqrt{k-1}]$. Specifically they showed in \cite{KleinSadel} that for every closed interval $I \subset (-2 \sqrt{K-1}, 2 \sqrt{K-1})$, there exists $\delta > 0$ such that $(E,\lambda) \mapsto N_\lambda(E)$ is jointly $\mathcal{C}^1$ for $(E, \lambda) \in I \times (-\delta, \delta)$. While Theorem \ref{thm:contbethe_weak} only obtains $\log$-H\"older continuity of the IDS in the disorder, our result is not restricted to compact subsets of the spectrum of the free Laplacian.

\begin{proof}
Part (i) follows by a straight-forward application of Theorem \ref{thm:quant_bethe} for $\eta_\alpha = \lambda$; here, from (\ref{eq:holderconst_bethe}), one finds
\begin{equation} \label{eq:hoelderconst_dosm_bethe_weak}
\widetilde{\gamma_\mathbb{B}} = 2( 2 (2 \sqrt{k} + \lambda_0) c_b \sqrt{\log(k)} + 1 ) ~\mbox{.}
\end{equation}

For part (ii), analogous arguments than the ones used in the proof of Theorem \ref{thm:IDS-quant1} (see in particular, (\ref{eq:ids-cont2})) yields by Theorem \ref{thm:quant_bethe} and (\ref{eq:IDS_BetheLapl}) that
\begin{equation}
\vert N_\lambda(E) - N_{\lambda =0}(E) \vert \leq \dfrac{\widetilde{\gamma_\mathbb{B}}}{\sqrt{ \log (\lambda^{-\xi_0})    } } \cdot \dfrac{1}{\min \{ a_- ; a_+ \}} + c_\mathbb{B} (a_- + a_+) ~\mbox{,}
\end{equation}
for arbitrary $a_-, a_+ > 0$ and all $0 < \lambda \leq \lambda_0$.
Hence, letting $a_- = a_+ = \frac{1}{2} ( \log (\lambda^{-\xi_0}) )^{-1/4}$, we conclude (\ref{eq:ids_bethe_weak}) with
\begin{equation} \label{eq:hoelderconst_ids_bethe_weak}
\widetilde{c}_\mathbb{B} = \widetilde{\gamma_\mathbb{B}} + c_\mathbb{B} ~\mbox{,}
\end{equation}
where the constant $c_\mathbb{B}$ is given in (\ref{eq:IDS_BetheLapl_1}). We note that the exponent of $1/4$ in (\ref{eq:ids_bethe_weak}) is optimized for our proof.
\end{proof}


\section{Appendix A: Proof of the counting Lemma (Lemma \ref{lem:rv-counting1})}\label{sec:appendix:rv-counting1}
\setcounter{equation}{0}

We prove Lemma \ref{lem:rv-counting1} in this section. We assume [H1].

\begin{proof}[Proof of Lemma \ref{lem:rv-counting1}]
Recall that by the set-up described in [H1], the lattice $\mathbb{Z}^d$ is partitioned into cubes consisting of $N=K^d$ sites. We begin by expanding the trace:
\bea\label{eq:count-exp1}
{\rm Tr} ~( P_0 H_\omega^n P_0 ) & = & \sum_{(k_1, \ldots, k_n ) \in \{ 0,1\}^n} {\rm Tr} ~(P_0 H_{k_n} \cdots H_{k_1} ) \nonumber \\
 &=&  \sum_{j \in \Z^d \cap [0,K-1]^d} \sum_{(k_1, \ldots, k_n ) \in \{ 0,1\}^n} \langle \delta_j, H_{k_n} \cdots H_{k_1} \delta_j \rangle .
 \eea
Because of the form of the discrete Laplacian $\Delta$, each non-vanishing summand in \eqref{eq:count-exp1} represents a walk on $\Z^d$ starting and ending at the {\em{same}} point in $[0,K-1]^d$. In particular, there cannot be more than $\lfloor n/2 \rfloor$ steps occurring in the same direction and parallel to one of the coordinate axes. This shows that the walks representing the non-vanishing terms in \eqref{eq:count-exp1} are confined within the cube
$$
\left[ - (K-1) -  \left\lfloor \frac{n}{2} \right\rfloor, (K-1) + \left\lfloor \frac{n}{2} \right\rfloor \right]^d .
$$
Using the fact that $H_\omega = H^{(0)} + V_\omega$, the function $\omega \mapsto {\rm Tr}~(P_0 H_\omega^n P_0 )$ thus depends on at most
\beq \label{eq:gammaaccurate}
\left( 2 \left\lceil \frac{1}{K-1} \left( K-1 + \left\lfloor \frac{n}{2} \right\rfloor \right) \right\rceil \right)^d \leq 2^d n^d
\eeq
random variables, which proves the lemma.
\end{proof}


\section{Appendix B: Alternate proof of the finite-rank lemma (Proposition \ref{prop:finite-rank-prop2})}\label{sec:appendix:alt-fr1}
\setcounter{equation}{0}

In this section, we provide an alternate proof of Proposition \ref{prop:finite-rank-prop2} using the Helffer-Sj\"ostrand functional calculus. This proof has the advantage of being slightly shorter than the proof presented in section \ref{sec:finite-rank-lemma1}. We briefly review the construction of the functional calculus and refer to \cite[section 2.2]{davies1} for a detailed discussion.

Suppose that $f \in C^\infty_0 (\R; \R)$. An {\em{almost-analytic extension of $f$ of order $n$}} is constructed as follows. Let $J \in C_c^\infty (\R; \R)$ be a real-valued function with $J |_{[-1,1]} = 1$ and ${\rm supp} ~J \subset [-2,2]$. We set $z = x + i y$, for $x,y \in \R$, and define $\tilde{f}$, for $n \in \N$, by
\beq\label{eq:aa-ext1}
\tilde{f}(z) := \sum_{k=0}^n \frac{1}{k!} f^{(k)}(x) (iy)^k J\left( \frac{y}{\langle x \rangle} \right),
\eeq
where $\langle x \rangle = (1+ \|x\|^2)^{\frac{1}{2}}$.
A straight-forward computation based on (\ref{eq:aa-ext1}) shows that every almost analytic extension $\tilde{f}$ of $f$ of order $n$ has the following properties:
\begin{description}
\item[{i.}] The derivatives of $f$ and $\tilde{f}$ on the real line agree:
$$
\frac{d^k \tilde{f}}{dx^k} (x) = \frac{d^k f}{dx^k}(x), ~~~\forall k \in \{0\} \cup \N, ~x \in \R.
$$

\item[{ii.}] The almost analytic property: Let $\frac{\partial}{\partial \overline{z}} := \frac{1}{2} \left( \frac{\partial }{\partial x} + i \frac{\partial }{\partial y} \right)$, then
the $\overline{z}$-derivative vanishes rapidly in a neighborhood of the real axis,
$$
\left| \frac{\partial \tilde{f}}{\partial \overline{z}} (z) \right| = \mathcal{O} (|y|^n| ), ~~~ y \rightarrow 0 .
$$
\end{description}

The almost analytic property $({\rm ii}.)$ makes possible the following definition. For every self-adjoint operator $H$ and $f \in C_c^\infty (\R; \R)$, the operator $f(H)$ may be represented in the form:
\beq\label{eq:hsj1}
 f(H) = \frac{1}{\pi} \int_{\C} \frac{\partial \tilde{f}}{\partial \overline{z}}(z) (H - z)^{-1} ~dx ~dy .
 \eeq
The integral on the right in \eqref{eq:hsj1} converges in the operator norm and defines a bounded, self-adjoint operator which is independent of the cut-off function $J$ and the order $n$ of $\tilde{f}$, see \cite[Lemma 2.2.4]{davies1}.

We also recall that for every compactly supported function $g \in C_c^1(\R^2)$, it is a basic consequence of the Cauchy-Pompeiu formula that for each $z_0 \in \C$, one has the representation
\beq\label{eq:pow1}
g(z_0) = \frac{1}{\pi} \int_{\C} \frac{\partial g}{\partial \overline{z}}(z) (z_0 - z)^{-1} ~dx ~dy
\eeq

\begin{proof}[Alternate proof of Proposition \ref{prop:finite-rank-prop2}]
1. By the same argument than at the beginning of the proof of Proposition \ref{prop:finite-rank-prop2}, it suffices to consider $f \in C_c^\infty (\R; \R)$. Let $\tilde{f}$ be an almost analytic extension of $f$ of order $n=2$ as in \eqref{eq:aa-ext1}. Then, for $\lambda \in [-C, C]$, application of the Helffer-Sj\"ostrand formula \eqref{eq:hsj1} yields
\beq\label{eq:hsj2}
{\rm Tr} (P_0 f(H_\lambda^{(\ell)}) P_0) = \frac{1}{\pi} \int_{\C} \frac{\partial \tilde{f}}{\partial \overline{z}}(z) {\rm Tr} ( P_0 ( H - z)^{-1} P_0 ) ~dx ~dy.
\eeq
By the second resolvent formula, for $z \in \C \backslash \R$, we have
\beq\label{eq:hsj3}
\frac{d}{d \lambda} {\rm Tr} (P_0 (H_\lambda^{(\ell)}-z)^{-1} P_0) = - {\rm Tr} ( P_0 ( H_\lambda^{(\ell)} - z)^{-1} P_\ell (H_\lambda^{(\ell)} - z)^{-1} P_0 ).
\eeq
Furthermore, for $z \in \C \backslash \R$, the fact that ${\rm Tr} P_0 = N$ and a standard resolvent estimate give
\beq\label{eq:hsj4}
| {\rm Tr} ( P_0 ( H_\lambda^{(\ell)} - z)^{-1} P_\ell (H_\lambda^{(\ell)} - z)^{-1} P_0 ) |
\leq \frac{N^2}{| \Im z|^2 } ~.
\eeq
It follows that if we choose an almost analytic extension of $f$ with order $n=2$, then the almost analytic property $({\rm ii}.)$ guarantees that
\beq\label{eq:hsj5}
 \left| \frac{\partial \tilde{f}}{\partial \overline{z}}(z) ~ {\rm Tr} ( P_0 ( H_\lambda^{(\ell)} - z)^{-1} P_\ell (H_\lambda^{(\ell)} - z)^{-1} P_0 ) \right| = \mathcal{O}(1) ~,
 \eeq
locally around $y = 0$.
Thus, differentiating under the integral sign yields
\beq\label{eq:hsj6}
\frac{d}{d \lambda} {\rm Tr} (P_0 f(H_\lambda^{(\ell)}) P_0) = - \frac{1}{\pi} \int_{\C} \frac{\partial \tilde{f}}{\partial \overline{z}}(z) ~ {\rm Tr} ( P_0 ( H_\lambda^{(\ell)} - z)^{-1} P_\ell (H_\lambda^{(\ell)} - z)^{-1} P_0 ) ~dx ~dy .
\eeq

\noindent
2. We recall the complex measure defined in \eqref{eq:measure2} and write, for $z \in \C \backslash \R$,
\beq\label{eq:hsj7}
 {\rm Tr} ( P_0 ( H_\lambda^{(\ell)} - z)^{-1} P_\ell (H_\lambda^{(\ell)} - z)^{-1} P_0 ) =
 \int_{\R \times \R} ~\frac{1}{s-z} ~\frac{1}{t-z} ~d \mu_{\lambda;\lambda} (s,t) .
 \eeq
For simplicity, we set $\mu := \mu_{\lambda;\lambda}$. By Fubini's Theorem, we then can re-write the right side of \eqref{eq:hsj6} as
\bea\label{eq:hsj8}
\lefteqn{ \frac{1}{\pi} \int_{\R \times \R} ~d \mu (s,t) \left\{ \int_{\C} \frac{\partial \tilde{f}}{\partial \overline{z}}(z) ~ \frac{1}{s-z} ~\frac{1}{t-z} ~dx ~dy \right\} } \nonumber \\
 &=& \int_{\R \times \R} ~d \mu (s,t) \chi_{\{s \neq t\}} \frac{1}{t-s} ~\frac{1}{\pi} \int_{\C} \frac{\partial \tilde{f}}{\partial \overline{z}}(z) \left( \frac{1}{s-z} - \frac{1}{t-z} \right) ~dx ~dy  \nonumber \\
  & & + \int_{\R \times \R} ~d \mu (s,t) \chi_{\{s = t\}}  ~\frac{1}{\pi} \int_{\C} \frac{\partial \tilde{f}}{\partial \overline{z}}(z) \frac{1}{(s-z)^2} ~dx ~dy.
\eea
We mention that the use of Fubini's Theorem is justified by \eqref{eq:hsj5} and the fact that $\tilde{f}$  has compact support in $\C$.
Applying \eqref{eq:pow1} to \eqref{eq:hsj8}, we obtain
\beq\label{eq:hsj9}
\frac{1}{\pi} \int_{\R \times \R} ~d \mu (s,t) \chi_{\{s \neq t\}} \frac{1}{t-s} (f(s) - f(t)) - \frac{1}{\pi} \int_{\R \times \R} ~d \mu (s,t) \chi_{\{s = t\}} f^\prime (s) .
\eeq

\noindent
3. In summary, combining \eqref{eq:hsj6} and  \eqref{eq:hsj9}, we conclude that
\beq\label{eq:hsj10}
\left| \frac{d}{d \lambda} {\rm Tr} (P_0 f(H_\lambda^{(\ell)}) P_0) \right|
\leq \int_{\R \times \R} ~ d | \mu| (s,t) ~| f^\prime( \xi_{s,t})| ,
\eeq
where $\xi_{s,t}$ satisfies $s \leq \xi_{s,t} \leq t$. Consequently, using estimate \eqref{eq:measure-total-var2}
in \eqref{eq:hsj10}, we obtain
\beq\label{eq:hsj11}
\left| \frac{d}{d \lambda} {\rm Tr} (P_0 f(H_\lambda^{(\ell)}) P_0) \right|
\leq 2 N^2 \| f^\prime \|_\infty.
\eeq
Finally, we observe that from \eqref{eq:hsj5} and the Dominated convergence Theorem, the right side of \eqref{eq:hsj6} is continuous in $\lambda$. In particular, the map
\beq\label{eq:hsj-cont1}
\lambda \rightarrow {\rm Tr} (P_0 f(H_\lambda^{(\ell)}) P_0)
 \eeq
is $C^1$ in $\lambda$. Thus the bounds \eqref{eq:hsj9} and \eqref{eq:hsj10} imply \eqref{eq:rank-N5}, which concludes our alternate proof of Proposition \ref{prop:finite-rank-prop2}.
\end{proof}

\begin{remark}
Since we only need an almost analytic extension of order $n=2$ in the above proof, we can relax the condition on $f$ to $f \in C^2_c(\R)$. For comparison, the proof presented in section \ref{sec:finite-rank-lemma1} requires slightly more, namely, that $\widehat{(f^\prime)} \in L^1 (\R)$, see \eqref{eq:fourier2}.
\end{remark}


\section{Appendix C: Nontangential limits of the Lyapunov exponent}\label{sec:appendix:nontangLE1}
\setcounter{equation}{0}

In this appendix, we explore the possibility of proving Proposition \ref{prop:le-quant1} under the condition that the DOSm is merely log-H\"older continuous,
\beq\label{eq:dosm-log1}
n([E-\epsilon, E+\epsilon]) \leq \frac{C}{\log \left(\frac{1}{\epsilon} \right)}, ~~~\forall 0 < \epsilon \leq \frac{1}{2},
\eeq
which is always satisfied by Craig and Simon \cite{craig-simon1}. Since the underlying probability measure $\nu$ will be fixed in this discussion, we simplify notation and denote the DOSm simply by $n$.

We consider replacing the left side of \eqref{eq:poisson-le2} by a statement about
\beq\label{eq:le-bv1}
\lim_{\epsilon \rightarrow 0^+ } \left| \frac{L(E+i\epsilon) - L(E)}{\Phi (\epsilon)} \right| ,
\eeq
for some strictly increasing function $0 \leq \Phi : [0, \epsilon_0] \rightarrow  \R$,
$\Phi \in C^1 ((0, \epsilon_0))$,  and with $\Phi' > 0$ on $(0, \epsilon_0)$.

We show that by following the same analysis as in the proof of Proposition \ref{prop:le-quant1}, but using \eqref{eq:dosm-log1} in place of $\beta$-H\"older continuity \eqref{eq:beta-cont-m1}, we obtain an upper bound that is infinite. We conclude that if a quantitative estimate on the boundary-value of the Lyapunov can be achieved with weaker conditions on the DOSm, another method of proof is needed.

We begin with the following lemma whose proof is an obvious modification of the proof in section \ref{subsec:le1}.
\begin{lemma}\label{lemma:le-general1}
For a function $\Phi$ as described above and all $E \in \R$, we have
\beq\label{eq:le-bv2}
\limsup_{\epsilon \rightarrow 0^+} \left\{ \frac{L(E+i\epsilon) - L(E)}{\Phi (\epsilon)} \right\} =
\limsup_{\epsilon \rightarrow 0^+} \left\{ \frac{P_n(E+i\epsilon)}{\Phi'(\epsilon)}  \right\}.
\eeq
\end{lemma}

\begin{proof}
We set $\eta := \Phi(\epsilon)$ so that $\epsilon = \Phi^{-1}(\eta)$. We consider the function $f$ defined by
\beq\label{eq:le-bv3}
f(\eta) := L(E + i \Phi^{-1}(\eta))
\eeq
on $[0, \eta_0)$, where $\eta_0 = \Phi(\epsilon_0)$. By the Mean Value Theorem, given $0 < \eta < \eta_0$, there exists $\tilde{\eta} = \tilde{\eta}(\eta) \in (0,\eta)$ so that
\bea\label{eq:le-bv4}
 \frac{L(E+i\Phi^{-1}(\eta)) - L(E)}{\eta} & = & \frac{f(\eta) - f(0)}{\eta}  =  f'(\tilde{\eta}) \nonumber \\
  & =& P_n( E + i \Phi^{-1}(\tilde{\eta})) ( \Phi^{-1})' (\tilde{\eta}) \nonumber \\
  &=& P_n( E + i \Phi^{-1}(\tilde{\eta})) \left( \frac{1}{ \Phi^\prime (\Phi^{-1}(\tilde{\eta}))} \right).
   \eea
Hence, changing coordinates $\eta = \Phi(\epsilon)$, we conclude that for each $0 < \epsilon < \epsilon_0$, there exists $\tilde{\epsilon} \in (0, \epsilon)$ so that
\beq\label{eq:le-bv5}
 \frac{L(E+i\epsilon) - L(E)}{\Phi (\epsilon)} =  P_n( E + i \tilde{\epsilon}) \left( \frac{1}{ \Phi^\prime (\tilde{\epsilon})} \right).
   \eeq
 Upon taking $\epsilon \rightarrow 0^+$ in \eqref{eq:le-bv5}, we establish the claim
  \eqref{eq:le-bv2}.
\end{proof}

We follow the proofs of section \ref{subsec:le1} replacing the assumption that the DOSm is $\beta$-H\"older continuous with the log-H\"older continuous property (\ref{eq:dosm-log1}) and $\epsilon^\beta$ by an appropriate function $\Phi(\epsilon)$ satisfying the conditions above. 
We provide evidence that, {\em{only relying on the {\underline{upper}} bound in (\ref{eq:dosm-log1})}}, there does not exist a nontrivial function $\Phi$ so that, using the methods of section \ref{subsec:le1}, we have
\beq\label{eq:le-bv7}
\limsup_{\epsilon \rightarrow 0^+} \left\{ \frac{P_n( E + i {\epsilon})}{ \Phi^\prime ({\epsilon})} \right\} < \infty ~.
\eeq
Here, the word ``nontrivial'' excludes the choice $\Phi(\epsilon) = | L(E+i \epsilon) - L(E)|$ for which one trivially has
$$
| L(E+i \epsilon) - L(E)| \leq  d_E \Phi (\epsilon) ~.
$$

\begin{lemma}\label{lemma:le-bound1}
Assuming that the DOSm $n$ is log-H\"older continuous, for all $0 < \epsilon < \frac{1}{2}$, we have
\beq\label{eq:le-bv8}
\epsilon ~P_n (E + i \epsilon) \leq \epsilon + C \int_{\epsilon}^1 \frac{d \alpha}{\log \left( \frac{1}{\epsilon} \left[ \frac{\alpha}{1- \alpha} \right]^{1/2} \right)} ~.
\eeq
\end{lemma}

\begin{proof}
Let $f_\epsilon (x) = \epsilon^2 ( \epsilon^2 + x^2)^{-1}$, for $\epsilon > 0$. Observe that for every $0 < \epsilon$, $f_\epsilon$ is even and decreasing with $0 \leq f_\epsilon \leq 1$. Thus, setting
$$
r_\epsilon(\alpha) := \epsilon \left[ \frac{1- \alpha}{\alpha}   \right]^{\frac{1}{2}} \in [0, \infty].
$$
we compute by the layer-cake representation,
\bea\label{eq:le-bv9}
 \epsilon P_n ( E+i\epsilon) 
  & = & \int_\R ~\left\{ \int_0^\infty ~\chi_{\{ f_\epsilon > \alpha \}} (x) ~d \alpha \right\} ~dn_E(x) \nonumber \\
  & = & \int_0^1 n_E ((-r_\epsilon(\alpha), r_\epsilon (\alpha)) ~d \alpha ,
\eea
where we abbreviated $d n_E(x) := dn(x + E)$.

To use the log-H\"older continuity in \eqref{eq:dosm-log1}, we divide the integral on the right of \eqref{eq:le-bv9} into $[0, \epsilon] \cup [\epsilon, 1]$ and obtain
\begin{equation} \label{eq:le-bv12}
\epsilon P_n ( E+i\epsilon)  \leq  \epsilon + \int_\epsilon^1 n_E ((-r_\epsilon(\alpha), r_\epsilon (\alpha))
 ~d \alpha \leq \epsilon +  C \int_\epsilon^1 \frac{d \alpha}{\log \left( \frac{1}{\epsilon} \left[\frac{\alpha}{1 - \alpha} \right]^{\frac{1}{2}}  \right)} ~\mbox{.}
\end{equation}
Notice that $r_\epsilon(\epsilon) \leq \frac{1}{2}$, whence (\ref{eq:dosm-log1}) applies on the interval $[\epsilon, 1]$, which results in the estimate of the integral in (\ref{eq:le-bv12}).
\end{proof}

Lemma \ref{lemma:le-bound1} and the log-H\"older continuity of the DOSm imply
\beq\label{eq:le-bv13}
 \limsup_{\epsilon \rightarrow 0^+} \left\{ \frac{ P_n( E + i {\epsilon})}{ \Phi^\prime ({\epsilon})} \right\} \leq \limsup_{\epsilon \rightarrow 0^+} \left\{ \frac{1}{\Phi^\prime ({\epsilon})} +  \frac{C}{\epsilon \Phi'(\epsilon)} ~ \int_\epsilon^1 \frac{d \alpha}{\log \left( \frac{1}{\epsilon} \left[\frac{\alpha}{1 - \alpha} \right]^{\frac{1}{2}}  \right)} \right\}  .
   \eeq

\begin{lemma}\label{lemma:le-bound2}
There exists $\epsilon_0 > 0$ such that for all $0 < \epsilon < \epsilon_0$, we have
\beq\label{eq:le-bv14}
\int_\epsilon^1 \frac{d \alpha}{\log \left( \frac{1}{\epsilon} \left[\frac{\alpha}{1 - \alpha} \right]^{\frac{1}{2}}  \right)} \geq \frac{1}{100} \frac{1}{\log \left( \frac{1}{\epsilon} \right)} .
\eeq
\end{lemma}

\begin{proof}
For $0 < \epsilon \leq 1/2$, let $h_\epsilon (\alpha) := \log \left( \frac{1}{\epsilon} \left[\frac{\alpha}{1 - \alpha} \right]^{\frac{1}{2}}  \right)$. The function $h_\epsilon \geq 0$ is continuous on $[\epsilon, 1)$, strictly decreasing, with
$$
h_\epsilon (\epsilon) = \log \left( \frac{1}{\sqrt { \epsilon }} \frac{1}{\sqrt{1 - \epsilon}}   \right) ~\mbox{, }
h_\epsilon (1) = 0.
$$
We define a sequence $x_j^{(\epsilon)} \in [\epsilon, 1]$ by $x_j^{(\epsilon)} := j \epsilon$, for $1 \leq j \leq \lfloor \frac{1}{\epsilon} \rfloor =: N_\epsilon + 1$, so that $x_{N_\epsilon+2}^{(\epsilon)} = 1$. Then, since $h_\epsilon \geq 0$ is strictly decreasing on $[\epsilon, 1]$, we have
\bea\label{eq:le-bv151}
\int_\epsilon^1 \frac{d \alpha}{\log \left( \frac{1}{\epsilon} \left[\frac{\alpha}{1 - \alpha} \right]^{\frac{1}{2}}  \right)} & \geq & \sum_{j=1}^{N_\epsilon - 1} h_\epsilon (x_{j+1}^{(\epsilon)})(x_{j+1}^{(\epsilon)} - x_{j}^{(\epsilon)}) \nonumber \\
 & & + h_\epsilon (x_{N_\epsilon+1}^{(\epsilon)}) \epsilon + h_\epsilon (x_{N_\epsilon + 2}^{(\epsilon)})(1-x_{N_\epsilon+1}^{(\epsilon)})
   \nonumber \\
& \geq & \epsilon \cdot \sum_{j=1}^{N_\epsilon - 1} h_\epsilon (x_{j+1}^{(\epsilon)}) \nonumber \\
 & \geq & \epsilon ( N_\epsilon - 1) h_\epsilon (1 - \epsilon) \geq
  \frac{2\epsilon}{3} \left(  \left\lfloor \frac{1}{\epsilon} \right\rfloor -1 \right) \cdot \frac{1}{ \log \left( \frac{1}{\epsilon}  \right) } ~,
 \eea
where we used the facts that $h_\epsilon (x_{N_\epsilon + 2}^{(\epsilon)}) = 0$, $h_\epsilon (x_{N_\epsilon + 1}^{(\epsilon)}) \geq 0$, and that $|x_{j+1}^{(\epsilon)} - 1 | \geq \epsilon$, for all $1 \leq j \leq N_\epsilon - 1$. In conclusion, we obtain the claim since $\lim_{\epsilon \rightarrow 0^+} \epsilon \left\lfloor \frac{1}{\epsilon} \right\rfloor  = 1$.
\end{proof}

\begin{lemma}\label{lemma:le-bound3}
For all functions $\Phi$ as defined above, we have 
\beq\label{eq:le-bv15}
\limsup_{\epsilon \rightarrow 0^+} \frac{1}{\epsilon \Phi^\prime (\epsilon)} \int_\epsilon^1 \frac{d \alpha}{\log \left( \frac{1}{\epsilon} \left[\frac{\alpha}{1 - \alpha} \right]^{\frac{1}{2}}  \right)} = + \infty.
 \eeq
\end{lemma}

\begin{proof}
Let us suppose to the contrary that there exists a function $\Phi$ so that the left side of \eqref{eq:le-bv15} is finite. Letting $\epsilon_0>0$ be as in Lemma \ref{lemma:le-bound2}, there thus exists $C_1 > 0$ and $0 < \tilde{\epsilon}_0 < \epsilon_0$ so that
for all $0 < \epsilon < \tilde{\epsilon}_0$, we have
\beq\label{eq:le-bv16}
\frac{C_1}{\epsilon} \int_\epsilon^1 \frac{d \alpha}{\log \left( \frac{1}{\epsilon} \left[\frac{\alpha}{1 - \alpha} \right]^{\frac{1}{2}}  \right)} \leq \Phi^\prime (\epsilon).
\eeq
By Lemma \ref{lemma:le-bound2}, we have for all $0 < \epsilon < \tilde{\epsilon}_0$,
\beq\label{eq:le-bv17}
\Phi^\prime (\epsilon) \geq
\frac{C_1}{100 \epsilon} \frac{1}{\log \left( \frac{1}{\epsilon} \right)}
 \eeq
Since $\Phi \in C^1 ((0, \epsilon_0))$, and as $(0, \tilde{\epsilon}_0] \subset (0, \epsilon_0)$, we find that for all $0 < \epsilon < \tilde{\epsilon}_0$,
\bea\label{eq:le-bv18}
\Phi (\tilde{\epsilon}_0) & \geq & \Phi (\epsilon) + \frac{C_1}{100} \int_{\epsilon}^{\tilde{\epsilon}_0} ~ \frac{dt}{ t \log \left( \frac{1}{t} \right)}  \nonumber \\
&=& \Phi (\epsilon) + \frac{C_1}{100} \left\{  \log \left(  \log \left( \frac{1}{\epsilon} \right) \right)
 -    \log \left(  \log \left( \frac{1}{\tilde{\epsilon}_0} \right) \right) \right\} ~,
 \eea
Now, by hypothesis $\lim_{\epsilon \rightarrow 0^+} \Phi (\epsilon) = \Phi (0) < \infty$ but
$$
\lim_{\epsilon \rightarrow 0^+} \log \left(  \log \left( \frac{1}{\epsilon} \right) \right) = +\infty ~\mbox{,}
$$
so we conclude that $\Phi (\tilde{\epsilon}_0) = + \infty$, a contradiction.\end{proof}

We note that this result is not a contradiction to the trivial statement that
$$
\lim_{\epsilon \rightarrow 0^+} \frac{P_n(E+ i \epsilon)}{\Phi^\prime (\epsilon)} = 1 < + \infty ,
$$
if $\Phi(\epsilon) = |L(E + i \epsilon) - L(E)|$, since \eqref{eq:le-bv13} is an upper bound. But, this upper bound is the best one can obtain {\em{using only the upper bound in (\ref{eq:dosm-log1})}}.

In this context, we also note that the decomposition of the integral in \eqref{eq:le-bv12} could be performed ``more carefully,'' replacing $\epsilon$ in the decomposition by a general decreasing function $\epsilon \leq \Theta(\epsilon) \leq 1$ with $\lim_{\epsilon \to 0^+} \Theta(\epsilon) = 0$, 
thereby resulting in the upper bound
\begin{equation} \label{eq:keyestimfinalequ}
 \limsup_{\epsilon \rightarrow 0^+} \left\{ \frac{ P_n( E + i {\epsilon})}{ \Phi^\prime ({\epsilon})} \right\} \leq \limsup_{\epsilon \rightarrow 0^+} \left\{ \frac{\Theta(\epsilon)}{\epsilon \cdot \Phi^\prime ({\epsilon})} +  \frac{C}{\epsilon \Phi'(\epsilon)} ~ \int_{\Theta(\epsilon)}^1 \frac{d \alpha}{\log \left( \frac{1}{\epsilon} \left[\frac{\alpha}{1 - \alpha} \right]^{\frac{1}{2}}  \right)} \right\}  ~\mbox{.}
\end{equation}
Similar to Lemma \ref{lemma:le-bound2} one then shows that for all sufficiently small $\epsilon$, one has the lower bound
\begin{equation}
\int_{\Theta(\epsilon)}^1 \frac{d \alpha}{\log \left( \frac{1}{\epsilon} \left[\frac{\alpha}{1 - \alpha} \right]^{\frac{1}{2}}  \right)} \geq \frac{1}{100} \frac{1}{\log \left( \frac{\sqrt{\Theta(\epsilon)}}{\epsilon} \right)} ~\mbox{,}
\end{equation}
Because $\lim_{\epsilon \to 0^+} \Theta(\epsilon) = 0$, one necessarily has $\frac{\sqrt{\Theta(\epsilon)}}{\epsilon} \leq \frac{1}{\epsilon^{3/2}}$ eventually in $\epsilon$ as $\epsilon \to 0^+$, thus arguments very much along the lines of the proof of Lemma \ref{lemma:le-bound3} imply that the $\limsup$ of the second term in (\ref{eq:keyestimfinalequ}) is infinite.



\begin{thebibliography} {[10]}
\frenchspacing \baselineskip=12 pt plus 1pt minus 1pt

\bibitem{aizenmanWarzel_book} M.\ Aizenman, S.\ Warzel, {\it Random Operators: Disorder Effects on Quantum Spectra and Dynamics}, Graduate Studies in Mathematics \textbf{168}, American Mathematical Society, Providence, RI, 2015.

\bibitem{Abou-ChacraThoulessAnderson_JPhysC_1973}   R.\ Abou-Chacra, D.\ J.\ Thouless, and P.\ W.\ Anderson. {\it A self-consistent theory of localization}, Journal of Physics C: Solid State Physics \textbf{6.10} (1973), pp. 1734 -- 1752.

\bibitem{AcostaKlein} V.\ Acosta, A.\ Klein, \emph{Analyticity of the density of states in the Anderson model on the Bethe lattice}, J.\ Statist.\ Phys. {\bf 69} (1992), no.\ 1--2, 277 -- 305.

\bibitem{AvilaEskinViana} A.\ Avila, A.\ Eskin, M.\ Viana, {\em{Continuity of Lyapunov exponents of random matrix products}}, in preparation.

\bibitem{beals1} R.\ Beals, {\em Analysis. An introduction}, Cambridge University Press, Cambridge, 2004. x+261 pp. ISBN: 0-521-84072-4; 0-521-60047-2 26-01 (00-01).

\bibitem{billingsley} P.\ Billingsley, {\em{ Convergence of Probability measures}}, 2nd edition, Wiley \& Sons, Inc., New York, 1999.

\bibitem{BockerViana_ETDS_2017} C.\ Bocker, M.\ Viana, {\em{Continuity of Lyapunov exponents for random two-dimensional matrices}}, Ergodic Theory Dynam.\ Systems \textbf{37} (2017), no.\ 5, 1413 -- 1442.

\bibitem{Bourgain_Jdanalyse_2012} J.\ Bourgain, {\em{On the Furstenberg measure and density of states for the Anderson-Bernoulli model at small disorder}}, J.\ d'Analyse Math. \textbf{117(1}), 273 -- 295 (2012).

\bibitem{bourgain-klein1} J.\ Bourgain, A.\ Klein, \emph{Bounds on the density of states for Schr\"odinger operators}, Invent.\ Math. {\bf 194}, no.\ 1, 41 -- 72 (2013).

\bibitem{bourgain_schlag_CMP_2000} J.\ Bourgain and W.\ Schlag, {\em{Anderson localization for Schr\"odinger operators on $\mathbb{Z}$ with strongly mixing potentials}}, Comm.\ Math.\ Phys. \textbf{215} (2000), no. 1, 143 -- 175.

\bibitem{bovier-campanino-klein-perez1} A.\ Bovier, M.\ Campanino, A.\ Klein, J.\ F.\ Perez, \emph{Smoothness of the density of states in the Anderson model at high disorder}, Comm.\ Math.\ Phys. {\bf 114} (1988), no.\ 3, 439 -- 461.

\bibitem{BovierKlein88} A.\ Bovier, A.\ Klein, \emph{Weak disorder expansion of the invariant measure for the one-dimensional Anderson model}, J.\ Statist.\ Phys. {\bf 51} (1988), no.\ 3--4, 501 -- 517.

\bibitem{Bucaj_KunzSoulliardprep2016} V.\ Bucaj, {\it On the Kunz-Souillard approach to localization for the discrete one dimensional generalized Anderson model}, (2016). Preprint available as arXiv:1608.01379.

\bibitem{campanino-klein86} M.\ Campanino, A.\ Klein, \emph{A supersymmetric transfer matrix and differentiability of the density of states in the one-dimensional Anderson model}, Comm.\ Math.\ Phys. {\bf 104} (1986), no.\ 2, 227 -- 241.

\bibitem{campanino-klein90} M.\ Campanino, A.\ Klein, \emph{Anomalies in the one-dimensional Anderson model at weak disorder}, Comm.\ Math.\ Phys. {\bf 130} (1990), no.\ 3, 441 -- 456.

\bibitem{carmonalacroix} R.\ Carmona and J.\ Lacroix, {\it Spectral theory of random Schr�dinger operators}, Birkh\"auser Boston, Inc., Boston, MA, 1990.

\bibitem{craig-simon1} W.\ Craig, B.\ Simon, \emph{Log H\"older continuity of the integrated density of states for stochastic Jacobi matrices}, Comm.\ Math.\ Phys. {\bf 90} (1983), no.\ 2, 207 -- 218.

\bibitem{chk1} J.-M.\ Combes, P.\ D.\ Hislop, F.\ Klopp, \emph{An optimal Wegner estimate and its application to the global continuity of the integrated density of states for random Schr\"dinger operators}, Duke Math.\ J. {\bf 140} (2007), no.\ 3, 469 -- 498.

\bibitem{damanik1} D.\ Damanik, {\it Lyapunov exponents and spectral analysis of ergodic Schr\"odinger operators: a survey of Kotani theory and its applications}, Spectral theory and mathematical physics: a Festschrift in honor of Barry Simon's 60th birthday, 539 -- 563, Proc.\ Sympos.\ Pure Math., {\bf 76}, Part 2, Amer.\ Math.\ Soc., Providence, RI, 2007.

\bibitem{DamanikSimsStolz_JFunAnal_2004} D.\ Damanik, R.\ Sims, G.\ Stolz, {\it Localization for discrete one-dimensional random word models}, J.\ Fun.\ Anal. \textbf{208} (2004) 423 -- 445.

\bibitem{davies1} E.\ B.\ Davies, \emph{Spectral theory and differential operators}, Cambridge studies in advanced mathematics vol.\ {\bf 42}, Cambridge University Press, Cambridge, 1995.

\bibitem{deBievre_Germinet_2000} S.\ De Bi\`evre, F.\ Germinet, \emph{Dynamical localization for the random dimer Schr �dinger operator}, J. Statist. Phys. 98 (2000), no. 5-6, 1135?1148.

\bibitem{DelyonSouillard_1984} F.\ Delyon, B.\ Souillard, \textit{Remark on the Continuity of the Density of States of Ergodic Finite Difference Operators}, Commun. Math. Phys 94, 289 -- 291 (1984).

\bibitem{DuarteKlein_monograph} P.\ Duarte and S.\ Klein, \textit{Lyapunov exponents of linear cocycles}, Atlantis Press (2016), Amsterdam.

\bibitem{Dudley_1966} R. M. Dudley, \textit{Convergence of Baire measures}, Studia Math. 27 (1966), 251 -- 268.

\bibitem{Dudley_1976_book} R.\ M.\ Dudley, {\em{Probabilities and metrics}}, Lecture Notes Series, No. 45. Matematisk Institut, Aarhus Universitet, Aarhus, 1976. ii+126 pp.

\bibitem{Furstenberg_Kiefer_IsrealJMath_1981}  H.\ Furstenberg and Yu.\ Kifer, \emph{Random matrix products and measures in projective spaces}, Israel J. Math 10, 12 -- 32 (1983).

\bibitem{HartVirag_CMP_2017} E.\ Hart and B.\ Vir\'ag, \emph{H\"older continuity of the integrated density of states in the one-dimensional Anderson model}, Comm.\ Math.\ Phys. \textbf{355} (2017), no.\ 3, 839--863.

\bibitem{hks1} P.\ D.\ Hislop, F.\ Klopp, J.\ Schenker, {\it Continuity with respect to disorder of the integrated density of states},  Illinois J.\ Math. {\bf 49}, no.\ 3, 893 -- 904 (2005).

\bibitem{MarxHislop_contDosm_followup} P.\ D.\ Hislop, C.\ A.\ Marx, {\em{Dependence of the density of states on the probability distribution - part II: Schr\"odinger operators on $\R^d$ and non-compactly supported probability measures}}, in preparation (2018).

\bibitem{hstoiciu} P.\ D.\ Hislop and M.\ Stoiciu, {\it{Eigenvalue statistics for the one-dimensional lattice Anderson model in the weak disorder limit}}, in preparation (2018).

\bibitem{SchulzBaldes_Jitomirskaya_Stolz_CMP_2003} S.\ Jitomirskaya, H.\ Schulz-Baldes, G.\ Stolz, \emph{Delocalization in random polymer models}, Comm.\ Math.\ Phys. \textbf{233}(2003), no.\ 1, 27 -- 48.

\bibitem{KappusWegner_1981} M.\ Kappus, F.\ Wegner, \emph{Anomaly in the band centre of the onedimensional Anderson model}, Z. Phys. \textbf{B45} (1981), 15 -- 21.

\bibitem{KirschMetzger_BarryFestschrift_2007} W.\ Kirsch, B.\ Metzger, {\em{The integrated density of states for random Schr\"odinger operators}}, Spectral theory and mathematical physics: a Festschrift in honor of Barry Simon's 60th birthday, 649--696, Proc.\ Sympos.\ Pure Math., \textbf{76}, Part 2, Amer.\ Math.\ Soc., Providence, RI, 2007.

\bibitem{KleinSpeis_JFunctAnal_1990} A.\ Klein and A.\ Speis, {\em{Regularity of the invariant measure and the density of states in the one--dimensional Anderson model}}, J.\ Funct.\ Anal. \textbf{88} (1990), 211--227.




\bibitem{KleinSadel} A.\ Klein, C.\ Sadel, \emph{Absolutely continuous spectrum for random Schr\"odinger operators on the Bethe strip},  Math.\ Nachr.\ {\bf 285} (2012), no.\ 1, 5 -- 26.

\bibitem{simon_kotani_CMP1988} S.\ Kotani, B.\ Simon, {\em{Stochastic Schr\"odinger operators and Jacobi matrices on the strip}}, Comm.\ Math.\ Phys. \textbf{119}(1988), no.\ 3, 403--429.

\bibitem{marx1} C.\ Marx, {\em Continuity of spectral averaging}, Proc.\ Amer.\ Math.\ Soc. {\bf 139}, no.\ 1, 283 -- 291 (2011).

\bibitem{pastur1} L.\ Pastur, \emph{Spectral properties of disordered systems in one-body approximation}, Comm.\ Math.\ Phys. {\bf 75} (1980), no.\ 2, 179 -- 196.

\bibitem{PasturFigotin_book} L.\ Pastur and A.\ Figotin, \emph{Spectra of random and almost-periodic operators}, Grundlehren der Mathematischen Wissenschaften [Fundamental Principles of Mathematical Sciences], \textbf{297}. Springer-Verlag, Berlin, 1992. viii+587 pp. ISBN: 3-540-50622-5

\bibitem{AleksandrovPeller_review_2016} A. B. Aleksandrov, V. V. Peller, \emph{Operator Lipschitz functions}. (Russian. Russian summary) Uspekhi Mat. Nauk \textbf{71} (2016), no. 4(430), 3--106; translation in Russian Math. Surveys 71 (2016), no. 4, 605--702.

\bibitem{rojas-molina-veselic1} C.\ Rojas-Molina, I.\ Veseli\'c, \emph{Scale-free unique continuation estimates and applications to random Schr\"odinger operators}, Comm.\ Math.\ Phys. {\bf 320} (2013), no.\ 1, 245 -- 274.

\bibitem{schenker1} J.\ Schenker,  {\it H\"older equicontinuity of the integrated density of states at weak disorder}, Lett.\ Math.\ Phys. {\bf 70}, no.\ 3, 195 -- 209 (2004).

\bibitem{Schulz-Baldes_GAFA_2004} H. Schulz-Baldes, \emph{Perturbation theory for Lyapunov exponents of an Anderson model on a strip}, Geom. Funct. Anal. \textbf{14} (2004), no. 5, 1089 -- 1117.

\bibitem{Schulz-Baldes_ OperThyAdvAppl_2007} H.\ Schulz-Baldes, \emph{Lyapunov exponents at anomalies of SL(2,$\mathbb{R}$)-actions}, Operator theory, analysis and mathematical physics, 159--172, Oper.\ Theory Adv.\ Appl., \textbf{174}, Birkh\"auser, Basel, 2007.

\bibitem{sikkema} P.\ C.\ Sikkema, {\em Der Wert einiger Konstanten in der Theorie der Approximation mit Bernstein-Polynomen}, Numer.\ Math. {\bf 3}, 107 -- 116 (1961).

\bibitem{simon1} B.\ Simon, {\em Spectral averaging and the Krein spectral shift function}, Proc.\ Amer.\ Math.\ Soc. {\bf 126}, no.\ 5, 1409--1413 (1998).

\bibitem{simontaylor} B. Simon and M. Taylor, {\em{Harmonic Analysis on $SL(2, \mathbb{R})$ and smoothness of the density of states in the one-dimensional Anderson model}}, Commun.\ Math.\ Phys. {\bf{101}}, 1 -- 19 (1985).

\bibitem{Speis_CMP_1992} A.\ Speis, {\em{Weak disorder expansions for the Anderson model on a one-dimensional strip at the center of the band}}, Comm.\ Math.\ Phys. \textbf{149} (1992), no.\ 3, 549 -- 571.

\bibitem{Speis_JStatPhys_1991} A.\ Speis, {\em{Instability of the anomalies in the one-dimensional Anderson model at weak disorder}}, J. Statist. Phys. \textbf{63} (1991), no.\ 3 -- 4, 541 -- 565.

\bibitem{Thouless_PRL_1977} D. J. Thouless, \emph{Maximum metallic resistance in thin wires}, Phys.\ Rev.\ Lett. \textbf{39} (1977), 1167 -- 1169.

\bibitem{Veselic_review_2004} I. Veseli\'c, {\em{Integrated density of states and Wegner estimates for random Schr\"odinger operators}}, Spectral theory of Schr\"odinger operators, 97 -- 183, Contemp. Math., \textbf{340}, Aportaciones Mat., Amer. Math. Soc., Providence, RI, 2004.

\bibitem{warzel_ICMP_2013}  S. Warzel, {\it Surprises in the phase diagram of the Anderson model on the Bethe lattice}, XVIIth International Congress on Mathematical Physics (2013), pp. 239 -- 253.

\end{thebibliography}
\end{document}